\documentclass[acmsmall,screen,nonacm, appendix,]{acmart}
\pdfoutput=1
\usepackage[T1]{fontenc}
\usepackage[utf8]{inputenc}

\newif\iflong\longfalse

\newif\ifprecat\precatfalse

\usepackage{style}
\newcommand{\khchanged}[1]{{#1}}

\title{A Physically-Realisable Quantum Programming Language with Quantum-Controlled Quantum Channels}

\author{Kengo Hirata}
\affiliation{%
  \institution{Chiba University}
  \city{Chiba}
  \country{Japan}
}
\affiliation{%
  \institution{University of Edinburgh}
  \city{Edinburgh}
  \country{UK}
}
\email{k.hirata@sms.ed.ac.uk}

\author{Takeshi Tsukada}
\affiliation{
  \institution{Chiba University}
  \city{Chiba}
  \country{Japan}
}
\email{t.tsukada@acm.org}
 
\begin{document}
\begin{abstract}
In contrast to a classical bit, which can only take the value $0$ or $1$, its quantum counterpart---a qubit---can exist in a superposition of $0$ and $1$.
This is a superposition of data values, and it is then tempting to superpose programs as well.
Indeed, a particular superposition of programs, known as the quantum SWITCH, has attracted considerable attention.
A natural way to superpose programs is conditional branching on quantum data, and several studies have pursued this idea.
However, there is a well-known semantic difficulty with this approach, and the fundamental question of whether the languages
proposed in prior work can be physically realised remains open.

This paper proposes a physically realisable quantum programming language with quantum branching over quantum channels.
We first analyse the semantic difficulty from an implementation perspective and show that the difficulty is caused by arbitrariness in the implementation.
We then design a language that is free from this arbitrariness.
As expected from its implementation-oriented design, our language is physically realisable.
We argue for this from two directions:
first, we give a method for compiling programs into quantum circuits;
and second, we show that our programs stay within the physically admissible range, namely, that they neither yield probabilities greater than one nor violate causal inequalities.
The absence of such violations is established through the inclusion of our language in QC-QC, a notion of realisability studied in physics.

\end{abstract}
 \maketitle

\section{Introduction}\label{sec:intro}

\emph{Quantum computation} is a paradigm that makes use of the principles of quantum mechanics to perform computations and is expected to solve problems that are intractable for classical computers.
While a classical bit is a binary data type that takes the value \( 0 \) or \( 1 \),
a \emph{quantum bit} or \emph{qubit} is represented by a unit vector in the space spanned by \(\ket{0}\) and \(\ket{1}\).
Thus, its general form is \( \alpha \ket{0} + \beta \ket{1} \), \( |\alpha|^2 + |\beta|^2 = 1 \), a superposition of \( \ket{0} \) and \( \ket{1} \).
Applying an operation \( U \) (more precisely, a unitary operator) to this state results in \( \alpha (U\!\ket{0}) + \beta (U\!\ket{1}) \), and the post-operation state, in a sense, encodes the effect of applying \( U \) to both \( \ket{0} \) and \( \ket{1} \).
This is in contrast to classical function calls, where a single invocation yields information about either \( f(0) \) or \( f(1) \).

Superpositions of \( \ket{0} \) and \( \ket{1} \) are superpositions of data, and a natural question is whether superpositions of programs are also possible.
In recent years, affirmative answers to this question have been obtained and actively investigated, mainly by the physics community.
A typical example is the \emph{quantum SWITCH}\footnote{
  Note that the quantum SWITCH is unrelated to the switch statement in C.
  The term quantum SWITCH refers to a specific operator defined by \citet{Chiribella2013}.
}~\cite{Chiribella2013}.
It takes a control qubit \( x \), two quantum operations \( F \) and \( G \), and an input \( y \).
When \( x = \ket{0} \), the SWITCH behaves as \( G\,(F\,y) \); when \( x = \ket{1} \), it behaves as \( F\,(G\,y) \); and when \( x = \alpha_0\ket{0} + \alpha_1\ket{1} \) (with \( \alpha_0 \neq 0 \neq \alpha_1 \)), it behaves as a certain ``superposition'' of \( G\,(F\,y) \) and \( F\,(G\,y) \).
Roughly speaking, \( x \) controls the order in which \( F \) and \( G \) are applied.
Notably, when the control qubit \( x \) is in a superposition of \( \ket{0} \) and \( \ket{1} \), the quantum SWITCH exhibits peculiar behaviour: there is no well-defined answer to whether $F$ is performed before $G$ or vice versa.
This phenomenon is often described as the absence of a definite causal order, or more technically as \emph{indefinite causal order}, and has been actively studied as an important topic in quantum computation since the work of \citet{Oreshkov2012} and \citet{Chiribella2013}.

The quantum SWITCH is known to bring computational advantages~\cite{Araujo2014,Colnaghi2012,switch-advantage2,Kristjansson2024,Kristjansson2020,Ebler2018,Chiribella2021}.
For example, even if both \( F \) and \( G \) are completely depolarising channels, which are unable to transmit any quantum information, their ``superposition'' via the quantum SWITCH
has a non-zero capacity~\cite{Ebler2018}.
Computational advantages are not limited to the quantum SWITCH.
For example, a superposition of $F(V(G(y)))$ and $G(V(F(y)))$ has been shown to reduce the number of gates required to solve a generalisation of Deutsch's problem~\cite{Liu2024}.

A programming language capable of expressing operations such as the quantum SWITCH is important not only for describing and analysing procedures involving the quantum SWITCH, but also for exploring as-yet-unknown but potentially interesting operations of a similar kind.
This paper aims to provide such a quantum programming language.

A natural idea would be to introduce a \emph{quantum conditional statement},
written \( \qif \), as a quantum analogue of the classical \( \cif \) statement.
We expect the following behaviour for \( \qifx{x}{P}{Q} \).
When \( x = \ket{1} \), it behaves like \( P \); when \( x = \ket{0} \), it behaves like \( Q \); and when \( x \) is in a superposition of \( \ket{0} \) and \( \ket{1} \), it behaves like a ``superposition'' of \( P \) and \( Q \).
If only such a construct were available, one might express the quantum SWITCH as
\begin{equation}
    \mathit{SWITCH}(x,y,F,G)
    \qquad=\qquad
    \qifx{x}{F\,(G\,y)}{G\,(F\,y)}.
    \label{eq:into:switch}
\end{equation}
However, this approach is known to face technical difficulties~\cite{Badescu2015} unless \(F\) and \(G\) are restricted to, e.g., unitary operations as in \cite{QML,Qsharp,Silq,SabValiViz18-FOSSACS-q-recusion}.
In a setting without such restrictions,\footnote{In more technical terms, this is the setting in which \(F\) and \(G\) may involve quantum measurements and classical conditional branching.}
the design of a quantum programming language with a quantum conditional statement has been identified as an important open problem in a survey by Valiron~\cite{Valiron2022}.
In addition to the work discussed in the survey~\cite{Badescu2015,YingYF12-first-Ying,YingYF14-alternation,Ying16-the-book}, further progress has been made since then~\cite{Barsse2026,DaveDZ25,barsse2026-arxiv}, but the problem remains open, as we discuss in detail below (see also Table~\ref{tab:comparison}).
In particular, the crucial question of whether programs with quantum conditional statements can be physically realised has yet to be addressed.
We propose a physically realisable quantum programming language with a quantum conditional statement.

\begin{table}[b]
  \centering
  \caption{Comparison of quantum programming languages with quantum branching over non-unitaries.
           In rows for prior work, the superscript $\star$ marks results established in this paper.}
  \label{tab:comparison}

  \small
  \setlength{\tabcolsep}{4pt}
  \renewcommand{\arraystretch}{1.25}

  \resizebox{\linewidth}{!}{%
    \begin{tabular}{@{}lccccccc@{}}
      \toprule
      \multirow{2}{*}{}
      & \multirow{2}{*}{\shortstack{\cref{eq:into:switch}\\expresses\\SWITCH}}
      & \multicolumn{3}{c}{Physical plausibility}
      & \multicolumn{3}{c}{Realisation}
      \\
      \cmidrule(lr){3-5}
      \cmidrule(l){6-8}
      &
      & \shortstack{Total probability\\$\leq 1$}
      & \shortstack{Well-defined\\on $\mathbf{CPM}$}
      & \shortstack{Satisfies causal\\inequalities}
      & \shortstack{Single\\direct use}
      & \shortstack{Dilation\\+ control + copy}
      & QC-QC
      \\
      \midrule

      \citet{YingYF12-first-Ying,YingYF14-alternation}, \citet{Ying16-the-book}
      & \xmark
      & \cmark
      & \xmark
      & ---
      & \xmark
      & \pmark$^{\star}$
      & \xmark
      \\

      \citet{Badescu2015}
      & \xmark
      & \cmark
      & \xmark
      & ---
      & \xmark
      & \cmark$^{\star}$
      & \xmark
      \\

      \citet{Barsse2026}
      & \xmark
      & \cmark
      & \xmark
      & ---
      & \xmark
      & \cmark$^{\star}$
      & \xmark
      \\

      \citet{DaveDZ25}
      & \cmark$^{\dagger}$
      & \xmark
      & \cmark
      & \xmark
      & \xmark
      & \xmark
      & \xmark
      \\

      \citet{barsse2026-arxiv}
      & \cmark
      & \cmark
      & \cmark
      & \cmark$^{\star}$
      & \xmark
      & \cmark$^{\star}$
      & \cmark$^{\star}$
      \\

      \midrule
      \textbf{Ours}
      & \cmark
      & \cmark
      & \cmark
      & \cmark
      & \xmark
      & \cmark
      & \cmark
      \\
      \bottomrule
    \end{tabular}%
  }

  \par\smallskip
  \begin{minipage}{\linewidth}
    \footnotesize
    \cmark: holds;
    \xmark: fails;
    \pmark: holds under restrictions;
    ---: not applicable in the setting considered here.

    \smallskip
    \noindent
    $^{\dagger}$SWITCH inputs cannot use classical branching.
  \end{minipage}
\end{table}

\paragraph{Physical Realisability}
No matter how elegant a program may be, it is of little use if it cannot be executed in the real world.
Physical realisability is rarely an issue for classical programs but far from trivial in the quantum setting.
While ordinary quantum circuits are widely accepted as realisable, the realisability of programs that receive quantum operations as inputs is much less clear.
For example, the quantum SWITCH cannot be implemented by a circuit with holes~\cite{Chiribella2013}, yet physical realisations have been reported~\cite{Friis2014,switch-experiment2,Rozema2024}.
The realisability of higher-order quantum operations is an active area of research, and no definitive answer is available to date.

This paper therefore addresses physical realisability from multiple perspectives.
The first perspective is the absence of physically implausible behaviour.
\begin{itemize}
    \item \emph{Bounded probabilities.}
        The total probability of the outcomes of a program must not exceed one.
        Most existing proposals satisfy this condition, but \citet{DaveDZ25} does not, since it allows superpositions with arbitrary weights.
    \item \emph{Well-definedness on completely positive maps.}
        It is widely believed that quantum procedures are described by \emph{completely positive maps}, and that procedures with the same interpretation as completely positive maps cannot be physically distinguished.
        The languages in \cite{Badescu2015,YingYF12-first-Ying,YingYF14-alternation,Ying16-the-book} are known to distinguish programs that are equal as completely positive maps, and for this reason, \citet{Badescu2015} question the physical realisability of such programs.
    \item \emph{No violation of causal inequalities.}
        The causal inequalities are inequalities on the correlations between the inputs and outputs of several parties, and it is generally believed that processes violating the causal inequalities could not be realised~\cite{Purves2021,wechs2023}.
        In particular, in the context of indefinite causal order, there exist processes such as the \emph{OCB process}~\cite{Oreshkov2012} and the \emph{Lugano process}~\cite{Baumeler2016} that violate causal inequalities and are generally believed to be physically unrealisable. It is therefore important to guarantee that our language cannot express such processes.
\end{itemize}
The second perspective is to provide a concrete implementation.
\begin{itemize}
    \item \emph{Compilation into quantum circuits.}
        We give a translation of programs, namely closed terms of first-order type, into quantum circuits.
        This does not contradict the fact that the quantum SWITCH cannot be implemented by a circuit with holes~\cite{Chiribella2013}.
        That impossibility assumes that each operation is given as a black box that may be used only once.
        Since we consider closed programs, every call to the quantum SWITCH comes with \(F\) and \(G\) written as concrete programs.
        As programs, \(F\) and \(G\) can be duplicated, and they also admit operations that are impossible for black boxes, such as constructing their controlled versions and dilations.
        Such operations, unavailable in the setting of \citet{Chiribella2013}, are essential to our implementation.
    \item \emph{Inclusion in QC-QC.}
        The compilation above handles closed programs that call SWITCH, but not SWITCH itself with its arguments left unspecified.
        We show that such programs can also be compiled by taking quantum circuits with quantum control of causal order (QC-QC)~\cite{Wechs2021} as the model of computation.
        QC-QC is a subclass of higher-order quantum operations proposed in physics as a candidate for the physically realisable ones;
        it contains the quantum SWITCH, its processes cannot violate causal inequalities, and a general implementation procedure has been outlined together with a photonic proposal.
        The translatability of our programs into QC-QC is meaningful regardless of whether one accepts QC-QC as physically realisable:
        all the conditions of physical plausibility discussed above in fact follow from it.
\end{itemize}
This paper proposes the first quantum programming language with quantum branching that comes with such arguments for physical realisability.\footnote{Independently of us, \citet{barsse2026-arxiv} have proposed an essentially equivalent language.
Their language also enjoys these properties of physical realisability, but their paper does not discuss them to this extent.
The entries marked with a star in Table~\ref{tab:comparison} are established for the first time in this paper.}

\paragraph{Our Approach}

\begin{mathfig}
  \newcommand{\controlledOp}{F}
  \begin{subfigure}{0.25\textwidth}
    \begin{align*}
      \controlledOp =
      \begin{quantikz}[baseline={(0,0)},row sep=0.3em,column sep=0.5em]
        & \gate[2]{U} & \\
        \ket0 & & \ground{} \\
      \end{quantikz}
    \end{align*}
    \raggedleft
    \vspace{-1ex}
    \caption{A dilation \( U \) of \( F \). A dilation is not unique.\\{}}
    \label{fig:control:a}
  \end{subfigure}
  \qquad
  \begin{subfigure}{0.3\textwidth}
    \begin{align*}
      \mathrm C^U\controlledOp =
      \begin{quantikz}[baseline={(0,0.1)},row sep=0.3em,column sep=0.5em]
        &  \ctrl{1} & \\
        & \gate[2]{U} & \\
        \ket0 &  & \ground{} \\
      \end{quantikz}
    \end{align*}
    \raggedleft
    \vspace{-1ex}
    \caption{An implementation of ``\(x\)-controlled \(F\)''. It depends on the choice of \(U\).}
    \label{fig:control:b}
  \end{subfigure}
  \qquad
  \begin{subfigure}{0.3\textwidth}
    \begin{align*}
      \mathrm C^U_V\controlledOp =
      \begin{quantikz}[baseline={(0,0.1)},row sep=0.3em,column sep=0.5em]
        & \octrl{2} & \ctrl{1} & \\
        & & \gate[2]{U} & \\
        \ket0 & \gate{V} & & \ground{} \\
      \end{quantikz}
    \end{align*}
    \raggedleft
    \vspace{-1ex}
    \caption{Another implementation of ``\(x\)-controlled \(F\)''. Here \(V\) is an arbitrary unitary.}
    \label{fig:control:c}
  \end{subfigure}
  \begin{subfigure}{0.49\textwidth}
    \begin{align*}
      \begin{quantikz}[baseline={(0,0.1)},row sep=0.3em,column sep=0.5em,
        wire types={q,q,n}]
        & \octrl{1} & & & \octrl{1} & & & \ctrl{1} & & & \ctrl{1} & \\
        & \gate[2]{U_1} & & &\gate[2]{U_2} & & &\gate[2]{U_2} & & &\gate[2]{U_1} & \\
        \ket0 & \wire[l][1][draw=red]{q} & \ground{}\wire[l][1][draw=red]{q}
        & \ket0 & \wire[l][1][draw=blue]{q} & \ground{}\wire[l][1][draw=blue]{q}
        & \ket0 & \wire[l][1][draw=blue]{q} & \ground{}\wire[l][1][draw=blue]{q}
        & \ket0 & \wire[l][1][draw=red]{q} & \ground{}\wire[l][1][draw=red]{q}
        \\
      \end{quantikz}
    \end{align*}
    \raggedleft
    \vspace{-1ex}
    \caption{An implementation of ``quantum SWITCH'' using \( C^U\,F \) as ``\(x\)-controlled \(F\)''.}
    \label{fig:control:d}
  \end{subfigure}
  \begin{subfigure}{0.49\textwidth}
    \tikzset{intro swap/.style n args={3}{
      draw=none,
      path picture={
        \draw[thickness,draw=#2]
          (path picture bounding box.west|-\tikzcdmatrixname-3-#1.center)
          to[out=0,in=180]
          (path picture bounding box.east|-\tikzcdmatrixname-4-#1.center);
        \draw[thickness,draw=#3,preaction={draw=white,line width=3pt}]
          (path picture bounding box.west|-\tikzcdmatrixname-4-#1.center)
          to[out=0,in=180]
          (path picture bounding box.east|-\tikzcdmatrixname-3-#1.center);
      }
    }}
    \begin{align*}
      \begin{quantikz}[baseline={(0,0.1)},row sep=0.3em,column sep=0.5em,
        wire types={q,q,n,n}]
        & \octrl{1} & & \octrl{1} & & \ctrl{1} & & \ctrl{1} & \\
        & \gate[2]{U_1} & &\gate[2]{U_2} & &\gate[2]{U_2} & &\gate[2]{U_1} & \\
        \ket0
        & \wire[l][1][draw=red]{q}
        & \gate[2,disable auto height,style={intro swap={3}{red}{blue}}][2.2em]{}
          \wire[l][1][draw=red]{q}
        & \wire[l][1][draw=blue]{q}
        & \wire[l][1][draw=blue]{q}
        & \wire[l][1][draw=blue]{q}
        & \gate[2,disable auto height,style={intro swap={7}{blue}{red}}][2.2em]{}
          \wire[l][1][draw=blue]{q}
        & \wire[l][1][draw=red]{q}
        & \ground{}\wire[l][1][draw=red]{q}
        \\
        \ket0
        & \wire[l][1][draw=blue]{q}
        & \wire[l][1][draw=blue]{q}
        & \wire[l][1][draw=red]{q}
        & \wire[l][1][draw=red]{q}
        & \wire[l][1][draw=red]{q}
        & \wire[l][1][draw=red]{q}
        & \wire[l][1][draw=blue]{q}
        & \ground{}\wire[l][1][draw=blue]{q}
      \end{quantikz}
    \end{align*}
    \raggedleft
    \vspace{-1ex}
    \caption{Our proposal.}
    \label{fig:control:e}
  \end{subfigure}
  \caption{
    Dilation-based implementations of quantum branching.
    Here \(\discardDiag\) is measurement (followed by discarding the outcome).
    The black dot represents positive control, and the white dot represents negative control.
  }
  \label{fig:intro:controlled-channel}
  \Description{}
\end{mathfig}
 
Our work starts from an analysis of the negative observation that quantum branching is not well-defined on completely positive maps~\cite{Abbott2020,Badescu2015,Ying16-the-book}.
In short, the technical problem is that the semantics of \( \qifx{x}{F}{\mathit{id}} \), which we call ``\( x \)-controlled \( F \)'', cannot be properly defined in general.
The expected behaviour is that it applies \( F \) when \( x = \ket{1} \) and does nothing when \( x = \ket{0} \).
Quantum operations satisfying this requirement always exist, but unfortunately they are not unique~\cite[Section~II.D]{Dong2019}.

Since such an operation exists, albeit not uniquely, let us consider a concrete procedure for constructing one.
For a general quantum operation \( F \), there exists a unitary \(U\) satisfying the condition in~\cref{fig:control:a}, which is called a \emph{dilation} of \( F \).
Since the dilation \( U \) is unitary, its \( x \)-controlled version is well known.
Using the \( x \)-controlled \( U \), we obtain the operation \( C^{U} F \) (see~\cref{fig:control:b}), which satisfies the requirement on the ``\( x \)-controlled \( F \)'' stated above.
Using this construction for the controlled operations, the program in \cref{eq:into:switch} can be translated into the circuit in~\cref{fig:control:d} (where \( U_1 \) and \( U_2 \) are dilations of \( G \) and \( F \), respectively).
Unfortunately, the resulting operation depends on the choice of the dilations \( U_1 \) and \( U_2 \).
Moreover, it differs from the quantum SWITCH.

Remarkably, however, this naive construction in fact amounts to the approach known as \emph{vacuum extension}~\cite{Abbott2020,Barsse2026}.
Moreover, by using other candidates for ``\( x \)-controlled \( F \)'', we can also cover the approach of \citet{Badescu2015}.
As shown in \citet{Dong2019}, the operation \( C^U_V F \) (see~\cref{fig:control:c}) satisfies the requirement on the ``\( x \)-controlled \( F \)'' for every \( V \).
The approach of \citet{Badescu2015} can be understood as using \( C^{U}_{H} F \), where \( H \) is the Hadamard gate.
This construction likewise depends on the choice of \( U_1 \) and \( U_2 \), and it does not coincide with the quantum SWITCH either.

Our new construction translates the program in \cref{eq:into:switch} into the circuit in \cref{fig:control:e}.
The key point is that the \( x \)-controlled \( U_2 \) (corresponding to \( F \) in the then-branch) and the \( \neg x \)-controlled \( U_2 \) (corresponding to \( F \) in the else-branch) share the same extra input (blue line).
Although the construction itself depends on the choice of the dilations \( U_1 \) and \( U_2 \), the semantics of the resulting circuit is independent of this choice and moreover coincides with the quantum SWITCH.

This construction requires a restriction on the programming language.
Namely, if the then-branch uses \( F \) once, then \( F \) must also appear exactly once in the else-branch; in other words, operations must be used linearly (in the sense of Girard's \emph{linear logic}~\cite{Girard1987}).
Our main contribution is to show that this mild restriction suffices for strong physical realisability.

\paragraph{Contributions of this Paper}
Our contributions consist of the observations on the implementation of quantum branching discussed above and the mathematical results supporting them, which are summarised as follows.
\begin{itemize}
    \item
        We show that the form of quantum control exhibited by the quantum SWITCH is governed by linearity (in the sense of linear logic~\cite{Girard1987}).
    \item
        We design a programming language that permits a definition of the quantum SWITCH~\cite{Chiribella2013} in the form of \cref{eq:into:switch} for quantum operations \( F \) and \( G \).
        \Cref{eq:into:switch} describes a situation where the causal order of quantum operations is indefinite, and such a situation has recently attracted significant attention.
        Although languages of this kind have long been anticipated and studied for over a decade, their feasibility has been questioned due to semantic difficulties.
    \item
        We develop a denotational semantics for the proposed language.
        Our denotational model is similar to those of \citet{DaveDZ25} and \citet{barsse2026-arxiv}, consisting of three components: the category \( \Hilb \) of Hilbert spaces and linear maps for modelling the quantum sublanguage, the category \( \CPM \) of completely positive maps for modelling the classical sublanguage, and a functor \( \EmbeddingFunctor \colon \Hilb \longrightarrow \CPM \) embedding the former into the latter.
        This denotational semantics treats every input function as a \emph{black-box} function, and thus gives a unique semantics regardless of their implementation.
        As a corollary, we prove that the semantics does not depend on the choice of dilations,
        suggesting that no ill-definedness issue remains in our language.
    \item
        We provide a circuit-based operational semantics based on a program transformation.
        It proves that every program in our language can be compiled down to a quantum circuit in the white-box setting, showing its physical realisability.
        We prove soundness, adequacy, and a full-abstraction theorem for the language.
        In particular, the soundness theorem proves that, even though each step of the operational semantics may depend on implementation, the final outcome does not.
    \item
        By adapting the ideas of our translation and operational semantics, we also give translations into circuits for the languages proposed by \citet{Ying16-the-book}, \citet{Badescu2015}, \citet{Barsse2026}, and \citet{barsse2026-arxiv}, for which no such translation was previously known.
    \item 
        We prove that programs in our language are describable as QC-QC~\cite{Wechs2021}.
        In the proof, the translatability into circuits plays an essential role.
        Inclusion in QC-QC is itself evidence of a certain degree of physical realisability, and moreover, various forms of physical plausibility follow from it as corollaries.
        In particular, the result that no program violates causal inequalities has not been established in any previous work.
\end{itemize}
Among these, the language design and the categorical semantics overlap with the independent work of \citet{barsse2026-arxiv}.

\paragraph{Organisation}
\Cref{sec:pre} gives the preliminaries of this paper and introduces some prior knowledge on quantum computation.
\Cref{sec:overview} is a technical overview of this paper, providing a more detailed discussion on the technical ideas.
We focus primarily on conveying the underlying intuition rather than providing rigorous formalisation.
The subsequent sections provide a formal development of the ideas introduced in \cref{sec:overview}.
\Cref{sec:proc} formalises the proposed language \Qif{}, and
\cref{sec:cat-sem} defines its categorical semantics.
In \cref{sec:op-sem}, we discuss the program transformation and the operational semantics. We also state many properties, including normalisation, soundness, adequacy, and full abstraction.
\Cref{sec:related} discusses related work,
and we conclude in \cref{sec:conc}.

\section{Preliminaries}\label{sec:pre}
This section briefly reviews notions related to quantum computation.
A quantum system is modelled as a Hilbert space, and there are two formulations of the states of a quantum system: one represents a state as a vector in the Hilbert space, and the other as a density operator, a certain kind of Hermitian operator on the Hilbert space.
We shall switch between these two formalisms as needed, so we introduce both.

\paragraph{Hilbert Spaces}
Let \( \mathcal{H} \) be a \emph{Hilbert space}, which is a \( \CC \)-vector space with an inner product such that \( \mathcal{H} \) is complete with respect to the metric induced by the inner product.
We write \( \langle x \,|\, y \rangle \) for the inner product, which is linear on the second argument \( y \).
Following physics convention, we write \( \ket{x} \) to denote a vector.
A vector \( x \) in a Hilbert space \( \mathcal{H} \) induces a \( \CC \)-linear map \( \mathcal{H} \longrightarrow \CC \) defined by \( \ket{y} \mapsto \langle x \,|\, y \rangle \), which we write as \( \bra{x} \).

This paper only deals with finite-dimensional Hilbert spaces.
A finite-dimensional Hilbert space is isomorphic to \( \CC^n \), \( n \in \Nat \), with the inner product defined by \( \langle x\,|\,y \rangle = \sum_{i=1}^n x_i^\dagger y_i \) for \( x = (x_1,\dots,x_n) \) and \( y = (y_1,\dots,y_n) \) (where \( ({-})^\dagger \) is the complex conjugate).
So we assume \( \mathcal{H} = \CC^n \) for some \( n \).
Note that \( \CC^n \) has a canonical orthonormal basis, consisting of \( e_i = (\delta_{i1}, \dots, \delta_{ij},\dots,\delta_{in}) \) where \( \delta_{ij} \) is the Kronecker delta (\ie,~\( \delta_{ii} = 1 \) and \( \delta_{ij} = 0 \) if \( i \neq j \)).
Since \( \CC^n \) and \( \CC^m \) have chosen bases, a linear map \( f \colon \CC^n \longrightarrow \CC^m \) is associated with a matrix representation \( A = (a_{i,j})_{i\in\{1,\dots,m\}, j\in\{1,\dots,n\}} \) where \( a_{i,j} = \langle e_i \,|\, f \,|\,e_j \rangle \).
So we shall identify a linear map with its matrix representation.

Let \( \Mat_{n,m}(\CC) \) be the set of \( (n \times m) \)-matrices over \( \CC \) and \( \Mat_n(\CC) = \Mat_{n,n}(\CC) \).
A matrix \( A = (a_{i,j})_{i,j \in \{ 1,\dots,n\}}\in \Mat_n(\CC) \) is \emph{self-adjoint} or \emph{Hermitian} if \( A^\dagger = A \), where \( A^\dagger \defeq (a_{j,i}^\dagger)_{i,j \in \{1,\dots,n\}} \) is the conjugate transpose of \( A \).
A self-adjoint \( A \in \Mat_n(\CC) \) is \emph{positive} if \( \bra{x} A \ket{x} \in \Real_{\ge 0} \) for every vector \( x \in \CC^n \).
We write \( 0 \le A \) to mean that \( A \) is a positive self-adjoint matrix.
The \emph{trace} \( \trace A \) of \( A \in \Mat_n(\CC) \) is the sum of the diagonal elements of \( A \) (\ie~\( \trace A = \sum_{i=1}^n a_{i,i} \) for \( A = (a_{i,j})_{i,j \in \{1,\dots,n\}} \)).
For positive \( A \), \( \trace A \in \Real_{\ge 0} \).
A matrix \( A = (a_{i,j})_{i,j \in \{ 1,\dots,n\}}\in \Mat_n(\CC) \) is \emph{unitary} if \( A \) has an inverse \( A^{-1} \) and satisfies \( A^\dagger = A^{-1} \).
A matrix \( U \) is unitary if and only if the corresponding linear map preserves the inner product, \ie~\( \langle x \,|\, y \rangle = \langle U\,x \,|\, U\,y \rangle \).

\paragraph{States as Unit Vectors}
A \emph{state} of a quantum system represented by a Hilbert space \( \mathcal{H} \) is a unit vector \( \ket{\varphi} \in \mathcal{H} \) (\ie~\( \|\varphi\| = \sqrt{\langle \varphi \,|\, \varphi \rangle} = 1 \)), and the dynamics of the system is represented by a unitary operator \( U \colon \mathcal{H} \longrightarrow \mathcal{H} \).
\begin{example}
    Let \( \mathcal{H} = \CC^2 \) (as we shall see, this is the denotation of the qubit type \( \qbit \)).
    We write \( \ket{0} \) and \( \ket{1} \) for the orthonormal basis of \( \mathcal{H} \).
    A state \( \ket{\varphi} \) in \( \mathcal{H} \) is of the form \( \alpha \ket{0} + \beta \ket{1} \) such that \( |\alpha|^2 + |\beta|^2 = 1 \).
    A state with non-zero \( \alpha \) and \( \beta \) is a superposition of \( \ket{0} \) and \( \ket{1} \), which has no counterpart in classical computation.
    Examples of unitary operators on \( \CC^2 \) are as follows: the operator \( X \) given by \( X (\alpha \ket{0} + \beta \ket{1}) := \beta \ket{0} + \alpha \ket{1} \) is known as the Pauli \(X\) operator, which corresponds to the negation of a classical bit (note that \( X \ket{0} = \ket{1} \) and \( X \ket{1} = \ket{0} \)); \( Z (\alpha \ket{0} + \beta \ket{1}) := \alpha \ket{0} - \beta \ket{1} \) is known as the Pauli \( Z \) operator; \( H \) defined by \( H \ket{0} := (\ket{0} + \ket{1})/\sqrt{2} \) and \( H \ket{1} := (\ket{0} - \ket{1})/\sqrt{2} \) is known as the Hadamard operator.
    We write \( \ket{+} \) and \( \ket{-} \) for \( H \ket{0} \) and \( H \ket{1} \) as usual.
    \thmend
\end{example}

An important non-unitary operation is \emph{measurement}.
Let \( \ket{\varphi} \) be a quantum state in a Hilbert space \( \mathcal{H} = \mathcal{H}' \otimes \CC^2 \).
Then \( \ket{\varphi} \) is canonically written as \( \ket{\varphi} = \ket{\varphi_0} \otimes \ket{0} + \ket{\varphi_1} \otimes \ket{1} \) for some vectors \( \ket{\varphi_0} \) and \( \ket{\varphi_1} \) in \( \mathcal{H}' \).
The measurement on the \( \CC^2 \) component yields the outcome \( 0 \) with probability \( \|\varphi_0\|^2 \) and the outcome \( 1 \) with probability \( \|\varphi_1\|^2 \).
Furthermore, when the outcome is \( i \), the state becomes \( \varphi_i/\|\varphi_i\| \) (a unit vector in \( \mathcal{H}' \)).

\paragraph{States as Density Operators}
A probabilistic mixture of quantum states is called a \emph{mixed state}.
Formally, a mixed state of a quantum system \( \mathcal{H} = \CC^n \) is a \emph{density operator}, a positive self-adjoint \( (n \times n) \)-matrix \( D \in \Mat_n(\CC) \) of trace \( 1 \).
The dynamics of the system is represented by a \emph{superoperator}, a \( \CC \)-linear function \( f \colon \Mat_n(\CC) \longrightarrow \Mat_m(\CC) \) that maps a mixed state to a mixed state and satisfies an additional condition.
The preservation of mixed states means that \( f \) is \emph{positive} (\ie~\( 0 \le x \) implies \( 0 \le f(x) \)) and \emph{trace-preserving} (\ie~\( \trace x = \trace f(x) \) for every \( x \ge 0 \)).
The additional condition strengthens the positivity, and is called \emph{complete positivity}.
\begin{definition}
    A linear map \( f \colon \Mat_n(\CC) \longrightarrow \Mat_m(\CC) \) is \emph{completely positive} if
    \begin{equation*}
        \textstyle
        f(D)
        \quad=\quad
        \sum_{i = 1}^k A_i D A_i^\dagger
    \end{equation*}
    for some \( k \) and \( A_i \in \Mat_{m,n}(\CC) \).
    The right-hand side of the above equation is called a \emph{Kraus decomposition} of the completely positive map \( f \).
    We write the above situation as \( f = \{ A_i \}_{i = 1,\dots,k} \).
    The decomposition is not unique.
    \thmend
\end{definition}

\section{Overview}\label{sec:overview}
This section presents the core ideas of the paper.
After reviewing the quantum SWITCH~\cite{Chiribella2013} in \cref{sec:overview:switch}, \cref{sec:overview:past} shows that \cref{eq:into:switch} fails under the semantics proposed in the literature.
\Cref{sec:overview:correspondence-problem-mathematically} studies the ill-definedness issue from a mathematical point of view and identifies it as the \emph{correspondence problem}.
\Cref{sec:overview:transformation} and \Cref{sec:overview:dummy-free} also discuss and analyse the ill-definedness issue, but in terms of program transformations.
\Cref{sec:overview:language-design} proposes our solution,
which will be formalised from \cref{sec:proc} onward.

\subsection{Quantum SWITCH}\label{sec:overview:switch}
The \emph{quantum SWITCH}~\cite{Chiribella2013} is an operation \( \mathit{SWITCH}(x,y,F,G) \) that takes a \emph{control qubit} \( x \), an operand qubit \( y \), and two quantum operations \( F \) and \( G \), and returns a pair \( (x', y') \) consisting of the control and operand qubits after the operation.
When \( F \) and \( G \) have Kraus decompositions \( \{ K_i \}_{i \in I} \) and \( \{ H_j \}_{j \in J} \), respectively, the quantum operation \( \mathit{SWITCH}({-},{-},F,G) \) on two qubits is defined by the following Kraus decomposition:
\begin{equation}
  \mathit{SWITCH}({-},{-},F,G) = \big\{\, \ketbra{0}{0} \otimes \diffB{H_j K_i} + \ketbra{1}{1} \otimes \diffA{K_i H_j} \,\big\}_{(i,j) \in I \times J}.
  \label{eq:switch-Kraus}
\end{equation}
At first glance, this definition appears to depend on the choice of Kraus decompositions.
Interestingly, however, it is in fact independent of that choice.
That is, for different Kraus decompositions \( \{K'_{i'}\}_{i' \in I'} \) and \( \{ H'_{j'} \}_{j' \in J'} \) for \(F\) and \(G\), even when \( \{ \ketbra{0}{0} \otimes \diffB{H'_{j'} K'_{i'}} + \ketbra{1}{1} \otimes \diffA{K'_{i'} H'_{j'}} \}_{(i',j') \in I' \times J'} \) differs from the above as Kraus decompositions, these Kraus decompositions define the same quantum operation.
Hence, the quantum SWITCH is well-defined on quantum operations.

It is known that the quantum SWITCH enables a variety of interesting phenomena~\cite{Araujo2014,Colnaghi2012,switch-advantage2,Kristjansson2024,Kristjansson2020,Ebler2018,Chiribella2021}.
Here, we introduce one such phenomenon~\cite{Chiribella2021}, which will be used in the discussion in the next subsection.

Let \( D \) be a quantum operation given by \( D(\varrho) = \ketbra{1}{0}\varrho\ketbra{0}{1} + \ketbra{0}{1}\varrho\ketbra{1}{0} \).
This is expressed as an operation involving measurement, \( \lambda x. \cifx{\meas(x)}{\ket{0}}{\ket{1}} \), which measures the input qubit \( x \) and returns \( \ket{0} \) (resp.~\( \ket{1} \)) when the measurement outcome is \( 1 \) (resp.~\( 0 \)).
Since \( D \) measures the input qubit, one cannot reconstruct \( y \) from \( D\,y \).
Interestingly, \citet{Chiribella2021} show that one can reconstruct \( y \) from \( \mathit{SWITCH}(\ket{+}, y, D, D) \).

The peculiarity of this phenomenon becomes apparent in the following consideration.
Let \( \mathit{SWITCH}^{?} \) be the following operation, whose definition resembles the quantum SWITCH:
\begin{align*}
  \mathit{SWITCH}^{?}({-},{-},F,G) = \big\{\, \ketbra{0}{0} \otimes \diffA{K_i H_j} + \ketbra{1}{1} \otimes \diffA{K_i H_j} \,\big\}_{(i,j) \in I \times J}.
\end{align*}
In a language that admits the definition \eqref{eq:into:switch} of the quantum SWITCH, we should have
\begin{align*}
  \mathit{SWITCH}^{?}(x,y,F,G) \quad=\quad \qifx{x}{\diffA{F(G(y))}}{\diffA{F(G(y))}}.
\end{align*}
Obviously, \( \mathit{SWITCH}^{?}(x,y,F,G) \) differs from \( \mathit{SWITCH}(x,y,F,G) \) for general \( F \) and \( G \), but one may expect them to coincide when \( F = G \), since
\begin{equation}
  \mathit{SWITCH}^{?}(x,y,F,F) \quad\stackrel{?}{=}\quad \qifx{x}{F(F(y))}{F(F(y))} \quad\stackrel{?}{=}\quad \mathit{SWITCH}(x,y,F,F).
  \label{eq:switch-vs-switchq}
\end{equation}
However, this equation fails.
Since \( \mathit{SWITCH}^{?}({-},{-},F,G) = \ident \otimes (\diffA{F \circ G}) \), one cannot reconstruct \( y \) from \( \mathit{SWITCH}^{?}(\ket{+},y,D,D) = \ket{+} \otimes (D(D(y))) \).
Hence \( \mathit{SWITCH}^{?}(\ket{+},y,D,D) \neq \mathit{SWITCH}(\ket{+},y,D,D) \).

\subsection{Existing Quantum Conditional Branching and Quantum SWITCH}\label{sec:overview:past}
\citet{Badescu2015}, \citet{Ying16-the-book}, and \citet{Barsse2026} proposed languages with quantum conditional branching for general quantum operations.
This subsection reviews their proposals and shows that \cref{eq:into:switch} does not define the quantum SWITCH~\cite{Chiribella2013} in their languages.

In their languages, the semantics of a program \( P \) uses a Kraus decomposition \( \sem{P} = \{ K_i \}_{i \in I} \).\footnote{Precisely speaking, the semantics of \citet{Ying16-the-book} is a function \( \Psi \) from a finite set \( I \) to operators.  Here we identify such an operator-valued function with the Kraus decomposition \( \{ \Psi(i) \}_{i \in I} \).}\footnote{\citet{Barsse2026} claim that their semantics does not use a Kraus decomposition. However, their semantics can be reformulated in a way that does depend on a specific Kraus decomposition. See \cref{app:sec:vacuum} for details.}
The semantics of the quantum conditional branching
by \citet{Badescu2015} is given by
\begin{equation}
  \sem{\lambda xy. \qifx{x}{P(y)}{Q(y)}}
  \quad\defeq\quad
  \left\{\,
  \ketbra{0}{0} \otimes \dfrac{H_j}{\sqrt{|I|}} + \ketbra{1}{1} \otimes \dfrac{K_i}{\sqrt{|J|}}
  \,\right\}_{(i,j) \in I \times J},
  \label{eq:qif-old-semantics}
\end{equation}
where \( \sem{\lambda y. P(y)} = \{ K_i \}_{i \in I} \), \( \sem{\lambda y.Q(y)} = \{ H_j \}_{j \in J} \), and \( |I| \) is the number of elements in \( I \).
The coefficient \( 1/\sqrt{|I|} \) is needed for normalisation because \( H_j \) appears \( |I| \) times in the above decomposition (as \( \ketbra{0}{0} \otimes H_j \)).
The semantics proposed by \citet{Ying16-the-book} is similar but uses coefficients different from \( 1/\sqrt{|J|} \) and \( 1/\sqrt{|I|} \).
Also, the semantics proposed by \citet{Barsse2026} can be understood as
\begin{align*}
  \sem{\lambda xy. \qifx{x}{P(y)}{Q(y)}}
  \quad & \defeq\quad
  \left\{\,
  \ketbra{0}{0} \otimes H_0 + \ketbra{1}{1} \otimes K_0
  \,\right\}
  \\
        &
  \qquad\
  \cup
  \left\{\,
  \ketbra{0}{0} \otimes H_j
  \,\right\}_{j \in J\setminus\{0\}}
  \cup
  \left\{\,
  \ketbra{1}{1} \otimes K_i
  \,\right\}_{i \in I\setminus\{0\}}
  .
\end{align*}
This definition may look different, but it can still be regarded as a variant of the others, as we will see below.
For the moment, we focus on the semantics in \cref{eq:qif-old-semantics}.

However, the ``defining equation'' of the quantum SWITCH \eqref{eq:into:switch} is not valid in their semantics.
Assume that \( \sem{\lambda y.F(y)} = \{ K_i \}_{i \in I} \) and \( \sem{\lambda y.G(y)} = \{ H_j \}_{j \in J} \).
Then \( \sem{\diffA{F \circ G}} = \{ \diffA{K_i H_j} \}_{(i,j) \in I \times J} \) and \( \sem{\diffB{G \circ F}} = \{ \diffB{H_j K_i} \}_{(i,j) \in I \times J} \), so
\begin{equation*}
  \sem{\lambda xy. \qifx{x}{\diffA{F(G(y))}}{\diffB{G(F(y))}}} = \left\{ \ketbra{0}{0} \otimes \dfrac{\diffB{H_j K_i}}{\sqrt{|I \times J|}} + \ketbra{1}{1} \otimes \dfrac{\diffA{K_{i'}H_{j'}}}{\sqrt{|I \times J|}} \right\}_{(i,j,i',j') \in I \times J \times I \times J}
\end{equation*}
whereas
\(
\mathit{SWITCH}({-},{-},F,G) = \left\{ \ketbra{0}{0} \otimes \diffB{H_j K_i} + \ketbra{1}{1} \otimes \diffA{K_{i} H_{j}} \right\}_{(i,j) \in I \times J}
\).
In particular, when $F = G = D$, a direct calculation shows that \( \sem{\lambda xy. \qifx{x}{D(D(y))}{D(D(y))}} \)
differs from \( \mathit{SWITCH}\allowbreak({-},{-},D,D) \).

This failure is not specific to these particular semantics;
it is already forced by the design of the languages.
Choose $F$ and $G$ to be the same quantum operation.
Then by definition of their semantics,
\begin{align*}
  \sem{\lambda xy. \qifx{x}{\diffA{F(G(y))}}{\diffB{G(F(y))}}}
  \quad= & \quad
  \sem{\lambda xy. \qifx{x}{F(F(y))}{F(F(y))}}
  \\
  \quad= & \quad
  \sem{\lambda xy. \qifx{x}{\diffA{F(G(y))}}{\diffA{F(G(y))}}}.
\end{align*}
Hence, whatever semantics is chosen, the former cannot denote
\( \mathit{SWITCH}(x,y,F,G) \) while the latter denotes
\( \mathit{SWITCH}^{?}(x,y,F,G) \):
otherwise \cref{eq:switch-vs-switchq} would hold,
contradicting the fact that $\mathit{SWITCH}(-,-,D,D)$ and $\mathit{SWITCH}^{?}(-,-,D,D)$ are distinct.

Another issue is that the semantics is ill-defined on completely positive maps.
A completely positive map has many different Kraus decompositions: for example, \( F = \{ I_2 \} \) and \( F' = \{ -I_2 \} \) (where \( I_2 \in \Mat_{2}(\CC) \) is the unit matrix) define the same completely positive map \( F(\varrho) = F'(\varrho) = \varrho \).
The semantics of a quantum conditional in these languages depends on the choice of Kraus decomposition, as the authors observed.
For example,
\begin{equation*}
  \sem{\lambda xy. \qifx{x}{F(y)}{y}} = \{ \redemph{I_2} \otimes I_2 \}
  \quad\mbox{and}\quad
  \sem{\lambda xy. \qifx{x}{F'(y)}{y}} = \{ \redemph{Z} \otimes I_2 \}.
\end{equation*}

\subsection{Correspondence Problem, Mathematically}
\label{sec:overview:correspondence-problem-mathematically}

The key difference between the quantum SWITCH and
the semantics proposed in the literature \cite{Badescu2015,Ying16-the-book,Barsse2026} is
the existence of a \emph{correspondence of indices} between the then- and else-branches.
In the quantum SWITCH \eqref{eq:switch-Kraus}, common indices \( (i,j) \) are used in both the then- and else-cases, whereas in the semantics of quantum conditional branching \eqref{eq:qif-old-semantics}, the indices \( (i,j) \) and \( (i', j') \) for the then- and else-branches are selected independently.

The independent choice of indices for the then- and else-branches is, in our view, a characteristic feature of the approach based on quantum-controlled quantum operations.
Under this approach,
the index set of the then-branch may differ from that of the else-branch.
This is already visible in the typical quantum-controlled quantum operation, \( \qifx{x}{F(y)}{y} \): if \(F\) is represented by a Kraus decomposition indexed by \(I\), then the then-branch has index set \(I\), whereas the else-branch has a singleton index set.
Since the index sets \(I\) and \(I'\) of the then- and else-branches need not coincide,
a natural choice is to choose the indices of the then- and else-branches independently, resulting in a Kraus decomposition indexed by \( I \times I' \).

Conversely, this discussion suggests a characteristic feature of the quantum control exhibited by the quantum SWITCH: the indices of the Kraus decompositions of the then- and else-branches must match. This is what we call the \emph{correspondence problem}.

It is not enough for the then- and else-branches to have the same number of Kraus operators.
Even when the two Kraus decompositions have the same number of elements, changing the correspondence between their elements changes the resulting quantum operation.
For example, consider two Kraus decompositions \( F = \{ K_1, K_2 \} \) and \( G = \{ H_1, H_2 \} \) where \( H_1 = K_1 = \frac{X}{\sqrt{2}} \) and \( H_2 = K_2 = \frac{\ident}{\sqrt{2}} \).
Then, matching \( (H_1,K_1) \) and \( (H_2,K_2) \) gives a different quantum operation from matching \( (H_1,K_2) \) and \( (H_2,K_1) \):
\begin{align*}
  \left\{
  \frac{\ketbra{0}{0} \!\otimes\! \redemph{X} \!+\! \ketbra{1}{1} \!\otimes\! \redemph{X}}{\sqrt{2}},\:
  \frac{\ketbra{0}{0} \!\otimes\! \redemph\ident \!+\! \ketbra{1}{1} \!\otimes\! \redemph{\ident}}{\sqrt{2}}
  \right\}
  \not\simeq
  \left\{
  \frac{\ketbra{0}{0} \!\otimes\! \redemph{X} \!+\! \ketbra{1}{1} \!\otimes\! \redemph{\ident}}{\sqrt{2}},\:
  \frac{\ketbra{0}{0} \!\otimes\! \redemph\ident \!+\! \ketbra{1}{1} \!\otimes\! \redemph{X}}{\sqrt{2}}
  \right\}.
\end{align*}
Therefore, the correspondence problem is not only about matching the sizes of the index sets: it requires a canonical correspondence between the elements of the two index sets.

The example of the quantum SWITCH gives an intuitive understanding of what such a canonical correspondence looks like.
The Kraus decomposition of the then-branch is \( \{ {K_i H_j} \}_{(i,j) \in I\times J} \), whereas that of the else-branch is \( \{ {H_j K_i} \}_{(i,j)\in I\times J} \).
The canonical correspondence matches \({K_iH_j}\) with \({H_jK_i}\), for each pair \((i,j)\in I\times J\).

We have now arrived at one self-contained explanation of the motivation behind our language design: the language should be designed so as to guarantee a canonical correspondence between the indices of the Kraus decompositions of the then- and else-branches.
We could now proceed directly to the solution (cf.~\cref{sec:overview:language-design}), but we first examine the same problem from another perspective, both to clarify the issue conceptually and to prepare for implementation.

\subsection{Analysis of Semantics and Ill-Definedness via Program Transformation}\label{sec:overview:transformation}

We analyse this problem from the perspective of programming languages,
and consider its operational meaning.
We propose a program transformation
\begin{align*}
   & \mbox{(a program with quantum branching involving \emph{quantum operations})}
  \\
   & \qquad\longmapsto\quad \mbox{(a program with quantum branching over \emph{unitary operations})}.
\end{align*}
Since the latter is known to be implementable, this transformation, together with an implementation of the target language, yields an implementation of the source language.

\paragraph{Source Programming Language}
Before proceeding to the discussion, we briefly describe the programming language used here.

The language has two types, \( \keyword{bool} \) and \( \keyword{qubit} \).
The \( \keyword{qubit} \) type is subject to the linearity constraint due to physical considerations.
For example, in the let-binding \( \mathtt{let}\,q = \mathtt{X}(p) \) (where \( \mathtt{X} \) is the X gate and \( q, p \colon \keyword{qubit} \)), the variable \( p \) is consumed here, so it can no longer be used thereafter.

The language provides two control constructs: the conventional \( \mathtt{if} \), based on Boolean values, and \( \qif \), a conditional branching construct based on \( \keyword{qubit} \).
They are expressions, so both forms of conditional branching return some values.
A conventional branching \( (\cifx{b}{P}{Q}) \) returns only the evaluation result of \( P/Q \).
A quantum branching \( (\qifx{x}{P}{Q}) \) returns the condition qubit \( x \) in addition to the evaluation result of \( P/Q \).
This corresponds to the behaviour of a controlled unitary gate in quantum circuits, which also does not consume the control qubit.
A standard but important constraint for both forms of conditional branching is that the types of \( P \) and \( Q \) must coincide.

In addition to these control constructs, the language includes built-in operators such as arbitrary unitary operations, the creation of a new qubit \( \ket{0} \), and quantum measurement \( \keyword{meas} \).
Among these, \( \keyword{meas} \) is particularly important.
It consumes a value of type \( \keyword{qubit} \) and produces a value of type \( \keyword{bool} \).

\paragraph{Program Transformation}
We propose a series of seemingly semantics-preserving translations, taking a program with quantum-controlled quantum operations and yielding a program with quantum-controlled unitary operations.

Our translation is inspired by \emph{Stinespring's dilation theorem}.
By Stinespring's dilation theorem, given a quantum operation \( f \colon  \Mat_{2^n}(\CC) \longrightarrow  \Mat_{2^m}(\CC) \) from the \( n \)-qubit system to the \( m \)-qubit system, there exists \( k \ge n, m \) and a unitary \( U \colon \CC^{2^k} \longrightarrow \CC^{2^k} \) such that \( f(\varrho) = \mathrm{tr}_{k-m}(U (\varrho \otimes \ketbra{\smash[t]{\vec{0}}}{\smash[t]{\vec{0}}}) U^\dagger) \), where \( \vec{0} \) is the sequence of \( 0 \)'s of length \( k-n \) and \( \mathrm{tr}_{k-m} \) is the partial trace over the \( \Mat_{2^{k-m}}(\CC) \) component in \( \Mat_{2^k}(\CC) \cong \Mat_{2^m}(\CC) \otimes \Mat_{2^{k-m}}(\CC) \).
In the programming language terminology,
it states that a given program \( F(x) \) can be transformed into
\begin{quote}
  \centering
  \begin{qifnoframe}
    let a = |0...0>;  let (y, z) = U(x, a);  let _ = meas z;  y
  \end{qifnoframe}
\end{quote}
where \( x \colon \keyword{qubit}^n \), \( y \colon \keyword{qubit}^m \), \( a \colon \keyword{qubit}^{k-n} \), \( z \colon \keyword{qubit}^{k-m} \) and \( U \) is a unitary, \ie,~a program without measurement.
Our goal is to implement a program transformation of this kind.\footnote{Note that the existence of such a transformation for our language is \textbf{NOT} a consequence of Stinespring's theorem.  The theorem is applicable only to semantic elements such as quantum operations and completely positive maps, whereas the ``semantics'' of our language has not yet been given.  In particular, the semantics of measurements inside \(\qif\) is questionable.  Instead, we aim to define the ``semantics'' via the translation.}

It is not difficult to transform programs that contain no \( \keyword{qif} \).
We can delay a measurement using the \emph{principle of deferred measurement} (see~\cite{QCQI})
\begin{quote}
  \centering
  \begin{minipage}{.27\linewidth}
    \begin{qifnoframe}
      let b = meas q;
      let r = if b { P } else { Q }
    \end{qifnoframe}
  \end{minipage}
  \qquad$=$\qquad\qquad
  \begin{minipage}{.33\linewidth}
    \begin{qifnoframe}
      let (q1, r) = qif q { P } else { Q };
      let b = meas q1
    \end{qifnoframe}
  \end{minipage}
\end{quote}
and move all qubit creations to the beginning (see \cref{fig:dilation}).
\begin{figure}[t]
  \centering
  \begin{subfigure}{0.25\textwidth}
    \begin{lstlisting}[language=qif]
      let b = meas q;
      let p = if b { |1>,, } else { |0>,, };
      p
    \end{lstlisting}
    \raggedleft
    \vspace{-1ex}
    \caption{}
    \label{fig:dilation:a}
  \end{subfigure}
  \qquad
  \begin{subfigure}{0.25\textwidth}
    \begin{lstlisting}[language=qif]
      let (q1, p) =
          qif q { |1>,, } else { |0>,, };
      let _ = meas q1;
      p
    \end{lstlisting}
    \vspace{-1ex}
    \caption{}
    \label{fig:dilation:b}
  \end{subfigure}
  \qquad
  \begin{subfigure}{0.2\textwidth}
    \begin{lstlisting}[language=qif]
      let p0 = |0>;
      let (q1, p) = CX(q, p0);
      let _ = meas q1;
      p
    \end{lstlisting}
    \vspace{-1ex}
    \caption{}
    \label{fig:dilation:c}
  \end{subfigure}
  \caption{Dilation of a program with no \( \keyword{qif} \).
    (a) Source program.
    (b) Deferring the measurement.
    (c) Moving the qubit creation operations \( \ket{0} \) and \( \ket{1} \) to the beginning.
    Here \qifinline{CX} is the controlled X gate.}
  \label{fig:dilation}
  \Description{}
\end{figure}

The transformation of quantum conditional branching is challenging and indeed problematic.
We explain the idea using the source program in \cref{fig:translation:a}.
Our goal is to \emph{hoist} the measurement \( \keyword{meas}\,q_1 \) on line~6 out of \( \keyword{qif} \).
A naive idea would be to add \( q_1 \) to the return value of the then-branch and perform the measurement on the return value outside \( \keyword{qif} \) (\cref{fig:translation:b}).
However, this leads to a type mismatch: the then-branch returns a pair of qubits, whereas the else-branch returns only a single qubit.
The simplest way to resolve this type inconsistency is to introduce a dummy qubit value \( \ket{\varphi} \) and transform the else-branch so that it additionally returns the dummy value (\cref{fig:translation:c}).
Although this may seem like a rather naive solution, the semantics discussed in previous work~\cite{Badescu2015,Ying16-the-book,Abbott2020,Barsse2026} can be understood as corresponding to this transformation.
Once the measurement is moved outside \( \keyword{qif} \), we can then move the qubit creations \( \ket{0} \) and \( \ket{\varphi} \) to the position preceding \( \keyword{qif} \), thereby obtaining the desired normal form.
\begin{figure}[t]
  \centering
  \begin{subfigure}{0.25\textwidth}
    \begin{lstlisting}[language=qif]
      // x : qbit, q : qbit
      let (x1, r) =
          qif x {
              let p0 = |0>;
              let (q1, p) = CX(q, p0);
              let _ = meas q1;
              p
          } else {
              q
          };
      (x1, r)
    \end{lstlisting}
    \vspace{-1ex}
    \caption{}
    \label{fig:translation:a}
  \end{subfigure}
  \qquad
  \begin{subfigure}{0.25\textwidth}
    \begin{lstlisting}[language=qif]
      // x : qbit, q : qbit
      let (x1, r, q2) =
          qif x {
              let p0 = |0>;
              let (q1, p) = CX(q, p0);
              (p, q1)
          } else {
              (q, ??)
          };
      let _ = meas q2;
      (x1, r)
    \end{lstlisting}
    \vspace{-1ex}
    \caption{}
    \label{fig:translation:b}
  \end{subfigure}
  \qquad
  \begin{subfigure}{0.25\columnwidth}
    \begin{lstlisting}[language=qif]
      // x : qbit, q : qbit
      let (x1, r, q2) =
          qif x {
              let p0 = |0>;
              let (q1, p) = CX(q, p0);
              (p, q1)
          } else {
              (q, |φ>)
          };
      let _ = meas q2;
      (x1, r)
    \end{lstlisting}
    \vspace{-1ex}
    \caption{}
    \label{fig:translation:c}
  \end{subfigure}
  \hphantom{\quad\begin{subfigure}{0.25\columnwidth}a\end{subfigure}}
  \vspace{-2ex}
  \caption{An example of the program transformation.
    (a) Source program.
    (b) Hoisting the measurement outside qif results in an ill-typed program.
    (c) Adding a dummy value \( \ket{\varphi} \) makes the program well-typed.
    Then, moving the qubit creations \( \mathtt{let}\,p_0 = \ket{0} \) before \( \keyword{qif} \) and replacing \( \ket{\varphi} \) with \( U\,p_0 \) yields a dilation.
    Here \( U \) is a unitary satisfying \( U \ket{0} = \ket{\varphi} \).
  }
  \label{fig:prog-delay}
  \Description{}
\end{figure}

The above discussion clarifies the source of the ill-definedness issue: it lies in the arbitrariness of the dummy value, which acts as the counterpart of a measurement in the other branch.
Perhaps surprisingly, we can show that the ill-definedness discussed in the literature~\cite{Abbott2020,Ying16-the-book,Badescu2015,Barsse2026} can be understood precisely as this freedom to choose dummy values.
For this program, \citet{Barsse2026} choose $\ket\varphi$ to be $\ket0$,
\citet{Badescu2015} and the canonical semantics of \citet{Ying16-the-book} choose $\ket+$,
and \citet{Abbott2020}, as well as the generalised semantics of \citet{Ying16-the-book}, both allow an arbitrary choice of $\ket\varphi$.
We will discuss this phenomenon in detail later in \cref{rem:syntactic-dilation-for-murao-style}.

\subsection{A Transformation Without Dummy Values and the Correspondence Problem}\label{sec:overview:dummy-free}
We now propose another program transformation without dummy values, which gives a semantics for \( \keyword{qif} \) different from those discussed above.
According to our analysis, their semantics introduces an else-branch dummy value for each then-branch measurement and vice versa (\Cref{fig:intro:after-moving-individual-meas}).
The new transformation hoists a pair of measurements, one from each branch (\Cref{fig:intro:after-moving-matched-meas}).
This transformation applies only to quantum conditional branches in which the then- and else-branches contain the same number of measurements, but it requires no dummy qubit values.
\begin{figure}[t]
  \centering
  \begin{subfigure}{0.2\columnwidth}
    \begin{lstlisting}[language=qif]
      let (y, v) = qif x {
          ...;
          let _ = meas q₁;
          V₁
      } else {
          ...;
          let _ = meas q₂;
          V₂
      }
    \end{lstlisting}
    \caption{Source program.}
    \label{fig:intro:match-before-moving}
  \end{subfigure}
  \hspace{1cm}
  \begin{subfigure}{0.27\columnwidth}
    \begin{lstlisting}[language=qif]
      let (y, v, p₁, p₂) = qif x {
          ...;
          (V₁, q₁, |φ₂>)
      } else {
          ...;
          (V₂, |φ₁>, q₂)
      };
      let _ = meas p₁;
      let _ = meas p₂
    \end{lstlisting}
    \caption{Previous transformation.}
    \label{fig:intro:after-moving-individual-meas}
  \end{subfigure}
  \hspace{1cm}
  \begin{subfigure}{0.22\columnwidth}
    \begin{lstlisting}[language=qif]
      let (y, v, p) = qif x {
          ...;
          (V₁, q₁)
      } else {
          ...;
          (V₂, q₂)
      };
      let _ = meas p
    \end{lstlisting}
    \caption{New transformation.}
    \label{fig:intro:after-moving-matched-meas}
  \end{subfigure}
  \caption{Comparison between the previous and new transformations.
    The previous transformation introduces the else-branch dummy value \( \ket{\varphi_1} \) corresponding to \( \keyword{meas}\,q_1 \) in the then-branch and the then-branch dummy value \( \ket{\varphi_2} \) corresponding to \( \keyword{meas}\,q_2 \) in the else-branch.
    The new transformation moves a pair of measurements from the then- and else-branches, so it does not require any dummy qubit values.}
  \label{fig:intro:matched-moving}
  \Description{}
\end{figure}

However, there is an important and subtle issue: the program transformation depends on the particular pairing of the measurements in the then- and else-branches, and different choices of pairing give rise to different semantics.

We illustrate this point by means of an example (\cref{fig:correspondence-is-needed}).
The program in \cref{fig:correspondence:a} has two measurements in the then- and else-branches.
In this example, we use subscripts to distinguish between different occurrences of the measurement.
There are two possible ways to establish a correspondence between the measurements in the branches: one matching \( \keyword{meas}_1 \) with \( \keyword{meas}_3 \), and another matching \( \keyword{meas}_1 \) with \( \keyword{meas}_4 \). These choices yield the programs shown in \cref{fig:correspondence:b} and \cref{fig:correspondence:c}, respectively.
\begin{figure}[tp]
  \centering
  \begin{subfigure}{0.42\columnwidth}
    \begin{lstlisting}[language=qif]
      let (z, r) = qif x {
              let b₁ = meas₁,, y;
              let q = if b₁ { |0>,, } else { |1>,, };
              let b₂ = meas₂,, q;
              let r = if b₂ { |0>,, } else { |1>,, };
              r
      } else {
              let b₃ = meas₃,, y;
              let q = if b₃ { |0>,, } else { |1>,, };
              let b₄ = meas₄,, q;
              let r = if b₄ { |0>,, } else { |1>,, };
              r
      }
    \end{lstlisting}
    \caption{A program with the same number of measurements in the then- and else-branches.}
    \label{fig:correspondence:a}
  \end{subfigure}
  \qquad\
  \begin{subfigure}{0.36\columnwidth}
    \begin{lstlisting}[language=qif]
      let (z, r) = qif x {
              let (d₁, q) = qif y { |0>,, } else { |1>,, };
              let _ = meas₁,, d₁;
              let (d₂, r) = qif q { |0>,, } else { |1>,, };
              let _ = meas₂,, d₂;
              r
      } else {
              let (d₃, q) = qif y { |0>,, } else { |1>,, };
              let _ = meas₃,, d₃;
              let (d₄, r) = qif q { |0>,, } else { |1>,, };
              let _ = meas₄,, d₄;
              r
      }
    \end{lstlisting}
    \caption{A program obtained by locally deferring measurements.}
    \label{fig:correspondence:a2}
  \end{subfigure}
  \\[7pt]
  \begin{subfigure}{0.37\columnwidth}
    \begin{lstlisting}[language=qif]
      let (z, r, m₁, m₂) = qif x {
          let (d₁, q) = qif y { |0>,, } else { |1>,, };
          let (d₂, r) = qif q { |0>,, } else { |1>,, };
          (r, @\diffA{d$_1$, d$_2$}@)
      } else {
          let (d₃, q) = qif y { |0>,, } else { |1>,, };
          let (d₄, r) = qif q { |0>,, } else { |1>,, };
          (r, @\diffA{d$_3$, d$_4$}@)
      };
      let _ = meas m₁;   // $$meas$$₁ and $$meas$$₃
      let _ = meas m₂    // $$meas$$₂ and $$meas$$₄
    \end{lstlisting}
    \caption{The program obtained from \\
      \( \keyword{meas}_1 \leftrightsquigarrow \keyword{meas}_3 \) and \( \keyword{meas}_2 \leftrightsquigarrow \keyword{meas}_4 \).}
    \label{fig:correspondence:b}
  \end{subfigure}
  \qquad\
  \begin{subfigure}{0.37\columnwidth}
    \begin{lstlisting}[language=qif]
      let (z, r, m₁, m₂) = qif x {
          let (d₁, q) = qif y { |0>,, } else { |1>,, };
          let (d₂, r) = qif q { |0>,, } else { |1>,, };
          (r, @\diffA{d$_1$, d$_2$}@)
      } else {
          let (d₃, q) = qif y { |0>,, } else { |1>,, };
          let (d₄, r) = qif q { |0>,, } else { |1>,, };
          (r, @\diffB{d$_4$, d$_3$}@)
      };
      let _ = meas m₁;   // $$meas$$₁ and $$meas$$₄
      let _ = meas m₂    // $$meas$$₂ and $$meas$$₃
    \end{lstlisting}
    \caption{The program obtained from \\
      \( \keyword{meas}_1 \leftrightsquigarrow \keyword{meas}_4 \) and \( \keyword{meas}_2 \leftrightsquigarrow \keyword{meas}_3 \).}
    \label{fig:correspondence:c}
  \end{subfigure}
  \vspace{-2ex}
  \caption{Correspondence of measurements affects the translation.}
  \label{fig:correspondence-is-needed}
  \Description{}
\end{figure}

The resulting programs differ in their semantics.
A direct calculation shows that \cref{fig:correspondence:b} maps \( (x,y) = (\ket{+}, \alpha\ket{0} + \beta{\ket{1}}) \) to \( (z,r) = (\ket{+}, \ket{0}) \) with probability \( |\alpha|^2 \) and to \( (\ket{+}, \ket{1}) \) with probability \( |\beta|^2 \) (\ie, \cref{fig:correspondence:b} measures \( y \)). By contrast, \cref{fig:correspondence:c} maps the same input to \( (z,r) = \alpha \ket{00} + \beta \ket{11} \) or \( \alpha \ket{10} + \beta \ket{01} \), each with probability \( 1/2 \).
The difference between \cref{fig:correspondence:b} and \cref{fig:correspondence:c} is, in fact, exactly the subtlety that we have already examined, namely the distinction between \( \mathit{SWITCH}^?(x,y,D,D) \) and \( \mathit{SWITCH}(x,y,D,D) \) in \cref{sec:overview:switch}.
Notice that the mapping \( y \mapsto q \) on lines 2--3 in \cref{fig:correspondence:a} is \( D \) in \cref{sec:overview:switch}; the same is true of the mappings on lines 4--5, 8--9, and 10--11.
Matching \( D \) on lines 2--3 with lines 8--9 yields \( \mathit{SWITCH}^?(x,y,D,D) \), whereas matching \( D \) on lines 2--3 with lines 10--11 yields \( \mathit{SWITCH}(x,y,D,D) \).

Since the choice of matching affects the semantics, a canonical correspondence between the measurements in the then- and else-branches is necessary in order to obtain semantics free of arbitrariness.
We refer to this as the \emph{correspondence problem}.

\begin{remark}
  The correspondence problem for measurements discussed here can be seen as the operational counterpart of the correspondence problem for the indices of Kraus decompositions discussed in \cref{sec:overview:correspondence-problem-mathematically}.
  Indeed, once a Kraus decomposition is chosen, it can be implemented by a unitary dilation in which the discarded environment is measured in the computational basis; under such an implementation, the possible measurement outcomes are indexed by the Kraus indices.
  Thus, when a quantum operation is described by a program, these indices correspond to the possible histories of outcomes of the $\meas$ operation appearing in the program.

  From this perspective, the absence of a correspondence between indices in the semantics of a general quantum controlled operation $\qifx{x}{F(y)}{G(y)}$, discussed in \cref{sec:overview:past}, is the same phenomenon as the absence of a correspondence between measurements in the transformation discussed in \cref{sec:overview:transformation}.
  Conversely, the situation in \cref{sec:overview:correspondence-problem-mathematically}, where the indices in the then- and else-branches have a canonical correspondence, corresponds operationally to the situation in \cref{sec:overview:dummy-free}, where the $\meas$ operations themselves are put in one-to-one correspondence.
  \thmend
\end{remark}

\subsection{Designing a Well-Behaved Language}\label{sec:overview:language-design}
Recall the ``definition'' of \( \mathit{SWITCH}(x,y,F,G) \) in \cref{eq:into:switch} (adjusted to the syntax here):
\begin{quote}
  \centering
  \begin{qifnoframe}
    qif x { let z₁ = G(y); let r₁ = F(z₁); r₁ } else { let z₂ = F(y); let r₂ = G(z₂); r₂ }
  \end{qifnoframe}
\end{quote}
The measurements are hidden within the definitions of the quantum operations \( F \) and \( G \), so we do not know in advance how many measurements each contains.
However, there exists a canonical correspondence between the measurements in the then- and else-branches.
Letting \( n \) and \( m \) denote the number of measurements hidden within \( F \) and \( G \), respectively, we can see that the then-branch contains \( n + m \) measurements, and the else-branch contains \( m + n \), confirming that the counts match.
Furthermore, there exists a canonical correspondence: each measurement hidden in \( F \) in the then-branch corresponds to the ``same'' measurement in \( F \) in the else-branch.
This is our explanation for why the semantics of the quantum SWITCH is well-defined.
Generalising this argument, the key to ensuring the well-definedness of the semantics is \emph{linearity} in the sense of linear logic~\cite{Girard1987}:
each quantum operation must appear exactly once in both branches.

This observation leads to the following language design.
We divide the language into two sublanguages: a \emph{classical sublanguage} in which measurements are available, and a \emph{quantum sublanguage} in which quantum conditional branching is available.
In the classical sublanguage, quantum branching is not directly permitted; in the quantum sublanguage, quantum measurement is not directly permitted.
However, through the following form of interoperability, each sublanguage can indirectly make use of them.
First, every program \( P \) in the quantum sublanguage can be used directly within the classical sublanguage.
To make the boundary between the quantum and classical components explicit, we shade its embedding in grey, as in \( \qpart{P} \); shaded parts belong to the quantum sublanguage.
There is no embedding in the opposite direction.
Nevertheless, we allow a variable defined in the classical sublanguage to be used within the quantum sublanguage.
That is, an expression of the form \( \Let x = (\cdots \keyword{meas} \cdots); \qpart{\cdots x \cdots} \) is permitted, enabling us to handle measurements indirectly within the quantum sublanguage.
For the reasons discussed above, every variable appearing in the quantum sublanguage must be used linearly.
This linearity constraint also applies to variables defined in the classical sublanguage but used inside the quantum sublanguage.

Our design resolves the mystery of the incorrect equation~\eqref{eq:switch-vs-switchq}: the middle term is prohibited because it violates the linearity constraint in the quantum sublanguage.
To describe \( \mathit{SWITCH}(x,y,D,D) \) in our language, we should first bind \( D \) to variables \( f \) and \( g \) in the classical sublanguage and use them in the quantum sublanguage, as in the third expression below:
\begin{align*}
  \Let f = D;\; \Let g = D;\; \qpart{\qif\, x \; \{\, \diffA{f\,(g\,y)} \,\} \mbox{ $\keyword{else}$ } \{\, \diffA{f\,(g\,y)}\,\}}
   & \quad=\quad
  \mathit{SWITCH}^{?}(x,y,D,D)
  \\
  \hphantom{\Let f = D;\; \Let g = D;\;} \qpart{\qif\, x \; \{\, D\,(D\,y) \,\} \mbox{ $\keyword{else}$ } \{\, D\,(D\,y)\,\}}
   & \quad
  \mbox{ violates the linearity constraint}
  \\
  \Let f = D;\; \Let g = D;\; \qpart{\qif\, x \; \{\, \diffA{f\,(g\,y)} \,\} \mbox{ $\keyword{else}$ } \{\, \diffB{g\,(f\,y)}\,\}}
   & \quad=\quad
  \mathit{SWITCH}(x,y,D,D)
\end{align*}
\( \mathit{SWITCH}^?(x,y,D,D) \) is similar, but it uses variables \( f \) and \( g \) differently in the quantum sublanguage.
In this way, our language appropriately distinguishes \( \mathit{SWITCH}(x,y,D,D) \) from \( \mathit{SWITCH}^?(x,y,D,D) \).

The final challenge is to prove the well-definedness of the semantics and the validity of \eqref{eq:into:switch} in our language.
We prove them using a categorical semantics in \cref{sec:cat-sem}.

\section{A Quantum Procedural Language with Well-Behaved Semantics}\label{sec:proc}

\subsection{Language}
This subsection defines the syntax of the language \( \Qif \).
It is a functional language with a linear type system, and its design follows the principles outlined in \cref{sec:overview:language-design}.

The syntax of \( \Qif \) appears in \cref{fig:procedural:syntax}.

\begin{figure}
  \begin{align*}
    \begin{array}{llrl}
      \textit{Types}
       & A,B          & \Coloneqq &
      \qbit \sor \bool \sor \unit \sor A \otimes B \sor A \li B
      \\[5pt]
      \textit{Terms}
       & \term,\termB & \Coloneqq &
      x \sor \unitval \sor \term; \termB \sor \lambda x. \term \sor \term\,\termB
      \sor \term \otimes \termB
      \\ &       & \sor &
      \letx{x}{\term}{\termB} \sor \letx{x\otimes y}{\term}{\termB}
      \sor \letx{x}{\F}{\term}
      \\ &       & \sor &
      \qifx{\term}{\termB_1}{\termB_2} \sor \cifx{\term}{\termB_1}{\termB_2}
      \\ &       & \sor &
      \ket{1} \sor U \sor \meas \sor \btrue \sor \bnot
      \\[5pt]
      \textit{Global Definitions}
       & \Defs        & \Coloneqq &
      \cdot \mid \Defs, \globaldef \F = \term
    \end{array}
  \end{align*}
  \vspace{-10pt}
  \caption{Syntax of \Qif{}.
    Here, $U$ is an $n$-qubit unitary operator,
    and $\F$ is a name for a globally defined term.}
  \label{fig:procedural:syntax}
  \Description{}
\end{figure}

The language provides three base types:
the qubit type $\qbit$, the boolean type $\bool$, and the unit type $\unit$.
It also includes two type constructors,
namely the tensor type $A \otimes B$
and the (linear) function type $A \li B$.
A type \(A\) is \emph{first order}, written \( \FirstOrderType{A} \), when it has no function type \(\li\).
A type is \emph{purely quantum} if it does not involve \(\bool\), and we
write $\Quantumness{A}$ to indicate that $A$ is purely quantum.
The boolean type \(\bool\) is not available in the quantum sublanguage.

Most term constructors are from the standard linear lambda calculus or from quantum \( \lambda \)-calculi as in \citet{Selinger2008}.
The first and second lines are standard constructs in the linear lambda calculus:
        variables $x$, the unit value $\unitval$,
        sequential composition $\term;\termB$,
        lambda abstraction $\lambda x. \term$,
        application $\term\,\termB$,
        the tensor product $\term\otimes\termB$,
        the standard binding $\letx{x}{\term}{\termB}$, and
        the tensor-destructing binding $\letx{x\otimes y}{\term}{\termB}$.
An exception is the binding $\letx{x}{\F}{\term}$, which assigns a globally defined constant $\F$ to a local variable $x$.
We distinguish between a variable \( x \) and a globally defined constant \(\F\) as the latter is duplicable and discardable.\footnote{A constant can be seen as a value of type \( {!}A \).
  The distinction between variables and globally defined constants can be avoided by introducing the \( ! \)-type into the classical sublanguage.
  In this paper, however, we do not permit general \(!\)-types, simplifying the exposition by restricting duplication to globally defined constants.}
The third line contains two conditionals, \(\keyword{qif}\) and \(\keyword{if}\).
The final line enumerates the primitive constants:
        qubit initialisation $\ket1$,
        all unitary operators $U$,
        measurement $\meas$,
        and the boolean primitives $\btrue$ and $\bnot$.

A \emph{typing context} \( \Delta \) (resp.~\(\Gamma\)) is a finite sequence of type bindings of the form \( x \colon A \) (resp.~\( \F \colon A \)).

\begin{mathfig}

\begin{proofrules}[5pt]
  \infer{
    \termJQ{\Delta}{\term \colon A}
  }{
    \termJC{\Delta}{\term \colon A}
  }

  \infer{
    \Quantumness{A}
  }{
    \termJQ{x \colon A}{x \colon A}
  }

  \infer{}{
    \termJC{x \colon A}{x \colon A}
  }

  \infer{
    \termJG{\Delta, x \colon A}{\term \colon B}
  }{
    \termJG{\Delta}{\lambda x. \term \colon A \li B}
  }

  \infer{
    \termJG{\Delta}{\term \colon \unit}
    \\
    \termJG{\Delta'}{\termB \colon A}
  }{
    \termJG{\Delta, \Delta'}{(\term; \termB) \colon A}
  }

  \infer{
    \termJG{\Delta}{\term \colon A \li B}
    \\
    \termJG{\Delta'}{\termB \colon A}
  }{
    \termJG{\Delta, \Delta'}{\term \, \termB \colon B}
  }

  \infer{
    \termJG{\Delta}{\term \colon A}
    \\
    \termJG{\Delta'}{\termB \colon B}
  }{
    \termJG{\Delta, \Delta'}{\term \otimes \termB \colon A \otimes B}
  }

  \infer{
    \termJG{\Delta}{\term \colon A}
    \\
    \termJG{\Delta', x \colon A}{\termB \colon B}
  }{
    \termJG{\Delta, \Delta'}{\letx{x}{\term}{\termB} \colon B}
  }

  \infer{
    \termJG{\Delta}{\term \colon A \otimes B}
    \\
    \termJG{\Delta', x \colon A, y \colon B}{\termB \colon C}
  }{
    \termJG{\Delta, \Delta'}{\letx{x \otimes y}{\term}{\termB} \colon C}
  }

  \infer{
    U \colon \CC^{2^n} \longrightarrow \CC^{2^n}, \mbox{unitary}
  }{
    \termJQ{\EmpEnv}{U \colon \qbit^{\otimes n} \li \qbit^{\otimes n}}
  }

  \infer{} %
  {
    \termJQ{\EmpEnv}{\unitval \colon \unit}
  }

  \infer{} %
  {
    \termJC{\EmpEnv}{\ket1 \colon \qbit}
  }

  \infer{} %
  {
    \termJC{\EmpEnv}{\btrue \colon \bool}
  }

  \infer{} %
  {
    \termJC{\EmpEnv}{\meas \colon \qbit \li \bool}
  }

  \infer{} %
  {
    \termJC{\EmpEnv}{\bnot \colon \bool \li \bool}
  }

  \infer{
    \termJQ{\Delta}{\term \colon \qbit}
    \\
    \termJQ{\Delta'}{\termB_1 \colon A}
    \\
    \termJQ{\Delta'}{\termB_2 \colon A}
    \\
    \FirstOrderType{A}
  }{
    \termJQ{\Delta, \Delta'}{\qifx{\term}{\termB_1}{\termB_2} \colon \qbit \otimes A}
  }

  \infer{
    \termJC{\Delta}{\term \colon \bool}
    \\
    \termJC{\Delta'}{\termB_1 \colon A}
    \\
    \termJC{\Delta'}{\termB_2 \colon A}
    \\
    \FirstOrderType{A}
  }{
    \termJC{\Delta, \Delta'}{\cifx{\term}{\termB_1}{\termB_2} \colon A}
  }

  \infer{
    (\F \colon A) \in \Gamma
    \\
    \termJC{\Delta, x \colon A}{\term \colon B}
  }{
    \termJC{\Delta}{\letx{x}{\F}{\term} \colon B}
  }

  \infer{
    \vdash \Defs \colon \Gamma
    \qquad
    \termJC{\EmpEnv}{\term \colon A}
  }{
    \vdash (\Defs, \globaldef \F = \term) \colon (\Gamma, \F \colon A)
  }
\end{proofrules}   \caption{Typing rules for \Qif{}.\
    The rules labelled $\vdash_\JG$ are shared between $\vdash_\JQ$ and $\vdash_\JC$.
    We omit the exchange rule, which permutes the order of elements in \(\Delta\).}
  \label{fig:procedural:typing-rules}
  \Description{}
\end{mathfig}
\Qif{} features a linear type system with two derivation modes,
$\vdash_\JQ$ and $\vdash_\JC$.
The judgement $\termJQ{\Delta}{\term\colon A}$ is for the quantum sublanguage.
Both the context $\Delta$ and the type $A$ must be purely quantum,
preventing copying or discarding without any exception.
The quantum conditional branching $\qif$ is available only in this fragment.
The boolean type \(\bool\), the classical conditional branching \(\cif\) and the measurement \(\meas\) are not available.
The judgement $\termJC{\Delta}{\term\colon A}$ is for the classical sublanguage.
It augments the linear context \(\Delta\) with a \emph{non-linear} context $\Gamma$
for globally defined constants that may be duplicated.
Booleans, classical conditional branching and measurement are available in this fragment, but \(\qif\) is not.

The typing rules appear in \cref{fig:procedural:typing-rules}.
Most rules are inherited from the standard linear lambda calculus or the quantum lambda calculus; some constructs are available in both judgements, whereas others are available only in \( \JQ \) or in \( \JC \), following the policy described above.

The most important aspect is the mechanism enabling interoperability between \(\JQ\) and \(\JC\).
The \emph{embedding rule} at the top left derives a $\CPM$-judgement from a $\Hilb$-judgement,
bridging the two derivation modes.
In the language of \cref{sec:overview:language-design}, embeddings were marked explicitly by grey shading, whereas in \(\Qif\) we choose to treat embeddings implicitly.
Conversely, the way in which \(\JC\) constructs (such as quantum measurement) become available inside \(\JQ\) is subtly hidden and thus requires attention.
The underlying trick is that the type environment \(\Delta\) does not record whether a variable \( x \) has been bound in \(\JC\) or in \(\JQ\).
Consequently, variables defined in \(\JC\) can be used in \(\JQ\), which indirectly enables the use of quantum measurement and other \(\JC\) constructs within \(\JQ\).
For example, this allows us to write a term of the form
\begin{equation*}
  \letx{d}{\bigl(\lambda q. \cifx{\meas\,q}{\keyword{X}\ket{1}}{\ket{1}}\bigr)}{\bigl(\qifx{x}{(\cdots d \cdots)}{(\cdots d \cdots)}\bigr)}.
\end{equation*}

We discuss other aspects.
Variable rules in $\vdash_\JQ$ and $\vdash_\JC$ are distinct for restricting global definitions and booleans from the $\JQ$ fragment.
Each conditional is confined to its respective fragment:
$\qif$ is available only in $\JQ$, where strict linearity is enforced,
while $\cif$ is available only in $\JC$.
The result types also differ:
$\qif$ returns the conditional qubit (as a qubit is not discardable),
whereas $\cif$ does not return the controlling boolean (following the convention).
We forbid function types to occur in the conditional results for the technical reason
discussed by \citet{HirataT26}.\footnote{
  They showed that
  the program $\letx{y\otimes f}{(\qifx{x}{\gX}{\gI})}{f y}$ is not physically realisable.
}
Global constants $\F$ in the non-linear context $\Gamma$ can be referenced only in $\CPM$.
We can bind $\F$ to a local variable $x$,
and once bound, $x$ becomes linear even though $\F$ remains reusable.
Any closed term $\term$ can be promoted to a global definition.

\begin{remark}
  The judgements \( \termJC{}{\unitval\colon \unit} \) and \( \termJC{}{U\colon \qbit^{\otimes n} \multimap \qbit^{\otimes n}} \) are derivable via the embedding rule, so the unit value and unitaries are also available in \( \JC \).
  Although the boolean values must be used linearly in \( \Qif \), they are duplicable and discardable through \(\cif\).
  We can define \( \bdiscard \colon \bool \multimap \unit \) and \( \bcopy \colon \bool \multimap \bool \otimes \bool \) by
  $
    \bdiscard \defeq \lambda x. \cifx{x}{\unitval}{\unitval}
  $ and $
    \bcopy \defeq \lambda x. \cifx{x}{(\btrue \otimes \btrue)}{(\bfalse \otimes \bfalse)}
  $, where \( \bfalse \defeq \bnot\,\btrue \).
  \thmend
\end{remark}

\begin{example}\label{eg:procedural:switch}
  Let \( \mathit{SWITCH}(x,y,f,g) \) be the term given by
  \begin{equation*}
    \SWITCH(x,y,f,g)
    \quad=\quad
    \qifx{x}{f(g(y))}{g(f(y))}.
  \end{equation*}
  Then \( x \colon \qbit, y \colon A, f \colon A \multimap A, g \colon A \multimap A \vdash_{\Hilb} \SWITCH(x,y,f,g) \colon \qbit \otimes A \)
  for any $A$ satisfying $\Quantumness{A}$ and $\FirstOrderType{A}$.
  Note that the function inputs \( f \) and \( g \) of \( \mathit{SWITCH}(x,y,f,g) \) can be \( \JC \)-terms, which may involve measurement, by promoting the above judgement to the corresponding \( \vdash_{\JC} \)-judgement by the embedding rule as follows:
  \begin{align*}
    \letx{f}{(\lambda z.\bdiscard(\meas(z)); \ket1)}{\,\letx{g}{(\lambda z.\,z)}{\,\SWITCH(x,y,f,g)}}.
    \tag*{\thmend}
  \end{align*}
\end{example}

\begin{example}
  A SWITCH-like operation in \citet{Liu2024} can be written as follows:
  \begin{align*}
    &
    M
    \quad=\quad
    \qifx{x}{f_1(f_2(\dots(f_n(y))\dots))}{f_n(f_{n-1}(\dots(f_1(y))\dots))}
    \\
    &\qquad
    \big[\,
      \termJQ{
        x \colon \qbit, y \colon A, f_1 \colon A \multimap A,\dots, f_n \colon A \multimap A
      }{
        M \colon \qbit\otimes A
      }
    \,\big].
  \end{align*}
  For another example, let \( M' \) be the term given by
  \begin{align*}
    &
    M'
    \quad=\quad
    \qifx{x}{
      \letx{a \otimes b}{f y}{g (ka \otimes hb)}
    }{
      h( k (g (f y)))
    }
    \\
    &\qquad
    \big[\,
    \termJQ{
      x \colon \qbit,
      y \colon A,
      f \colon A \li A \otimes A,
      g \colon A \otimes A \li A,
      k, h \colon A \li A
    }{
      M' \colon \qbit\otimes A
    }
    \,\big].
  \end{align*}
  This term represents the superposition of the following two circuits:
  
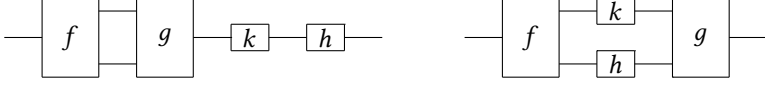
\begin{align*}
\begin{tikzpicture}[yscale=0.7]
	\begin{pgfonlayer}{nodelayer}
		\node [style=none] (0) at (-2, 0.75) {};
		\node [style=none] (1) at (-2, -0.75) {};
		\node [style=none] (2) at (-1.25, 0.75) {};
		\node [style=none] (3) at (-1.25, -0.75) {};
		\node [style=none] (4) at (-0.75, 0.75) {};
		\node [style=none] (5) at (-0.75, -0.75) {};
		\node [style=none] (6) at (0, 0.75) {};
		\node [style=none] (7) at (0, -0.75) {};
		\node [style=none] (8) at (0.5, 0.25) {};
		\node [style=none] (9) at (1, 0.25) {};
		\node [style=none] (10) at (1, -0.25) {};
		\node [style=none] (11) at (0.5, -0.25) {};
		\node [style=none] (12) at (1.5, 0.25) {};
		\node [style=none] (13) at (2, 0.25) {};
		\node [style=none] (14) at (2, -0.25) {};
		\node [style=none] (15) at (1.5, -0.25) {};
		\node [style=none] (16) at (-2.5, 0) {};
		\node [style=none] (17) at (-2, 0) {};
		\node [style=none] (18) at (-1.25, 0.5) {};
		\node [style=none] (19) at (-0.75, 0.5) {};
		\node [style=none] (20) at (-1.25, -0.5) {};
		\node [style=none] (21) at (-0.75, -0.5) {};
		\node [style=none] (22) at (0, 0) {};
		\node [style=none] (23) at (0.5, 0) {};
		\node [style=none] (24) at (1, 0) {};
		\node [style=none] (25) at (1.5, 0) {};
		\node [style=none] (26) at (2, 0) {};
		\node [style=none] (27) at (2.5, 0) {};
		\node [style=none] (28) at (0.75, 0) {$k$};
		\node [style=none] (29) at (1.75, 0) {$h$};
		\node [style=none] (30) at (-1.625, 0) {$f$};
		\node [style=none] (31) at (-0.375, 0) {$g$};
	\end{pgfonlayer}
	\begin{pgfonlayer}{edgelayer}
		\draw (0.center) to (2.center);
		\draw (2.center) to (3.center);
		\draw (3.center) to (1.center);
		\draw (1.center) to (0.center);
		\draw (4.center) to (6.center);
		\draw (6.center) to (7.center);
		\draw (7.center) to (5.center);
		\draw (5.center) to (4.center);
		\draw (8.center) to (9.center);
		\draw (9.center) to (10.center);
		\draw (10.center) to (11.center);
		\draw (11.center) to (8.center);
		\draw (12.center) to (13.center);
		\draw (13.center) to (14.center);
		\draw (13.center) to (13.center);
		\draw (14.center) to (15.center);
		\draw (15.center) to (12.center);
		\draw (16.center) to (17.center);
		\draw (18.center) to (19.center);
		\draw (20.center) to (21.center);
		\draw (22.center) to (23.center);
		\draw (24.center) to (25.center);
		\draw (26.center) to (27.center);
		\draw (23.center) to (23.center);
	\end{pgfonlayer}
\end{tikzpicture}
     \qquad\quad
\begin{tikzpicture}[yscale=0.7]
	\begin{pgfonlayer}{nodelayer}
		\node [style=none] (0) at (-2, 0.75) {};
		\node [style=none] (1) at (-1.25, 0.75) {};
		\node [style=none] (2) at (-1.25, -0.75) {};
		\node [style=none] (3) at (-2, -0.75) {};
		\node [style=none] (4) at (-0.75, 0.75) {};
		\node [style=none] (5) at (-0.25, 0.75) {};
		\node [style=none] (6) at (-0.25, 0.25) {};
		\node [style=none] (7) at (-0.75, 0.25) {};
		\node [style=none] (8) at (-0.75, -0.25) {};
		\node [style=none] (9) at (-0.25, -0.25) {};
		\node [style=none] (10) at (-0.25, -0.75) {};
		\node [style=none] (11) at (-0.75, -0.75) {};
		\node [style=none] (12) at (0.25, 0.75) {};
		\node [style=none] (13) at (1, 0.75) {};
		\node [style=none] (14) at (1, -0.75) {};
		\node [style=none] (15) at (0.25, -0.75) {};
		\node [style=none] (16) at (-2.5, 0) {};
		\node [style=none] (17) at (-2, 0) {};
		\node [style=none] (18) at (-1.25, 0.5) {};
		\node [style=none] (19) at (-0.75, 0.5) {};
		\node [style=none] (20) at (-0.25, 0.5) {};
		\node [style=none] (21) at (0.25, 0.5) {};
		\node [style=none] (22) at (-1.25, -0.5) {};
		\node [style=none] (23) at (-0.75, -0.5) {};
		\node [style=none] (24) at (-0.25, -0.5) {};
		\node [style=none] (25) at (0.25, -0.5) {};
		\node [style=none] (26) at (1, 0) {};
		\node [style=none] (27) at (1.5, 0) {};
		\node [style=none] (28) at (-1.625, 0) {$f$};
		\node [style=none] (29) at (-0.5, 0.5) {$k$};
		\node [style=none] (30) at (-0.5, -0.5) {$h$};
		\node [style=none] (31) at (0.625, 0) {$g$};
	\end{pgfonlayer}
	\begin{pgfonlayer}{edgelayer}
		\draw (0.center) to (1.center);
		\draw (1.center) to (2.center);
		\draw (2.center) to (3.center);
		\draw (3.center) to (0.center);
		\draw (4.center) to (5.center);
		\draw (5.center) to (6.center);
		\draw (6.center) to (7.center);
		\draw (7.center) to (4.center);
		\draw (8.center) to (9.center);
		\draw (9.center) to (10.center);
		\draw (10.center) to (11.center);
		\draw (11.center) to (8.center);
		\draw (12.center) to (13.center);
		\draw (13.center) to (14.center);
		\draw (14.center) to (15.center);
		\draw (15.center) to (12.center);
		\draw (16.center) to (17.center);
		\draw (18.center) to (19.center);
		\draw (20.center) to (21.center);
		\draw (8.center) to (8.center);
		\draw (22.center) to (23.center);
		\draw (15.center) to (15.center);
		\draw (9.center) to (9.center);
		\draw (24.center) to (25.center);
		\draw (26.center) to (27.center);
	\end{pgfonlayer}
\end{tikzpicture}
     .
  \end{align*}
  Although various SWITCH-like operators have been proposed~\cite{Procopio2019,Procopio2020,Das2022,Liu2024}, all of them focus on controlling the order of multiple operations of the same type \( A \multimap A \).
  A SWITCH-like operator involving tensor products, such as the example above, has not been explored, to the best of our knowledge.
  \thmend
\end{example}

\begin{example}
  Also, $\qif$ can be nested.
  Let $f, g$
  be variables with the function type $\qbit \li \qbit$.
  In the term
  $
    \qifx{x}
    {\SWITCH(y, z, f, g)}
    {f y \otimes g z}
  $,
  the qubit $y$
  is used as a control in the then-branch
  and as an input of $f$ in the else-branch.
  Therefore, this function cannot be represented
  in a language that separates
  ``control'' qubits and ``target'' qubits
  as in \citet{ClmPerd20-LIPIcs-PBS}.
  \thmend
\end{example}

\section{Categorical Semantics}\label{sec:cat-sem}
In this section, we define a categorical semantics of \Qif{} and see how the correspondence problem is resolved in \Qif{}.

\subsection{Preliminaries: the Categories of Linear Maps and Completely Positive Maps}
We briefly review the definitions and properties of two specific compact closed categories, $\Hilb$ and $\CPM$:
for details, see, e.g., \cite{Selinger2004,Selinger2004a,Selinger2008}.

The category $\Hilb$ is the category of finite-dimensional $\CC$-vector spaces and linear maps.
More concretely, the objects of the category are natural numbers,
each object $n$ representing the $n$-dimensional $\CC$-vector space $\CC^n$
spanned by $\ket{0},\dots,\ket{n-1}$,
and the morphisms are $\CC$-linear maps $f \colon \CC^n \longrightarrow \CC^m$.
The category $\Hilb$ has a standard monoidal product
corresponding to the tensor product of the vector spaces,
\ie, $n \otimes m \defeq n \times m$.
The basis of $n \otimes m$ is given by the product basis
$\ket{00},\ket{01},\dots,\ket{0(m-1)},\ket{10},\dots,\ket{(n-1)(m-1)}$.
The tensor unit is $1$.
The dual object \(n^*\) of $n$ is given by the dual space, which is again $n$.
The unit and counit morphisms are defined by
\begin{equation*}
  \textstyle
  \eta_n \colon 1 \longrightarrow n \otimes n^*;\
  \alpha \longmapsto {\alpha} \sum_{i = 0}^{n-1} \ket{i} \otimes \ket{i},
  \qquad
  \varepsilon_n \colon n^* \otimes n \longrightarrow 1;\
  \sum_{i,j = 0}^{n-1} \alpha_{ij} \ket{i} \otimes \ket{j}
  \longmapsto \sum_{i = 0}^{n-1} \alpha_{ii}.
\end{equation*}
This category $\Hilb$ also admits biproducts defined by $n \oplus_{\Hilb} m \defeq n + m$.
In particular, $1 \oplus_{\Hilb} 1 = 2$ corresponds to $\CC^2$,
in which we interpret the qubit type.

The other category $\CPM$ of completely positive maps is defined as follows.
An object of $\CPM$ is a finite sequence $\vec{n} = (n_0,\dots,n_{k-1})$ of natural numbers.
A morphism $f \in \CPM(\vec{n}, \vec{m})$ is a matrix of completely positive maps
$(f_{ij} \colon \Mat_{n_i}(\CC) \longrightarrow \Mat_{m_j}(\CC))$.
The monoidal product in $\CPM$ is defined by
$
  (n_0,\dots,n_{k-1}) \otimes (m_0,\dots,m_{\ell-1}) := (n_0m_0, \dots,n_0m_{\ell-1}, n_1m_0,\dots, n_{k-1}m_{\ell-1})
$.
The monoidal unit is $(1)$, and the dual \(\vec{n}^*\) of $\vec{n}$ is again $\vec{n}$.
The unit and counit morphisms are defined by
\begin{align*}
  (\eta_{\vec n})_{1,ij}        &
  = \begin{cases}
      1 \longmapsto \sum_{k,\ell = 0}^{ n_i - 1 }
      \ketbra{k}{\ell} \otimes \ketbra{k}{\ell}
                      & \text{if } i = j
      \\
      1 \longmapsto 0 & \text{otherwise},
    \end{cases}
  \\
  (\varepsilon_{\vec n})_{ij,1} &
  = \begin{cases}
      \ketbra{k}{\ell} \otimes \ketbra{k'}{\ell'}
      \longmapsto
      \delta_{k,k'} \delta_{\ell,\ell'}
        & \text{if } i = j  \\
      \ketbra{k}{\ell} \otimes \ketbra{k'}{\ell'}
      \longmapsto
      0 & \text{otherwise}.
    \end{cases}
\end{align*}
This category $\CPM$ also admits biproducts defined by the concatenation of vectors
$\vec{n} \oplus_{\CPM} \vec{m} = \vec{n}\vec{m}$.
The biproduct of the monoidal unit $(1) \oplus_{\CPM} (1) = (1,1)$
will be the semantics of booleans.

There are several important maps in this category.
\begin{itemize}
  \item %
        For each object $\vec{n}$ in $\CPM$,
        there is a \emph{discarding map} $\discardDiag_{\vec{n}} \colon \vec{n} \longrightarrow (1)$
        defined by the componentwise trace operator
        $(\discardDiag_{\vec{n}})_{\ell,1}
          \colon \sum_{ij=0}^{n_\ell-1} \alpha_{ij} \ketbra{i}{j} \longmapsto \sum_{i = 0}^{n_\ell-1}\alpha_{ii}$.
  \item %
        Each object $\vec{n} = (n_0,\dots,n_{k-1})$ in $\CPM$ is a retraction of $(N)$ where $N \defeq \sum_{\ell=0}^{k-1} n_\ell$.
        The \emph{section} $s_{\vec n} \colon \vec{n} \longrightarrow (N)$ is the componentwise embedding of $n_\ell \times n_\ell$ matrices
        into $N \times N$ matrices defined by $\ketbra{i}{j} \mapsto \ketbra{N_\ell + i}{N_\ell + j}$
        where $N_\ell \defeq n_0 + \cdots + n_{\ell - 1}$.
        The corresponding \emph{retraction} map
        $r_{\vec{n}} \colon (N) \longrightarrow \vec{n}$ is the projection
        to the corresponding block of the matrix.
        For the case of $\vec{n} = (1,1)$,
        we define $m \defeq r_{(1,1)}$ and call it the \emph{measurement map}.
\end{itemize}

There is a canonical functor in one direction
$\EmbeddingFunctor \colon \Hilb \longrightarrow \CPM$, which we call the \emph{embedding}.
This functor is defined by $\EmbeddingFunctor(n) \defeq (n)$
and $\EmbeddingFunctor(f)(\rho) \defeq f \rho f^\dagger$.
This functor preserves the compact closed structure on the nose.
Equally important is that the biproduct is not preserved by this embedding.

\subsection{Hilb-CPM Model}
Our model of \Qif{} consists of two categories: \(\Hilb\) for modelling the quantum sublanguage, and \( \CPM \) for the classical sublanguage.
We define $\sem{-}_\Hilb$ and $\sem{-}_\CPM$ in the categories $\Hilb$ and $\CPM$, respectively, which are often identified through the embedding functor.

The categorical semantics of types $A$ is given by $\sem{A}$ as follows.
\begin{gather*}
  \sem{\unit} = 1,
  \ \ \
  \sem{\qbit} = 2,
  \ \ \
  \sem{A\otimes B} = \sem{A} \otimes \sem{B},
  \ \ \
  {\sem{A\multimap B} = \sem{A}^* \otimes \sem{B}},
  \ \ \
  \sem{\bool} = (1, 1) \in \CPM.
\end{gather*}
Since booleans are considered only in $\CPM$-context,
the semantics is defined only in $\CPM$.

We define the categorical semantics of $\termJQ{\Delta}{\term\colon A}$
as a map $\sem{\term}_\Hilb \colon \sem{\Delta}_\Hilb \longrightarrow \sem{A}_\Hilb$, where $\sem{\Delta}$ is the tensor product of the semantics of types in $\Delta$,
following the convention of linear lambda calculus.
For $\CPM$-terms, since they may have a non-linear context $\Gamma$,
we define the categorical semantics of a derivation
$\termJC{\Delta}{\term\colon A}$ parameterised by a valuation $\eval \in \prod_{(\F \colon B) \in \Gamma} \sem{B}_\CPM$ of the non-linear context,
\ie, a map $\sem{\term}_{\CPM,\eval} \colon \sem{\Delta}_\CPM \longrightarrow \sem{A}_\CPM$.

\begin{mathfig}
  \begin{gather*}
    \sem{\termJQ{\Delta}{\term \colon A}}_\JC
    = \EmbeddingFunctor \sem{\termJQ{\Delta}{\term \colon A}}_\JQ
    \hspace*{1.8em}
    \sem{x\colon A} = \ident_{\sem{A}}
    \hspace*{1.8em}
    \sem{\term;\termB} = \sem{\term}\otimes\sem{\termB}
    \\[-0.1em]
    \sem{\lambda x^A. \term} = \Lambda_A \sem{\term}
    \hspace*{2.5em}
    \sem{\term\,\termB} = \ev (\sem{\term} \otimes \sem{\termB})
    \hspace*{2.5em}
    \sem{\term \otimes \termB} = \sem{\term} \otimes \sem{\termB}
    \\[-0.1em]
    \sem{\letx{x}{\term}{\termB}} = \sem{\termB} \sigma_{A,\Delta'} (\sem{\term} \otimes \ident_{\Delta'})
    \hspace*{2.5em}
    \sem{\letx{x\otimes y}{\term}{\termB}} = \sem{\termB} \sigma_{AB,\Delta'} (\sem{\term} \otimes \ident_{\Delta'})
    \\[-0.1em]
    \sem{\unitval} = \ident_{1}
    \hspace*{2.em}
    \sem{U} = U
    \hspace*{2.em}
    \sem{\ket1}_\CPM = \ketbra{1}{1}
    \hspace*{2.em}
    \sem{\meas}_\CPM = m
    \hspace*{2.em}
    \sem{\btrue}_\CPM = \mathrm{inr}_{I,I}
    \\[-0.1em]
    \sem{\bnot} = \mathrm{swap}_{\oplus_\CPM}
    \hspace*{2.em}
    \sem{\qifx{\term}{\termB_1}{\termB_2}}_\Hilb
    = d_{\oplus_\Hilb}^{-1} (\sem{\termB_2} \oplus_\JQ \sem{\termB_1}) d_{\oplus_\Hilb} (\sem{\term}\otimes \ident_{\Delta'})
    \\[-0.1em]
    \sem{\cifx{\term}{\termB_1}{\termB_2}}_\CPM
    = (\discardDiag_{(1,1)}\otimes\ident_A)d_{\oplus_\CPM}^{-1}(\sem{\termB_2} \oplus_\JC \sem{\termB_1})d_{\oplus_\CPM}(\sem{\term}\otimes \ident_{\Delta'})
    \\[-0.1em]
    \sem{\letx{f}{\F}{\term}}_{\CPM,\varrho}
    = \sem{\term}_{\CPM,\varrho} \circ (\ident_{\Delta} \otimes \varrho(\F))
    \hspace*{1.5em}
    \sem{\Defs, \globaldef \F = \term} = \sem{\Defs} \cup (\F \mapsto \sem{\term})
  \end{gather*}
  \vspace{-3ex}
  \caption{Categorical semantics for \Qif{}.
    Here,
    $\ev \colon (A \li B) \otimes A \rightarrow B$ is the evaluation map,
    $\sigma_{A,B} \colon A \otimes B \cong B \otimes A$ is the braiding,
    $\mathrm{swap} \colon (1,1) \rightarrow (1,1)$ is the swapping morphism,
    $d_{\oplus_\JG}\colon A \otimes (1\oplus_\JG 1) \cong A\oplus_{\JG} A$ is the distribution isomorphism,
    $m\colon (2)\rightarrow (1,1)$ is the measurement map,
    ${\protect\discardDiag}\colon I\oplus I \rightarrow I$ is the discard map,
    and
    $\Lambda_Y(-)$ maps a morphism \( f \colon X \otimes Y \to Z \) to a morphism
    $X \to Y\li Z$.}
  \label{fig:procedural:categorical-semantics}
  \Description{}
\end{mathfig}
The semantics of derivations is defined in \cref{fig:procedural:categorical-semantics}.
The rule that promotes a $\Hilb$-term into a $\CPM$-term is interpreted via
the embedding functor $\EmbeddingFunctor$.
Note that $\EmbeddingFunctor\bigl(\sem{\Delta}_\Hilb\bigr) = \sem{\Delta}_\CPM$
uses the fact that $\EmbeddingFunctor$ preserves the compact closed structures on the nose.
The two conditional terms $\qif$ and $\cif$ utilise the biproducts in each category.
It is important that these terms are interpreted in different categories
as the biproducts differ in $\Hilb$ and $\CPM$.
The semantics of $\letx{x}{\F}{\term}$ in $\CPM$ is the only rule that
accesses $\eval$.

\begin{remark}
  A judgement for a $\CPM$-term may have several different derivations because
  the place to use the promotion rule from $\Hilb$-term is not unique.
  However, we can prove the uniqueness of the categorical semantics for the terms.
  The proof is immediate from the fact that
  the embedding functor commutes with every structure of the compact closed category.
  \thmend
\end{remark}

We can prove that the SWITCH term does actually define the quantum SWITCH.

\begin{theorem}\label{thm:cat:switch-semantics}
  The categorical semantics of $\SWITCH$ coincides with the quantum SWITCH.
  \thmend
\end{theorem}

The fact that we can define a natural semantics without making any arbitrary choices demonstrates that there is no remaining ambiguity and that the correspondence problem has been resolved in our language. %
This is because the language is designed to follow the categorical structure of the model: quantum branching lives in the \(\Hilb\)-fragment, measurement and classical control live in the \(\CPM\)-fragment, and the two are connected only through the embedding functor \(\EmbeddingFunctor \colon \Hilb \to \CPM\).
Since this embedding preserves the compact closed structure, any pure higher-order construction defined in \(\Hilb\), such as the quantum SWITCH, lifts canonically to \(\CPM\).
Together with the strict linearity of \(\Hilb\)-terms, this means that the required correspondence is built into the syntax and its categorical interpretation, rather than imposed by an arbitrary choice.
 
\section{Operational Semantics}\label{sec:op-sem}
Here, we discuss how a program in our language can be compiled down to a quantum circuit.
We do this by using the transformation technique we discussed in \cref{sec:overview:dummy-free}.
We first formalise the transformation from \Qif{} to the target language \QifUnitary{} in \cref{sec:op-sem:syndil},
and then define an operational semantics of \QifUnitary{} in \cref{sec:op-sem:op-sem}.
In \cref{sec:op-sem:QCQC},
we prove our main theorem, \cref{thm:QCQC},
which states that every first-order program in our language is included in QC-QC~\cite{Wechs2021}.
All the details can be found in \appendixref{app:sec:definitions,app:sec:theorems}.

\subsection{Syntactic Dilation}
\label{sec:op-sem:syndil}

We now formalise the program transformation from \Qif{} to a measurement-free language, whose core idea was discussed in \cref{sec:overview}.
Inspired by Stinespring's dilation theorem, we aim to transform a term of \Qif{} into the form \( \letx{\vec{a}}{\ket{1\ldots1}}{\letx{y \otimes \vec{z}}{U(x \otimes \vec{a})}{(\overrightarrow{\bdiscard(z)}; y)}} \), where \( U \) is a \(\JQ\)-term (a generalisation of a unitary).
However, this transformation is inconvenient because the transformed program \(U\) has a different type from the original program.
We therefore introduce a language \QifUnitary{} that treats the auxiliary input qubits \( \vec{a} \) and output qubits \( \vec{z} \) in a special manner, enabling us to view \( U \) as a mapping \( x \mapsto y \) in the type system.

\begin{mathfig}
  \begin{gather*}
    \begin{array}{llrl}
      \addnew{\textit{Input channels}}
       &
      \addnew{\inch, \inchB, \inchC, \dots}
      \hspace{5em}
      \addnew{\textit{Output channels}}
      \quad
      \addnew{\outch, \outchB, \outchC, \dots \hspace{-50em}}
      \\[5pt]
      \textit{Types}
       & A,B              & \Coloneqq &
      \qbit \sor \unit \sor \cancel{\bool} \sor A \otimes B \sor A \li B
      \\[5pt]
      \textit{Terms}
       & \term,\termB     & \Coloneqq &
      x \sor \unitval \sor \term; \termB \sor \lambda x. \term
      \sor \term\,\termB \sor \term \otimes \termB
      \\ &       & \sor &
      \letx{x}{\term}{\termB} \sor \letx{x \otimes y}{\term}{\termB}
      \sor \cancel{\letx{x}{\F}{\term}}
      \sor \addnew{\F\langle \vinch, \voutch \rangle}
      \\ &       & \sor &
      \qifx{\term}{\termB_1}{\termB_2}
      \sor \cancel{\cifx{\term}{\termB_1}{\termB_2}}
      \\ &       & \sor &
      \cancel{\ket{1}} \sor U \sor \cancel{\meas} \sor \cancel{\btrue} \sor \cancel{\bnot}
      \sor
      \addnew{\inch}
      \sor \addnew{(\outch \assign \term); \termB}
      \\[5pt]
      \textit{Glob.~Def.}
       & \widehat\Defs    & \Coloneqq &
      \cdot \sor
      \widehat\Defs, \globaldef \F\addnew{\langle \vinch, \voutch \rangle} = \term
      \\[5pt]
      \textit{Non-lin.~Ctxt.}
       & \widehat{\Gamma} & \Coloneqq &
      \cdot \sor \widehat{\Gamma}, \F \colon \addnew{(A, n, m)} \quad \text{where}\quad |A| = n - m.
    \end{array}
    \\[5pt]
    |\bool| = |\qbit| = 1, \quad |\unit| = 0, \quad |A \otimes B| = |A| + |B|, \quad |A \li B| = |B| - |A|.
    \\[-20pt]
  \end{gather*}
  \caption{Syntax of \QifUnitary{}. The table shows the differences from \Qif{}
    (red: deleted, blue: added).}
  \label{fig:procedural:qif-unitary:syntax}
  \Description{}
\end{mathfig}
\Cref{fig:procedural:qif-unitary:syntax} summarises the syntax of \QifUnitary{}.
It is essentially the \(\JQ\)-terms of \Qif{} with two differences: (1) it has \emph{input channels} \( \alpha \) and \emph{output channels} \( \zeta \), corresponding to \( \vec{a} \) and \( \vec{z} \) in the above dilation; and (2) each use of \( \F \) is annotated with channels as in \( \F\langle \vec{\alpha}, \vec{\zeta}\rangle \), which intuitively means that a dilation of \( \F \) is of the form \( \letx{y \otimes \vec{\zeta}}{U_\F(x \otimes \vec{\alpha})}{(\bdiscard(\vec{\zeta}); y)} \).
An input channel \( \alpha \) and an output channel \( \zeta \) can be used as terms of types \( \qbit \) and \( \qbit \multimap \unit \), respectively; we write \( \zeta \assign \term \) instead of \( \zeta\,\term \), and its occurrence is restricted to \( ({-}; N) \) for a technical reason.

\begin{mathfig}
  \begin{proofrules}[5pt]
    \infer{}{
      \auxJQ{\EmpEnv}{\inch \colon \qbit}{(\inch,\cdot)}
    }

    \infer{
      \auxJQ{\Delta}{\term \colon \qbit}{(\vinch, \voutchB)}
      \\
      \auxJQ{\Delta'}{\termB \colon A}{(\vinchB, \voutchC)}
    }{
      \auxJQ{\Delta, \Delta'}{(\outch \assign \term); \termB \colon A}
      {(\vinch\vinchB, \outch\voutchB\voutchC)}
    }

    \infer{
      \auxJQ{\!\Delta}{\term \colon \qbit}{(\vinch, \voutch)}
      \\
      \auxJQ{\!\Delta'}{\termB_1 \colon A}{(\vinchB, \voutchB)}
      \\
      \auxJQ{\!\Delta'}{\termB_2 \colon A}{(\vinchB, \voutchB)}
      \\
      \!\!\FirstOrderType{A}
    }{
      \auxJQ{\Delta, \Delta'}
      {(\qifx{\term}{\termB_1}{\termB_2}) \colon \qbit \otimes A}
      {(\vinch\vinchB, \voutch\voutchB)}
    }

    \infer{
      (\F \colon (A, n, m)) \in \widehat\Gamma
      \\\!\!\!
      |\vinch| = n
      \\\!\!\!
      |\voutch| = m
    }{
      \auxJQ{\EmpEnv}
      {\F\langle \vinch, \voutch \rangle \colon A}
      { (\vinch, \voutch) }
    }

    \infer{
      \vdash \Defs \colon \widehat\Gamma
      \\
      \auxJQ{\EmpEnv}{\term \colon A}{(\vinch, \voutch)}
    }{
      \!\vdash (\Defs, \globaldef \F\langle \vinch, \voutch \rangle = \term)
      \colon (\widehat\Gamma, \F \colon (A, |\vinch|, |\voutch|))
    }

    \dots
  \end{proofrules}
  \caption{Selected typing rules for \QifUnitary{}.}
  \label{fig:procedural:qif-unitary:typing-rules}
  \Description{}
\end{mathfig}
The judgement \( \auxJQ{\Delta}{\term \colon A}{(\vinch, \voutch)} \) means that \( \term \) is a term of type \( A \) that uses auxiliary input channels \( \vinch \) and output channels \( \voutch \).
Both $\vinch$ and $\voutch$ are regarded as sets, so permutation does not change the meaning.
Here $\widehat{\Gamma}$ is a non-linear context equipped with the number of auxiliary qubits,
defined in \cref{fig:procedural:qif-unitary:syntax}.
Representative rules are displayed in \cref{fig:procedural:qif-unitary:typing-rules}.
Input channels have the qubit type,
and qubits can be discarded by assigning them to output channels.
For $\qif$, both branches must use the same channel names (cf.~the correspondence problem in \cref{sec:overview:dummy-free}).
Every other rule, such as $\term;\termB$ or $\term\,\termB$,
requires that the two premises use disjoint channel sets.
Global constants from $\Gamma$ may be invoked when they are supplied with
the right numbers of channels.

\begin{mathfig}
  \begin{proofrules}[5pt]
    \meas \SynDil
    \lambda x^\qbit. \letx{y \otimes z}{\gCX(x \otimes \inch)}{(\outch \assign z; y)}

    \ket1 \SynDil \inch

    \btrue \SynDil \inch

    \bnot \SynDil \gX

    \infer{
      \term \SynDil \term' \withaug (\vinch, \voutch)
      \\
      \termB_i \SynDil \termB'_i \withaug (\vinchB_i, \voutchB_i)
      \\
      |\vinchB_1| - |\voutchB_1|
      =
      |\vinchB_2| - |\voutchB_2|
      =
      |\voutchC| - |\vinchC|
    }{
      \cifx{\term}{\termB_1}{\termB_2} \SynDil
      \letx{x \otimes y}{\bigl(
        \qifx{\term'}
        {(\voutchB_2\voutchC\assign\vinchB_2\vinchC; \termB'_1)}
        {(\voutchB_1\voutchC\assign\vinchB_1\vinchC; \termB'_2)}
        \bigr) }{(\outchD \assign x); y}
    }

    \infer{
      \term \SynDil \term'
    }{
      \letx{x}{\F}{\term} \SynDil
      \letx{x}{\F\langle\vinch, \voutch\rangle}{\term'}\!\!
    }

    \hspace{-.4em}A \SynDil A[\qbit/\bool]\hspace{-.4em}

    \infer{
      \Gamma \SynDil \widehat\Gamma
      \hspace*{-1em}\\
      A \SynDil A'
      \hspace*{-1em}\\
      |A| = n\! -\! m
    }{
      \Gamma, \F \colon A \SynDil \widehat\Gamma, \F \colon (A', n, m)
    }
  \end{proofrules}
  \caption{Selected rules for syntactic dilation for terms, types and non-linear contexts.
    $\term \SynDil \term' \withaug (\vinch, \voutch)$ means $\term \SynDil \term'$
    and $(\vinch,\voutch)$ are the i/o channels that appear in $\term'$.}
  \label{fig:procedural:syntactic-dilation}
  \Description{}
\end{mathfig}
We call the program transformation \emph{syntactic dilation} because it is inspired by Stinespring's dilation theorem.
Some selected rules for the syntactic dilation are presented in \cref{fig:procedural:syntactic-dilation}.
Measurements expand to a $\gCX$ gate that interacts with a freshly introduced auxiliary qubit,
capturing the minimal dilation of the measurement map.
The new qubit is introduced as $\inch$ and discarded as $\outch$.
Both $\ket1$ and $\btrue$ translate to the identity on an auxiliary qubit,
and hence become $\inch$.
Classical conditionals $\cif$ are rewritten as quantum branches $\qif$.
Because the branches of $\cif$ may allocate different numbers of ancillary qubits,
we extend each branch with auxiliary inputs $\vinchB_i$
that are immediately discarded as $\voutchB_i$.
The rule requires that the total difference
$|\vinchB_i| - |\voutchB_i|$ is balanced between the branches,
which follows automatically from the type-balance constraint $|A|$ on well-typed terms.

\begin{theorem}[Type preservation for syntactic dilation]
  Assume
  $\Gamma \SynDil \widehat\Gamma$,
  $\Delta \SynDil \Delta'$ and $A \SynDil A'$.
  Then,
  for any term $\termJC{\Delta}{\term \colon A}$ in \Qif{},
  there exists a \QifUnitary{} term
  $\term'$ and $(\vinch, \voutch)$ such that
  $\term \SynDil \term'$ and
  $\auxJQ[\widehat\Gamma]{\Delta'}{\term' \colon A'}{(\vinch, \voutch)}$.
  \thmend
\end{theorem}

\begin{remark}\label{rem:syntactic-dilation-for-murao-style}
  We can define translations into \QifUnitary{} from the
  \khchanged{loop-free fragment of the} languages studied by
  \citet{Ying16-the-book}, \citet{Badescu2015}, and \citet{Barsse2026}
  by modifying the rule for $\qif$ as follows, for suitable choices of $U_i$:
  \begin{proofrule}
    \infer{
      \term_i \SynDil \term'_i \withaug (\vinch_i, \voutch_i)
      \\
      |\vinch_1| - |\voutch_1|
      =
      |\vinch_2| - |\voutch_2|
      =
      |\voutchB| - |\vinchB|
    }{
      \qifx{x}{\term_1}{\term_2} \SynDil
      \qifx{x}
      {(\voutch_2\voutchB \assign \gU_{2}(\vinch_2\vinchB); \term'_1)}
      {(\voutch_1\voutchB \assign \gU_{1}(\vinch_1\vinchB); \term'_2)}
      .
    }
  \end{proofrule}
  One can see that the auxiliary qubits $\vinch_1$ used in the then-branch
  are, in the else-branch, initialised to $U_1 \ket1$ and then immediately discarded.
  This state $U_1 \ket1$ is what \citet{Abbott2020} calls the \emph{initial state of the environment}.
  Roughly speaking,
  in \citet{Barsse2026}, $U_i$ is $I$ (the identity);
  in \citet{Badescu2015}, $U_i$ is a unitary that prepares a uniform superposition
  over the computational basis;
  and in \citet{Ying16-the-book},
  the degree of freedom in the generalised semantics corresponds
  to the freedom to choose $U_i$.
  See \appendixref{app:sec:vacuum,appx:sec:ying} for details.
  \thmend
\end{remark}

The categorical semantics of \QifUnitary{} is defined similarly to that of Hilb-terms of \Qif{}.
A judgement $\auxJQ{\Delta}{\term\colon A}{(\vinch, \voutch)}$
is interpreted as a map
$\sem{\term}_\eval \colon
  \sem{\Delta} \otimes 2^{\otimes |\vinch|} \longrightarrow \sem{A} \otimes 2^{\otimes |\voutch|}$
in $\Hilb$
for each $\eval\in\prod_{(\F\colon(B,n,m))\in\widehat\Gamma}
\Hilb(\CC^{2^n},\sem{B}_\Hilb\otimes\CC^{2^m})$.
Let $\csem{\term}$ denote the semantics in $\CPM$ in which qubit initialisation
and discarding are supplied to the i/o channels, \ie,
$\csem{\term} \defeq (\ident_A \otimes \discardDiag_{\voutch})
  \circ \EmbeddingFunctor\sem{\term}
  \circ (\ident_\Delta \otimes \ketbra11_{\vinch})$.

\begin{theorem}[Soundness of the syntactic dilation]
  \label{thm:soundness-syntactic-dilation}
  For each $\CPM$-term $\term$ whose
  types of free variables and the term itself are first order,
  if $\term \SynDil \term'$,
  $
    \sem{\term} = r_{A} \circ \csem{\term'} \circ s_{\Delta}
  $.
  \thmend
\end{theorem}

We also state that the translations from the languages in the literature
\cite{Ying16-the-book,Badescu2015,Barsse2026} are sound.

\begin{theorem}[Soundness of the translation from other languages]
  \label{thm:soundness-of-translation-murao-style}
  For each language for which we defined a translation to \QifUnitary{} in
  \cref{rem:syntactic-dilation-for-murao-style},
  the translation preserves the semantics.
  \thmend
\end{theorem}

\subsection{Operational Semantics for QifUnitary}
\label{sec:op-sem:op-sem}

We next sketch the small-step operational semantics of \QifUnitary{}
in \cref{fig:procedural:qif-unitary:ev-ctxt,fig:procedural:qif-unitary:operational-semantics}.
It is based on the operational semantics of the quantum lambda calculus~\cite{Selinger2008}, with two major differences.
The first stems from the presence of \(\qif\), and the second pertains to the output channel \(\outch\).
\begin{definition}
  A \emph{configuration} $(\vinch \mapsto \ket\psi, \term)$
  consists of a quantum state $\ket\psi \in (\CC^2)^{\otimes |\vinch|}$
  together with a term $\term$ equipped with a derivation
  $\auxJQ{\EmpEnv}{\term\colon A}{(\vinchB, \voutch)}$
  such that $\vinchB \subseteq \vinch$.
  A term is in \emph{value form} if it has the shape
  $(\outch_1\assign\inch_1); \dots; (\outch_n\assign\inch_n); \val$,
  which we abbreviate as $(\voutch\assign\vinch); \val$.
  A normal form is an irreducible value form; abstractions may still have
  reducible bodies.
  \thmend
\end{definition}

\begin{mathfig}
  \begin{align*}
    \begin{array}{llrl}
      \textit{Value}
       & \val, \valB                       & \Coloneqq &
      \inch
      \sor \unitval
      \sor \val \otimes \valB
      \sor \lambda x. \term
      \sor U
      \\[5pt]
      \textit{Context with control}
       & \ctxt[][\emptyset]                & \Coloneqq &
      [\cdot]
      \\
       & \ctxt[][\qctrl]                   & \Coloneqq &
      \ctxt; \term
      \sor \term\,\ctxt
      \sor \ctxt\,\term
      \sor \ctxt \otimes \term
      \sor \term \otimes \ctxt
      \\ &       & \sor &
      \qifx{\ctxt}{\termB_1}{\termB_2}
      \sor \cdots
      \\
       & \ctxt[][\qctrl, \alpha \mapsto 1] & \Coloneqq &
      \qifx{\inch}{\ctxt[][\qctrl]}{\term}
      \\
       & \ctxt[][\qctrl, \alpha \mapsto 0] & \Coloneqq &
      \qifx{\inch}{\term}{\ctxt[][\qctrl]}
    \end{array}
  \end{align*}
  \caption{Values and context with control information in \QifUnitary{}. Here, $c$ is an assignment of $0$ or $1$ to a finite set of input channels $\vinch$.}
  \label{fig:procedural:qif-unitary:ev-ctxt}
  \Description{}
\end{mathfig}
The pure rules, including $\beta$-reduction, apply in arbitrary
contexts $\ctxt[][]$, also under abstractions.
Only the gate rule changes the quantum state.
It uses a controlled context $\ctxt[]$: every enclosing qif branch has a known
input channel as its condition, and $c$ records the chosen branch values.
A gate at its hole acts as $C_cU$ ($U$ gate controlled by $c$).
For a qif with two output-prefix ground values, the pure elimination rule
generates a controlled permutation followed by tensor bindings.
Its later execution uses the same gate rule and the enclosing control
information.
\begin{mathfig}
  \begin{proofrules}[5pt]
    \ctxt[\letx{x}{\termB}{\term}][]
    \step
    \ctxt[\term[\termB/x]][]

    \ctxt[\term\,((\outch \assign \termC); \termB)][]
    \step
    \ctxt[(\outch \assign \termC); (\term\,\termB)][]

    \ctxt[\unitval; \term][]
    \!\!\step\!
    \ctxt[\term][]
    \hspace*{-.5em}

    \ctxt[(\lambda x. \term) \termB][]
    \!\!\step\!
    \ctxt[\term[\termB/x]][]
    \hspace*{-.5em}

    \infer{
      \globaldef \F\langle\vinch, \voutch\rangle = \term_\F
      \in \widehat\Defs
    }{
      \ctxt[\F\langle\vinchB, \voutchB\rangle][]
      \step_{\widehat\Defs}
      \ctxt[\term_\F[\vinchB/\vinch, \voutchB/\voutch]][]\!\!
    }

    \vinchB \mapsto \ket\psi,\,
    \ctxt[U (\inch_1 \otimes \cdots \otimes \inch_n)][\qctrl]
    \step
    \vinchB \mapsto \bigl(C_\qctrl U_{\inch_1,\dots,\inch_n}\bigr) \ket\psi,\,
    \ctxt[\inch_1 \otimes \cdots \otimes \inch_n][\qctrl]

    \ctxt[][]
    \left[\begin{matrix*}[l]
        \qif\ {\inch}\ \ifthen
        \\ \quad
        (\outch_1 \assign \inchB_{\sigma(1)}); \dots; (\outch_n \assign \inchB_{\sigma(n)});
        \\ \quad
        \val[\inchB_{\sigma(n+1)}, \dots, \inchB_{\sigma(m)}]
        \\ \ifelse
        \\ \quad
        (\outch_1 \assign \inchB_1); \dots; (\outch_n \assign \inchB_n);
        \\ \quad
        \val[\inchB_{n+1}, \dots, \inchB_{m}]
      \end{matrix*}\right]
    \step
    \ctxt[][]
    \left[\begin{matrix*}[l]
        \Let {x \otimes y_1 \otimes \cdots \otimes y_m} =
        \\
        \quad
        {\mathsf{C}\ceil{\sigma}\,(\inch \otimes \inchB_1 \otimes \cdots \otimes \inchB_m)}
        \\
        \In
        (\outch_1 \assign y_1); \dots; (\outch_n \assign y_n);
        \\
        \quad\ x \otimes \val[y_{n+1}, \dots, y_m]
      \end{matrix*}\right]
  \end{proofrules}
  \caption{Selected reduction rules for \QifUnitary{}.
    We implicitly assume a well-typed definition \( \Defs \).
    Only the gate rule changes the state; all other rules preserve it.
    Unsubscripted $\ctxt[][]$ is an arbitrary context, whereas $\ctxt[]$ is controlled.
    Here
    $C_\qctrl U_{\vinch}$ denotes the unitary $U$ applied to qubits $\vinch$
    controlled by $\qctrl$.
    For a permutation $\sigma\colon \{1,\dots,m\} \to \{1,\dots,m\}$,
    $\lceil\sigma\rceil$ denotes the permutation unitary, \ie,
    it maps $\ket{i_1 \cdots i_m}$
    to $\ket{i_{\sigma(1)} \cdots i_{\sigma(m)}}$.
    The last rule has fresh bound variables and abbreviates nested tensor-let;
    $\mathsf C\lceil\sigma\rceil=\ketbra00\otimes I+
      \ketbra11\otimes\lceil\sigma\rceil$.
    We identify $(\outch \assign \inch); (\outchB \assign \inchB); \term$
    with $(\outchB \assign \inchB); (\outch \assign \inch); \term$.
  }
  \label{fig:procedural:qif-unitary:operational-semantics}
  \Description{}
\end{mathfig}

The other rules are analogous to those of the quantum lambda calculus~\cite{Selinger2008}, except for the following points:
(1) Global calls substitute $\F$ with its defining term after $\alpha$-renaming channels; and (2) there are many commuting conversions that pull assignments $(\outch \assign \inch)$ out of evaluation contexts.
The complete list of rules is found in \appendixref{app:sec:definitions}.
\begin{theorem}[Type safety, progress and termination]
  \label{thm:op-sem-theorem}
  Reduction preserves types, and every well-typed closed term strongly normalises to a value.
  \thmend
\end{theorem}

The evaluation of a \QifUnitary{} program is defined by first reducing the program to its normal form and subsequently measuring the qubits on the output channels.

The soundness holds for the operational semantics.
Together with the properties we prove in \cref{thm:op-sem-theorem}, the adequacy theorem follows.

\begin{theorem}[Soundness]
  \label{thm:soundness-opsem}
  If $(\ket\psi, \term) \step_\Defs (\ket\phi, \termB)$,
  then
  $\sem{\term}_{\eval}\ket\psi = \sem{\termB}_{\eval}\ket\phi$.
  \thmend
\end{theorem}

\begin{theorem}[Adequacy]
  \label{thm:adequacy}
  For each closed $\CPM$-term $\term$ of first-order type $A$,
  there exists a \QifUnitary{} term $\term'$ such that
  $\term \SynDil \term'$ and there is a reduction sequence
  $(\vinch \mapsto \ket{1\dots1}, \term') \step_\Defs^*
    (\vinch \mapsto \ket{\varphi}, (\voutch\assign\vinch'); \val)$
  satisfying
  $
    \trace_{\vinch'}(
    {\sigma}\ketbra\varphi\varphi{\sigma^{-1}}
    )
    =
    s_A \circ \sem{\term}_{\sem{\Defs}}
  $
  where $\sigma$ is a canonical permutation unitary that derives from $V$,
  and $\trace_{\vinch'}$ is the partial trace operator that corresponds to $\vinch'$.
  \thmend
\end{theorem}

When the term $\term$ has a first-order type $A$,
the section $s_A$ and retraction $r_A$ can be realised by a term consisting of a sequence of $\ket{1}$ and $\meas$.
Therefore, we can conclude that the combination of
the syntactic dilation and operational semantics of \QifUnitary{}
correctly implements the \Qif{} program semantics.

\begin{corollary}
  The operational semantics of $\SWITCH$ gives an execution of
  the quantum SWITCH.
\end{corollary}

\begin{remark}
  Our compilation into ordinary gate circuits does not contradict the
  impossibility of implementing SWITCH with a single black-box use of each
  input operation~\cite{Chiribella2013}. The compilation uses
  \emph{white-box} access: it is given the program definitions of the supplied
  functions and can inspect and transform them.
  Syntactic dilation first translates these definitions into measurement-free
  \QifUnitary{} terms with explicit auxiliary input and output channels.
  Reduction then substitutes the translated terms at their uses and executes
  their gates under the appropriate quantum controls.
  For example, for known gates $U,V$ and input channels $\inch,\inchB$,
  \[
    \begin{aligned}
       & \letx{f}{\gU}{\letx{g}{\gV}{\qifx{\inch}{f(g(\inchB))}{g(f(\inchB))}}}
      \\[-.2ex]
       & \qquad\step^*\qifx{\inch}{\gU(\gV(\inchB))}{\gV(\gU(\inchB))}.
    \end{aligned}
  \]
  Executing the then-branch first gives the temporal gate sequence
  $C_1V; C_1U; C_0U; C_0V$, where $C_bW$ applies $W$ to $\inchB$
  when $\inch$ is $\ket b$.
  Thus each of $U$ and $V$ appears twice, under different controls,
  with the controlled gates synthesised from their known descriptions.
  \thmend
\end{remark}

We note that the operational semantics of \QifUnitary{} can be used to provide an operational semantics
for the languages discussed in \cref{rem:syntactic-dilation-for-murao-style}.

\begin{corollary}
  \label{cor:op-sem-for-murao-style}
  Together with \cref{rem:syntactic-dilation-for-murao-style,thm:soundness-of-translation-murao-style},
  we can provide a circuit-based operational semantics for the while-free fragment of the languages
  studied by \citet{Ying16-the-book}, \citet{Badescu2015}, and \citet{Barsse2026}.
  \thmend
\end{corollary}

\subsection{Theorems} %
\paragraph{Inclusion in QC-QC}
\label{sec:op-sem:QCQC}

To study the realisability of programs in our language, we relate their
denotations to \emph{quantum circuits with quantum control}
(QC-QCs)~\cite{Wechs2021}.
QC-QCs form a class of quantum supermaps, \ie,
higher-order (linear) maps
$\bigotimes_i (A_i \li B_i) \li (P \li F)$
that transform quantum operations into quantum operations,
while preserving the trace preserving property.
For unitary input operations $U_i$, a QC-QC contains controlled applications of the $U_i$ at multiple stages,
arranged so that each $U_i$ acts exactly once along each coherent execution history.
Different histories may apply the operations in different orders, as in the quantum SWITCH.
The exact definition, which also covers general quantum operations, is
recalled in \appendixref{app:qcqc:qcqc}.

\begin{theorem}[QC-QC]
  \label{thm:QCQC}
  Consider a \Qif{} term
  $\termJC{f_1 \colon A_1\li B_1,\ldots,f_n \colon A_n\li B_n}{\term \colon X \li Y}$
  where $A_i,B_i,X,Y$ are first order and purely quantum.
  Then the map $\sem{\term}$ defines a QC-QC map.
  \thmend
\end{theorem}

\begin{proof}[Proof sketch]
  First consider terms of such types built only from quantum data,
  with no internal higher-order functions or measurements.
  Those programs can be thought of as direct descriptions of
  a quantum circuit with holes for the external operations.
  QC-QCs include known unitary operations,
  wiring, and calls to the external operations.
  They are closed under composition, tensor products,
  and under $\qif$ between branches using the same
  external inputs, as linear typing requires.
  These closure properties match the term constructors, giving a QC-QC
  interpretation of this fragment.

  For a general term that may internally contain measurements and higher-order functions,
  our compilation procedure allows us to reduce it to a term in the preceding fragment.
  Syntactic dilation first replaces measurements and
  classical data by a measurement-free computation with auxiliary
  qubits.
  Operational normalisation then removes internal higher-order
  functions while leaving external calls unevaluated.
  Soundness of these transformations
  identifies its denotation with the original one.
  These operations are allowed
  in a QC-QC. The full proof is given in \appendixref{app:qcqc}.
\end{proof}

\paragraph{Full Abstraction}

We can add divergence \( \abort : \unit \) as a constant to the language \Qif{},
and define its type derivation and semantics in a canonical way:
\begin{proofrules}
  \textit{Terms}\hspace{1.5em}
  \term,\termB\ \Coloneqq\ \cdots \sor \abort

  \termJC[\EmpEnv]{\EmpEnv}{\abort\colon\unit}

  \sem{\,\abort\,}\ =\ 0.
\end{proofrules}
We can prove that this language is fully abstract.

\begin{theorem}[Full abstraction]
  In \Qif{} extended with divergence,
  for each $\JG = \JQ$ or $\JC$ and
  terms
  $\termJG{\Delta}{\term \colon A}$
  and
  $\termJG{\Delta}{\termB \colon A}$,
  $\sem{C[\term]}_\CPM = \sem{C[\termB]}_\CPM$ for any context $C[\cdot]$
  of type $\unit$,
  if and only if
  $\sem{\term} = \sem{\termB}$.
  The full proof is given in \appendixref{app:sec:full-abstraction}.
  \thmend
\end{theorem}

In \appendixref{app:sec:theorems}, we also extend the language \QifUnitary{} so that
a syntactic dilation of \Qif{} with divergence can be defined,
and prove the normalisation and adequacy theorem.
 
\section{Related Work}\label{sec:related}

\paragraph{Quantum Conditional on Pure Quantum Computation.}

As discussed in the Introduction, quantum control of operations
involving measurement poses well-known semantic difficulties.
Many languages therefore impose syntactic restrictions that limit
quantum control to pure quantum computations, such as unitaries
or isometries.
Even within this restricted setting, various syntactic approaches
to quantum control have been proposed.
Besides QML~\cite{QML}, one of the earliest proposals, these include
symmetric pattern matching~\cite{SabValiViz18-FOSSACS-q-recusion}
and its successors~\cite{ChardonnetSV21,ChardonnetSV23,DaveLPZ25},
the work of \citet{HeunenK22} and its
successors~\cite{Carette2023,CaretteHKS24},
and extensions of Quipper~\cite{GreenLRSV13-quipper,FuKRS25}.

A few of these languages come with a
circuit-based operational semantics.
\citet{FuKRS25} study an extension of
Quipper~\cite{GreenLRSV13-quipper} with quantum control and
circuit reversal.
Their language distinguishes several levels of reversibility
and controllability, and its linear-logic-based type system
supports higher-order constructs such as lambda abstraction.
Granthi~\cite{Granthi2026} is another language closely related to
our purely quantum fragment.
Like ours, it draws on ideas from linear logic:
cut-elimination-style normalisation exposes the underlying
wiring, which is then compiled into quantum circuits.

\citet{HirataT26} define a language with $\qif$ whose branches can be higher-order terms of pure quantum computation.
However, an operational semantics for this language remains an open problem.
In relation to our work, our categorical semantics can also be defined using the functor $\iota \colon \CausHilb \longrightarrow \CausCPM$ that they introduced, instead of $\iota \colon \Hilb \longrightarrow \CPM$.
From this semantics, it immediately follows that
$\sem{\term \colon \qbit^{\otimes n} \li \qbit^{\otimes n}}_\Hilb$ is unitary,
$\sem{\termB \colon \qbit^{\otimes n} \li \qbit^{\otimes m}}_\CPM$ is CPTP, and 
$\sem{\termC \colon \bigotimes_i (\qbit^{\otimes n_i} \li \qbit^{\otimes m_i}) \li \qbit^{\otimes k} \li \qbit^{\otimes \ell}}$
is a supermap.
Note that \cref{thm:QCQC} establishes the stronger claim that
$\sem{\termC}$ is not only a supermap but also a QC-QC map.

\paragraph{Quantum Conditional on Quantum Channels.}

To the best of our knowledge,
there are two existing languages that can express the quantum SWITCH
while also supporting quantum measurement.

The first is the language of \citet{DaveDZ25}.
It also has a two-layer structure.
Indeed, their categorical model is essentially the same as ours:
it is based on a functor \( \iota \colon \Hilb \longrightarrow \CPM \).
Their term calculus is based on the algebraic lambda calculus,
in which terms can be freely summed.
Sums of terms are interpreted as sums in \(\Hilb\) or \(\CPM\),
depending on the context.
In particular, this allows one to write terms that succeed
with probability greater than one,
which are not physically realisable.

The other is the language of \citet{barsse2026-arxiv},
developed independently and concurrently with ours.
Their language has quantum data types
but no classical data types.
Instead of separate constructs for measurement,
which produces a Boolean, and for classical branching on that Boolean,
it has a single construct that combines both.
Their type system prohibits measurement in the branches of a \(\qif\),
so these branches correspond to our purely quantum fragment.
Their categorical interpretation is also based on the functor
\( \iota \colon \Hilb \longrightarrow \CPM \);
moreover, they prove that the semantics lies in the category \(\CausCPM\)
of \citet{Simmons2022}.
Consequently, no term in their language succeeds with
probability greater than one,
and on this basis they argue that their language is physically meaningful.
In particular, they prove that a term of type
$\bigotimes_i(q^{n_i} \li q^{m_i}) \li q^k \li q^\ell$,
where $q^n$ is the type of $n$ qubits,
denotes a supermap.
They also illustrate how a subclass of QC-QCs,
which they call QC-QCs with memory,
can be expressed in their language.

In contrast, we establish considerably stronger
physical realisability results.
First, our operational semantics provides a concrete
circuit-synthesis procedure via syntactic dilation.
Second, we prove that every term of type
$\bigotimes_i (\qbit^{\otimes n_i} \li \qbit^{\otimes m_i}) \li \qbit^{\otimes k} \li \qbit^{\otimes \ell}$
defines a QC-QC map.
Such terms therefore admit the in-principle photonic implementations
discussed by \citet[Sec.~V.E]{Wechs2021}.
Furthermore, the QC-QC maps they define do not violate causal inequalities
without postselection~\cite[Sec.~V.F]{Wechs2021}.

\citet{barsse2026-arxiv} also discuss an extension
of their language with recursion and give a term-rewriting
operational semantics for it.
However, they do not provide a categorical semantics for this
extension or establish its physical realisability.
In addition, their repeat-until-success example
(Example~7.2), as written, does not appear to be typable
under their typing rules.
For the recursive function \(f : q \multimap B\), the success branch \(\langle y, f \rangle\) has type \(q \otimes (q \multimap B)\), whereas the retry branch \(f\,y\) has type \(B\). Typing the measurement requires these types to coincide, but the resulting equation \(B = q \otimes (q \multimap B)\) has no solution in their type grammar.

Several more languages support quantum control of general, not necessarily unitary,
quantum channels~\cite{Badescu2015,Ying16-the-book,Barsse2026}.
In each of them, if a term $\term$ denotes a channel $F \in \CPM(n,n)$,
the term $\qifx{x}{\term\,y}{y}$ denotes a channel $\mathrm{C}F$ satisfying
$
  \mathrm{C}F(\ketbra{0}{0}_x \otimes \rho) = \ketbra{0}{0}_x \otimes \rho
$
and
$
  \mathrm{C}F(\ketbra{1}{1}_x \otimes \rho) = \ketbra{1}{1}_x \otimes F(\rho).
$
As discussed in \cref{sec:intro}, these equations do not determine
$\mathrm{C}F$ uniquely, and the three languages resolve this ambiguity in different ways.
The translations in \cref{rem:syntactic-dilation-for-murao-style} make the
resolution explicit: they map the while-free fragment of each language into
\QifUnitary{}, and its controlled operations correspond to particular choices
of the unitary $V$ in \cref{fig:intro:controlled-channel}.
Combined with our operational semantics, these translations yield a
compilation of these fragments into quantum circuits
(\cref{cor:op-sem-for-murao-style}); to the best of our knowledge, no such
compilation has been given before.
The translations are detailed in \cref{app:sec:vacuum,appx:sec:ying}.\footnote{%
  Unlike the other two, the translation for \citet{Ying16-the-book} is not
  purely syntactic: choosing the unitaries requires the norms of the Kraus
  operators of the controlled term, which must be computed from its semantics.}

\paragraph{Other Work.}

Quantum programming languages with only classical branching are well studied.
For example, fully abstract models based on \( \CPM \) are known for several settings~\cite{Selinger2004,Pagani2014,Clairambault2020,Tsukada2024}.
Quantum branching over unitary operations has also been well studied.
Such branching can be implemented by circuits, and various languages have been proposed to express it in programs~\cite{QML,Qsharp,Silq,HirataPOPL25,SabValiViz18-FOSSACS-q-recusion}.

\section{Conclusion and Future Work}\label{sec:conc}
We have presented a programming language for describing quantum-controlled computational processes, such as the quantum SWITCH~\cite{Chiribella2013}.
It has been known that quantum control poses semantic challenges~\cite{Badescu2015}, and a recent survey has suggested that its realisation has a fundamental difficulty~\cite{Valiron2022}.
By analysing existing semantics through program transformation, we identified the source of the difficulty and demonstrated that a simple linear type system can resolve the semantic issues.

There are several important directions for future research.
The first direction is the exploration of possible extensions to the expressive power of the language.
In particular, we are interested in what extensions would allow us to describe the quantum protocol in \citet{Oreshkov2012}.
Another direction is to integrate with the language defined by \citet{HirataT26}, which allows higher-order functions as the bodies of $\qif$.
Finally, we would like to describe and analyse new quantum procedures using the proposed language.
 
\begin{acks}
  This work was supported by JST CREST, Japan, Grant Number JPMJCR25I5.
\end{acks}

\bibliographystyle{ACM-Reference-Format}
\bibliography{references}

\appendix

\section{Correspondence Problem via Kraus Decompositions}
\label{app:sec:correspondence-problem-from-Kraus-perspective}

The key difference between the quantum SWITCH and general controlled operations
lies in the \emph{correspondence of indices} between the then- and else-branches.
Let $\{K_j\}_{j\in J}$ and $\{H_i\}_{i\in I}$ be Kraus decompositions of $F$ and $G$, respectively.
In the Kraus decomposition of the quantum SWITCH,
\begin{equation*}
  \mathit{SWITCH}({-},{-},F,G) = \big\{\, \ketbra{0}{0} \otimes H_i K_j + \ketbra{1}{1} \otimes K_j H_i \,\big\}_{(i,j) \in I \times J},
\end{equation*}
the same indices $(i,j)$ are used in the then- and else-branches,
whereas in the semantics of general quantum control $\qifx{x}{F(y)}{G(y)}$,
the indices $i$ and $j$ are selected independently.
In contrast, for the statement $\qifx{x}{F(y)}{G(y)}$,
if the Kraus decompositions of $F$ and $G$ happen to share an index set,
say $\{K_i\}_{i \in I}$ and $\{H_i\}_{i \in I}$,
then we can define its semantics as
$\{ \ketbra00 \otimes H_i + \ketbra11 \otimes K_i \}_{i \in I}$
in the style of the quantum SWITCH,
without introducing coefficients for normalisation.

In order to define a semantics for quantum conditional statements independently of a
specific choice of Kraus decomposition, merely sharing a common index set does
not suffice.
In general, a Kraus decomposition is not unique.
If $\{ K_i \}_{i \in I}$ is a Kraus decomposition of $F$
and $U = (u_{ji}) \colon \CC^{|I|} \longrightarrow \CC^{|J|}$ is an isometry, \ie, $U^\dagger U = I$,
then $\{ \sum_{i} u_{ji} K_i \}_{j \in J}$ is again a Kraus decomposition of $F$.
Conversely, any two Kraus decompositions of $F$ are related in this way:
after padding them with zero operators to a common index set,
they differ by a unitary matrix~\cite{QCQI}.
Let us say that $\{K_i\}_{i \in I}$ and $\{H_i\}_{i \in I}$ are Kraus decompositions of
$F$ and $G$, respectively.
For each unitary $U$ on $\CC^{|I|}$, replacing $\{K_i\}$ with
$\{ K'_j \defeq \sum_{i \in I} u_{ji} K_i \}_{j \in I}$,
which is also a Kraus decomposition of $F$,
generally yields an inequivalent Kraus decomposition of the semantics of
$\qifx{x}{F(y)}{G(y)}$:
\begin{align*}
  \textstyle
  \left\{ \ketbra{0}{0} \otimes H_i
  + \ketbra{1}{1} \otimes K_i \right\}_{i\in I}
  \quad\not\simeq\quad
  \left\{ \ketbra{0}{0} \otimes H_j
  + \ketbra{1}{1} \otimes {\sum_{i \in I} u_{ji} K_i} \right\}_{j\in I}
  .
\end{align*}

However, instead of changing only the Kraus decomposition of $F$,
suppose that we \emph{simultaneously reindex} the Kraus decompositions of $F$ and $G$.
That is, suppose that we replace $K_i$ and $H_i$ with
$\{ K'_j \defeq \sum_i u_{ji} K_i \}$
and
$\{ H'_j \defeq \sum_i u_{ji} H_i \}$,
respectively, for an isometry $U = (u_{ji}) \colon \CC^{|I|} \longrightarrow \CC^{|J|}$.  Then we obtain equivalent completely positive maps:
\begin{align*}
  \textstyle
  \left\{
  \ketbra{0}{0} \otimes H'_j
  + \ketbra{1}{1} \otimes K'_j
  \right\}_{j \in J}
   & = \textstyle
  \left\{
  \ketbra{0}{0} \otimes \sum_{i \in I} u_{ji} H_i
  + \ketbra{1}{1} \otimes \sum_{i \in I} u_{ji} K_i
  \right\}_{j \in J}
  \\ &= \textstyle
  \left\{
  \sum_{i \in I} u_{ji} (
  \ketbra{0}{0} \otimes H_i
  + \ketbra{1}{1} \otimes K_i
  )
  \right\}_{j \in J}
  \\ &\simeq \textstyle
  \left\{
  \ketbra{0}{0} \otimes H_i
  + \ketbra{1}{1} \otimes K_i
  \right\}_{i \in I}.
\end{align*}
Returning to the definition of the quantum SWITCH, this observation explains why
the quantum SWITCH admits a canonical semantics.
For example, take another Kraus decomposition $\{ K'_\ell \}_{\ell \in L}$ of $F$.
Padding with zero operators if necessary, we may assume that
$K'_\ell = \sum_{j} u_{\ell j} K_j$ for some unitary $U = (u_{\ell j})$.
Then we can calculate as follows:
\begin{align*}
  \textstyle
  \left\{
  \ketbra{0}{0} \otimes H_i K'_\ell
  + \ketbra{1}{1} \otimes K'_\ell H_i
  \right\}_{(i, \ell)\in I \times L}
  &= \textstyle
  \left\{
  \sum_j u_{\ell j} \bigl(
  \ketbra{0}{0} \otimes H_i K_j
  + \ketbra{1}{1} \otimes K_j H_i
  \bigr) \right\}_{(i, \ell)\in I \times L}
  \\
  &\simeq \textstyle
  \left\{
  \ketbra{0}{0} \otimes H_i K_j
  + \ketbra{1}{1} \otimes K_j H_i
  \right\}_{(i, j)\in I \times J}.
\end{align*}

As in the quantum SWITCH, if each quantum channel appears exactly once in each branch, then the simultaneous-reindexing property holds.
Therefore, enforcing the linear use of channels, as in our language,
yields a unique canonical semantics for each term.
\section{Exhaustive Definitions of Languages}
\label{app:sec:definitions}

\begin{mathfig}
  \begin{proofrules}
    \infer{}{
      \auxJQ{\EmpEnv}{\inch \colon \qbit}{(\inch,\cdot)}
    }

    \infer{
      \auxJQ{\Delta}{\term \colon \qbit}{(\vinch, \voutchB)}
      \\
      \auxJQ{\Delta'}{\termB \colon A}{(\vinchB, \voutchC)}
    }{
      \auxJQ{\Delta, \Delta'}{(\outch \assign \term); \termB \colon A}
      {(\vinch\vinchB, \outch\voutchB\voutchC)}
    }

    \infer{}{
      \auxJQ{x \colon A}{x\colon A}{(\cdot,\cdot)}
    }

    \infer{
      \auxJQ{\Delta, x \colon A}{\term \colon B}{(\vinch, \voutch)}
    }{
      \auxJQ{\Delta}{\lambda x^A. \term \colon A \li B}{(\vinch, \voutch)}
    }

    \infer{
      \auxJQ{\Delta}{\term \colon \unit}{(\vinch, \voutch)}
      \\
      \auxJQ{\Delta'}{\termB \colon A}{(\vinchB, \voutchB)}
    }{
      \auxJQ{\Delta,\Delta'}{\term;\termB \colon A}{(\vinch\vinchB, \voutch\voutchB)}
    }

    \infer{
      \auxJQ{\Delta}{\term \colon A \li B}{(\vinch, \voutch)}
      \\
      \auxJQ{\Delta'}{\termB \colon A}{(\vinchB, \voutchB)}
    }{
      \auxJQ{\Delta,\Delta'}{\term\,\termB \colon B}{(\vinch\vinchB, \voutch\voutchB)}
    }

    \infer{
      \auxJQ{\Delta}{\term \colon A}{(\vinch, \voutch)}
      \\
      \auxJQ{\Delta'}{\termB \colon B}{(\vinchB, \voutchB)}
    }{
      \auxJQ{\Delta,\Delta'}{\term\otimes\termB \colon A \otimes B}{(\vinch\vinchB, \voutch\voutchB)}
    }

    \infer{
      \auxJQ{\Delta}{\term \colon A}{(\vinch, \voutch)}
      \\
      \auxJQ{\Delta', x\colon A}{\termB \colon B}{(\vinchB, \voutchB)}
    }{
      \auxJQ{\Delta,\Delta'}{\letx{x}{\term}{\termB} \colon B}{(\vinch\vinchB, \voutch\voutchB)}
    }

    \infer{
      \auxJQ{\Delta}{\term \colon A \otimes B}{(\vinch, \voutch)}
      \\
      \auxJQ{\Delta', x\colon A, y \colon B}{\termB \colon C}{(\vinchB, \voutchB)}
    }{
      \auxJQ{\Delta,\Delta'}{\letx{x\otimes y}{\term}{\termB} \colon C}{(\vinch\vinchB, \voutch\voutchB)}
    }

    \infer{
      U \colon \CC^{2^n} \longrightarrow \CC^{2^n}, \mbox{unitary}
    }{
      \auxJQ{\EmpEnv}{U \colon \qbit^{\otimes n} \li \qbit^{\otimes n}}{(\cdot,\cdot)}
    }

    \infer{} %
    {
      \auxJQ{\EmpEnv}{\unitval \colon \unit}{(\cdot,\cdot)}
    }

    \infer{
      \auxJQ{\!\Delta}{\term \colon \qbit}{(\vinch, \voutch)}
      \\
      \auxJQ{\!\Delta'}{\termB_1 \colon A}{(\vinchB, \voutchB)}
      \\
      \auxJQ{\!\Delta'}{\termB_2 \colon A}{(\vinchB, \voutchB)}
      \\
      \!\!\FirstOrderType{A}
    }{
      \auxJQ{\Delta, \Delta'}
      {(\qifx{\term}{\termB_1}{\termB_2}) \colon \qbit \otimes A}
      {(\vinch\vinchB, \voutch\voutchB)}
    }

    \infer{
      (\F \colon (A, n, m)) \in \widehat\Gamma
      \\
      |\vinch| = n
      \\
      |\voutch| = m
    }{
      \auxJQ{\EmpEnv}
      {\F\langle \vinch, \voutch \rangle \colon A}
      { (\vinch, \voutch) }
    }

    \infer{
      \vdash \Defs \colon \widehat\Gamma
      \\
      \auxJQ{\EmpEnv}{\term \colon A}{(\vinch, \voutch)}
    }{
      \vdash (\Defs, \globaldef \F\langle \vinch, \voutch \rangle = \term)
      \colon (\widehat\Gamma, \F \colon (A, |\vinch|, |\voutch|))
    }
  \end{proofrules}
  \caption{Exhaustive typing rules for \QifUnitary{}.
  The sets of input channels $\vinch$ and $\vinchB$,
  and those of output channels $\voutch$ and $\voutchB$, are disjoint.}
  \label{fig:procedural:ver2-ext:typing-rules}
  \Description{}
\end{mathfig}
The complete collection of typing rules for \QifUnitary{} appears in \cref{fig:procedural:ver2-ext:typing-rules}.
Apart from the record of i/o channels, the rules mirror those of \Qif{}.
The quantum conditional $\qif$ enforces that both branches use exactly the same
input and output channel names, whereas every other rule requires the premises
to mention pairwise disjoint channel sets.

\begin{mathfig}
  \begin{proofrules}
    \meas \SynDil
    \lambda x^\qbit. \letx{y \otimes z}{\gCX(x \otimes \inch)}{(\outch \assign z; y)}

    \ket1 \SynDil \inch

    \btrue \SynDil \inch

    \bnot \SynDil \gX

    \infer{
      \term \SynDil \term'
      \\
      \termB_1 \SynDil \termB'_1
      \\
      \termB_2 \SynDil \termB'_2
      \\
      (\vinch_i, \voutch_i) = \text{input/output channels occurring in}\ \termB'_i
      \\
      |\vinchB_i| = |\voutchB_i|
      \\
      \vinch_1\vinchB_1 = \vinch_2\vinchB_2
      \\
      \voutch_1\voutchB_1 = \voutch_2\voutchB_2
    }{
      \cifx{\term}{\termB_1}{\termB_2} \SynDil
      \letx{x \otimes y}{\bigl(
        \qifx{\term'}
        {(\voutchB_1\assign\vinchB_1; \termB'_1)}
        {(\voutchB_2\assign\vinchB_2; \termB'_2)}
        \bigr) }{(\outchC \assign x); y}
    }

    \infer{
      \term \SynDil \term'
    }{
      \letx{x}{\F}{\term} \SynDil
      \letx{x}{\F\langle\vinch, \voutch\rangle}{\term'}
    }

    \infer{}{x \SynDil x}

    \infer{
      \term \SynDil \term'
    }{
      \lambda x.\term \SynDil \lambda x.\term'
    }

    \infer{
      \term \SynDil \term'
      \\
      \termB \SynDil \termB'
    }{
      \term;\termB \SynDil \term';\termB'
    }

    \infer{
      \term \SynDil \term'
      \\
      \termB \SynDil \termB'
    }{
      \term\,\termB \SynDil \term'\,\termB'
    }

    \infer{
      \term \SynDil \term'
      \\
      \termB \SynDil \termB'
    }{
      \term\otimes\termB \SynDil \term'\otimes\termB'
    }

    \infer{
      \term \SynDil \term'
      \\
      \termB \SynDil \termB'
    }{
      \letx{x}{\term}{\termB} \SynDil \letx{x}{\term'}{\termB'}
    }

    \infer{
      \term \SynDil \term'
      \\
      \termB \SynDil \termB'
    }{
      \letx{x\otimes y}{\term}{\termB} \SynDil \letx{x\otimes y}{\term'}{\termB'}
    }

    \infer{}{U\SynDil U}

    \infer{}{\unitval\SynDil\unitval}

    \infer{}{\qifx{\term}{\termB_1}{\termB_2}\SynDil\qifx{\term}{\termB_1}{\termB_2}}

    \infer{
      \Defs \SynDil \widehat\Defs
      \\
      \term\SynDil\term'
    }{
      (\Defs, \globaldef \F = \term) \SynDil 
      (\widehat\Defs,
      \globaldef \F \langle\vinch,\voutch\rangle
      = \term')
    }
  \end{proofrules}
  \caption{Exhaustive rules for syntactic dilation.}
  \label{app:fig:procedural:ver2:syntactic-dilation}
  \Description{}
\end{mathfig}
\Cref{app:fig:procedural:ver2:syntactic-dilation} lists the full syntactic dilation rules,
spelling out the program transformation from \Qif{} to \QifUnitary{}.
Every $\Hilb$-term is left unchanged by dilation, \ie, $\term \SynDil \term$.
All remaining rules that were not presented in \cref{fig:procedural:syntactic-dilation}
simply recurse structurally on the syntax.

\begin{mathfig}
  \begin{proofrules}
    \sem{\inch} = \sigma_{\unit,\qbit} %

    \sem{(\outch \assign \term); \termB}
    = (\sigma_{\qbit,\unit}\sem{\term}) \otimes \sem{\termB}

    \sem{\F\langle \vinch, \voutch \rangle}_\varrho
    = \varrho(\F)

    \dots
  \end{proofrules}
  \caption{Categorical semantics of \QifUnitary{}.}
  \label{fig:procedural:ver2-ext:categorical-semantics}
  \Description{}
\end{mathfig}
The categorical semantics of
$\auxJQ{\Delta}{\term \colon A}{(\vinch, \voutch)}$
is defined by a map
\[
  \textstyle
  \prod_{(\F\colon (B, n, m)) \in \widehat\Gamma}\,
  \Hilb( \CC^{2^n}, \sem{B}_\Hilb \otimes \CC^{2^m} )
  \longrightarrow
  \Hilb(\sem{\Delta} \otimes \CC^{2^{|\vinch|}}, \sem{A} \otimes \CC^{2^{|\voutch|}});
  \quad
  \varrho \longmapsto \sem{\term}_\varrho.
\]
Non-trivial ones are defined
in \cref{fig:procedural:ver2-ext:categorical-semantics}.
For function definitions $\widehat\Defs$, its categorical semantics is defined by
a map that maps
$(\globaldef \F \langle \vinch, \voutch \rangle = M_\F) \in \widehat\Defs$
to the semantics of its body $\sem{M_\F}$.
All the omitted rules are similarly defined as in \Qif{}
(though many swap operators are inserted to move channels behind).
\begin{mathfig}
  \begin{align*}
    \begin{array}{llrl}
      \textit{Frames}
       & D                                 & \Coloneqq &
      \lambda x. [\cdot]
      \sor [\cdot]\,\term
      \sor \term\,[\cdot]
      \sor [\cdot]; \term
      \sor \term; [\cdot]
      \sor [\cdot] \otimes \term
      \sor \term \otimes [\cdot]
      \\ &       & \sor &
      \letx{x}{[\cdot]}{\term}
      \sor \letx{x}{\term}{[\cdot]}
      \sor (\outch \assign [\cdot]); \term
      \sor (\outch \assign \term); [\cdot]
      \\ &       & \sor &
      \letx{x \otimes y}{[\cdot]}{\term}
      \sor \letx{x \otimes y}{\term}{[\cdot]}
      \sor \qifx{[\cdot]}{\termB_1}{\termB_2}
      \\[5pt]
      \textit{Contexts}
       & \ctxt[][]                             & \Coloneqq &
      [\cdot]
      \sor D[\ctxt[][]]
      \sor \qifx{\term}{\ctxt[][]}{\termB}
      \sor \qifx{\term}{\termB}{\ctxt[][]}
      \\[5pt]
      \textit{Ctrl.~Ctxt.}
       & \ctxt[][\emptyset]                & \Coloneqq &
      [\cdot]
      \\
       & \ctxt[][\qctrl]                   & \Coloneqq &
      D[\ctxt[][\qctrl]]
      \\
       & \ctxt[][\qctrl, \inch \mapsto 1] & \Coloneqq &
      \qifx{\inch}{\ctxt[][\qctrl]}{\termB}
      \\
       & \ctxt[][\qctrl, \inch \mapsto 0] & \Coloneqq &
      \qifx{\inch}{\termB}{\ctxt[][\qctrl]}
    \end{array}
  \end{align*}
  \caption{Exhaustive definition of contexts of \QifUnitary{}.
    Here $\qctrl$ is an assignment of $0$ or $1$ to a finite set of input channels.}
  \label{app:fig:procedural:ver2-ext:ctxt}
  \Description{}
\end{mathfig}

\begin{mathfig}
  \raggedright Commuting conversions of output channel: \\
  \begin{proofrules}[0.35cm \lineskiplimit=0.35cm]
    \ctxt[\term\,((\outch \assign \termB); \termC)][]
    \step
    \ctxt[(\outch \assign \termB); (\term\,\termC)][]

    \ctxt[((\outch \assign \termB); \term)\, \termC][]
    \step
    \ctxt[(\outch \assign \termB); (\term\,\termC)][]

    \ctxt[((\outch \assign \termB); \term); \termC][]
    \step
    \ctxt[(\outch \assign \termB); (\term; \termC)][]

    \ctxt[((\outch \assign \termB); \term) \otimes \termC][]
    \step
    \ctxt[(\outch \assign \termB); (\term\otimes\termC)][]

    \ctxt[\term \otimes ((\outch \assign \termB); \termC)][]
    \step
    \ctxt[(\outch \assign \termB); (\term\otimes\termC)][]

    \ctxt[\letx{x}{((\outch \assign \termB); \term)}{\termC}][]
    \step
    \ctxt[(\outch \assign \termB); \letx{x}{\term}{\termC}][]

    \ctxt[\letx{x \otimes y}{((\outch \assign \termB); \term)}{\termC}][]
    \step
    \ctxt[(\outch \assign \termB); \letx{x \otimes y}{\term}{\termC}][]

    \ctxt[\qifx{((\outch \assign \termB); \term)}{\termC_1}{\termC_2}][]
    \step
    \ctxt[(\outch \assign \termB); \qifx{\term}{\termC_1}{\termC_2}][]

    \ctxt[(\outch \assign ((\outchB \assign \termB); \term)); \termC][]
    \step
    \ctxt[(\outchB \assign \termB); (\outch \assign \term); \termC][]
  \end{proofrules}

  \bigskip
  \raggedright Commuting conversions of $\Let\otimes$ and sequencing: \\
  \begin{proofrules}
    \infer{
      x, y \notin \mathrm{fv}(\termB)
    }{
      \ctxt[(\letx{x \otimes y}{\termC}{\term})\,\termB][]
      \step
      \ctxt[\letx{x \otimes y}{\termC}{(\term\,\termB)}][]
    }

    \ctxt[(\termC; \term)\,\termB][]
    \step
    \ctxt[\termC; (\term\,\termB)][]

    \infer{
      x', y' \notin \mathrm{fv}(\termB)
    }{
      \ctxt[\letx{x \otimes y}{(\letx{x' \otimes y'}{\termC}{\term})}{\termB}][]
      \step
      \ctxt[\letx{x' \otimes y'}{\termC}{(\letx{x \otimes y}{\term}{\termB})}][]
    }

    \ctxt[\letx{x \otimes y}{(\termC; \term)}{\termB}][]
    \step
    \ctxt[\termC; (\letx{x \otimes y}{\term}{\termB})][]

    \infer{
      x', y' \notin \mathrm{fv}(\termB)
    }{
      \ctxt[(\letx{x' \otimes y'}{\termC}{\term}); \termB][]
      \step
      \ctxt[\letx{x' \otimes y'}{\termC}{(\term; \termB)}][]
    }

    \ctxt[(\termC; \term); \termB][]
    \step
    \ctxt[\termC; (\term; \termB)][]
  \end{proofrules}

  \bigskip
  \raggedright Non-trivial rules:\\
  \begin{proofrules}
    \ctxt[\unitval; \term][]
    \step
    \ctxt[\term][]

    \ctxt[(\lambda x. \term)\, \termB][]
    \step
    \ctxt[\term[\termB/x]][]

    \ctxt[
    \letx{x}{\termB}{\term}
    ][]
    \step
    \ctxt[
    \term[\termB/x]
    ][]

    \ctxt[
    \letx{x \otimes y}{\termB \otimes \termC}{\term}
    ][]
    \step
    \ctxt[
    \term[\termB/x, \termC/y]
    ][]

    \infer{
      \bigl(\globaldef \F\langle\vinch, \voutch\rangle = \term_\F\bigr)
      \in \widehat\Defs
    }{
      \ctxt[\F\langle\vinchB, \voutchB\rangle][]
      \step_{\widehat\Defs}
      \ctxt[\term_\F[\vinchB/\vinch, \voutchB/\voutch]][]
    }

    \vinchB \mapsto \ket\psi,\,
    \ctxt[U (\inch_1 \otimes \cdots \otimes \inch_n)][\qctrl]
    \step
    \vinchB \mapsto \bigl(C_\qctrl U_{\inch_1,\dots,\inch_n}\bigr) \ket\psi,\,
    \ctxt[\inch_1 \otimes \cdots \otimes \inch_n][\qctrl]

    \ctxt[][] \left[\begin{matrix*}[l]
        \qif\ {\inch}\ \ifthen
        \\ \quad
        (\outch_1 \assign \inchB_{\sigma(1)}); \dots; (\outch_n \assign \inchB_{\sigma(n)});
        \\ \quad
        \val[\inchB_{\sigma(n+1)}, \dots, \inchB_{\sigma(m)}]
        \\ \ifelse
        \\ \quad
        (\outch_1 \assign \inchB_1); \dots; (\outch_n \assign \inchB_n);
        \\ \quad
        \val[\inchB_{n+1}, \dots, \inchB_{m}]
      \end{matrix*}\right]
    \step
    \ctxt[][] \left[\begin{matrix*}[l]
        \Let {x \otimes y_1 \otimes \cdots \otimes y_m} =
        \\
        \quad
        {\mathsf{C}\ceil{\sigma}\,(\inch \otimes \inchB_1 \otimes \cdots \otimes \inchB_m)}
        \\
        \In
        (\outch_1 \assign y_1); \dots; (\outch_n \assign y_n);
        \\
        \quad\ x \otimes \val[y_{n+1}, \dots, y_m]
      \end{matrix*}\right]
  \end{proofrules}
  \caption{Exhaustive reduction rules of \QifUnitary{}.
    The state $\vinch \mapsto \ket\psi$ appears only in the gate rule, the only rule that changes it.
    $\mathrm{fv}(\termB)$ denotes the set of free variables of $\termB$.
    In the last rule,
    we identify $(\outch \assign \inch); (\outchB \assign \inchB); \term$
    with $(\outchB \assign \inchB); (\outch \assign \inch); \term$,
    the variables $x, y_1, \dots, y_m$ are fresh,
    and $\Let\ x \otimes y_1 \otimes \cdots \otimes y_m = \term \In \termB$
    abbreviates the nested $\Let\otimes$.
    For a permutation $\sigma$ of $\{1,\dots,m\}$,
    $\lceil\sigma\rceil$ maps $\ket{i_1 \cdots i_m}$ to $\ket{i_{\sigma(1)} \cdots i_{\sigma(m)}}$, and
    $\mathsf{C}\ceil{\sigma} \defeq \ketbra00 \otimes I + \ketbra11 \otimes \lceil\sigma\rceil$
    is the $\lceil\sigma\rceil$ controlled by the first qubit.
  }
  \label{fig:procedural:ver2-ext:operational-semantics}
  \Description{}
\end{mathfig}
The contexts are defined in \cref{app:fig:procedural:ver2-ext:ctxt},
and the reduction rules in \cref{fig:procedural:ver2-ext:operational-semantics}.
The reduction does not fix an evaluation strategy: any redex may be reduced.
All rules apply in an arbitrary context $\ctxt[][]$, except the gate rule,
the only rule that changes the quantum state.
The gate rule applies only in a controlled context $\ctxt[][\qctrl]$,
that is, a context in which every $\qif$ having the hole in one of its branches
has an input channel $\inch$ as its condition;
$\qctrl$ maps each such $\inch$ to $1$ or $0$ according to the branch containing the hole,
and a gate applied at the hole acts as the $\qctrl$-controlled gate.

The nine commuting conversions move an output assignment $(\outch\assign\termB)$ outward.
None of them moves an assignment out of a branch of $\qif$;
the assignments in the two branches are aligned by the last rule.

The six commuting conversions of $\Let\otimes$ and sequencing move a $\Let\otimes$-binding of $\termC$
or a sequencing $\termC; {-}$
out of the function position of an application (the first two)
and out of the term bound by a $\Let\otimes$ or a sequencing (the other four),
keeping the left-to-right order of $\termC$, $\term$, and $\termB$.
An output assignment $(\outch \assign \termC); \term$ is the sequencing of $\outch\,\termC \colon \unit$,
and the conversions that move it in the same way as $\termC; {-}$
are the second, third, and seventh commuting conversions above.
We identify terms up to $\alpha$-renaming of bound variables.
The side conditions ensure that the moved binder does not capture a free variable of $\termB$;
they can always be satisfied by renaming its variables to fresh ones.
Since the free variables of a well-typed term are exactly those in its typing context,
they also ensure that the typing context of the right-hand side is well formed;
both sides have the same typing context, type, and i/o channels.
None of these rules moves a subterm into or out of a branch of $\qif$,
so a subterm is at the hole of a controlled context $\ctxt[][\qctrl]$ after the conversion
if and only if it is so before, with the same $\qctrl$.
These conversions expose $\beta$-redexes blocked by a binding
whose bound term does not reduce to a tensor or to $\unitval$, typically a variable:
for a variable $z$ bound outside,
$(\letx{x \otimes y}{z}{\lambda w. \term})\,\termB$
reduces to $\letx{x \otimes y}{z}{((\lambda w. \term)\,\termB)}$,
which has a $\beta$-redex.
No such conversion is needed for $\letx{x}{\termC}{\term}$, which is always a $\beta$-redex.

The first four non-trivial rules are the $\beta$-rules of the linear $\lambda$-calculus;
the substituted term need not be a value.
It is duplicated only when the variable is shared by the two branches of a $\qif$.
The last rule eliminates $\qif$ once both branches are normal forms
that differ only by a permutation $\sigma$ of their input channels:
the controlled permutation $\mathsf{C}\ceil\sigma$ moves the qubits of the then-branch
to the positions used by the else-branch.
The remaining rules are explained in \cref{sec:op-sem:op-sem}.

\clearpage

\section{Theorems and Proofs}
\label{app:sec:theorems}

This appendix proves the theorems stated in the main text, together with some further results.
After defining the $\CPM$-semantics of \QifUnitary{} terms,
\cref{app:sec:theorems-list} lists all the theorems:
those on syntactic dilation,
those on the operational semantics of \QifUnitary{},
those on the extension with divergence, including full abstraction,
and the semantics of $\SWITCH$, together with QC-QC inclusion.
The subsequent subsections prove them in this order.

The following definition of $\CPM$-semantics of \QifUnitary{} terms
is used in the statements of the theorems and their proofs.

\begin{definition}
  We define the $\CPM$-semantics of a \QifUnitary{} term
  $\csem{\auxJQ{\Delta}{\term \colon A}{(\vinch, \voutch)}}$
  as
  \[
    \csem{\term} \ = \
    (\ident_{A} \otimes \discardDiag_{\voutch}) \circ
    \EmbeddingFunctor\bigl(\sem{\term}\bigr) \circ
    (\ident_{\Delta} \otimes \ketbra{1}{1}_{\vinch})
    .
    \tag*{\thmend}
  \]
\end{definition}

\subsection{Theorems}
\label{app:sec:theorems-list}

\begin{theorem}[Type safety of syntactic dilation]
  \label{thm:syn-dil:type-safety}
  Assume
  $\Gamma \SynDil \widehat\Gamma$,
  $\Delta \SynDil \Delta'$ and $A \SynDil A'$.
  Then,
  for any term $\termJG{\Delta}{\term \colon A}$ in \Qif{},
  there exists a \QifUnitary{} term
  $\term'$ and $(\vinch, \voutch)$ such that
  $\term \SynDil \term'$ and
  $\auxJQ{\Delta'}{\term' \colon A'}{(\vinch, \voutch)}$.
  \thmend
\end{theorem}

\begin{theorem}[Semantic preservation of syntactic dilation]
  \label{thm:syn-dil:sem-pres}
  For each $\CPM$-term $\term$ whose types of free variables and whose own type are first-order,
  if $\term \SynDil \term'$ and $(\vinch, \voutch)$
  are the free i/o channels in $\term'$,
  \[
    \sem{\term} =
    r_A \circ \csem{\term'} \circ s_{\Delta}
    \tag*{\thmend}
  \]
\end{theorem}

\begin{theorem}[Type safety]
  \label{thm:opsem:type-sefety}
  If $C$ is a configuration and $C \step_\Defs C'$,
  then $C'$ is also a configuration of the same type.
  \thmend
\end{theorem}

\begin{theorem}[Progress]
  \label{thm:opsem:progress}
  For each configuration $C \defeq (\vinch \mapsto \ket\psi, \term)$
  such that $\auxJQ{\EmpEnv}{\term\colon A}{(\vinchB, \voutch)}$
  and definitions $\Defs \colon \Gamma$,
  there exists a configuration $C'$ such that
  $C \step_\Defs C'$,
  or $\term$ is of the form
  $(\voutch \assign \vinch); \val$.
  \thmend
\end{theorem}

\begin{theorem}[Termination]
  \label{thm:opsem:termination}
  The underlying reduction rules are strongly normalising on every
  well-typed term, including terms with ordinary free variables.
  In particular, there is no infinite sequence of configuration reductions
  under $\Defs$, and every irreducible configuration has the value form
  $(\voutch \assign \vinch); \val$.
  \thmend
\end{theorem}

\begin{theorem}[Semantic preservation]
  \label{thm:opsem:sem-pres}
  If $(\ket\psi, \term) \step_\Defs (\ket\phi, \termB)$,
  then
  $\sem{\term}_{\eval}\ket\psi = \sem{\termB}_{\eval}\ket\phi$
  where $\eval(\F) \defeq \sem{\term_\F}$
  for each
  $(\globaldef \F\langle\vinch,\voutch\rangle \colon \term_\F) \in \Defs$.
  \thmend
\end{theorem}

\begin{corollary}[Adequacy]
  \label{cor:adequacy}
  For any \Qif{} function definitions $\Defs\colon\Gamma$ and any term
  $\termJC{\EmpEnv}{\term:A}$ such that
  $\FirstOrderType{A}$,
  the term $\term$ evaluates to $\sem{\term}_{\sem{\Defs}}$.
  That is,
  \begin{itemize}
    \item there exist \QifUnitary{}
          definitions $\widehat\Defs$,
          non-linear context $\widehat\Gamma$,
          and a term $\term'$
          such that
          $\Defs \SynDil \widehat\Defs$,
          $\Gamma \SynDil \widehat\Gamma$, and
          $\term \SynDil \term'$,
    \item they admit type derivations
          $\widehat\Defs \colon \widehat\Gamma$
          and
          $\auxJQ{\EmpEnv}{\term'\colon A[\qbit/\bool]}{(\vinch, \voutch)}$,
    \item the term $\term'$ evaluates to $\ket\varphi$ as
          $(\vinch \mapsto \ket{1 \dots 1}, \term') \step^*_{\widehat\Defs}
            \bigl(
            \vinch \mapsto \ket\varphi
            ,
            (\outch_1\assign\inch_{\sigma(1)});
            \dots;
            (\outch_m\assign\inch_{\sigma(m)});
            \val[\inch_{\sigma(m+1)}, \dots, \inch_{\sigma(n)}]
            \bigr)
          $, and
    \item by reordering the qubits and by tracing out the auxiliary qubits,
          we obtain the resulting state:
          \begin{align*}
            \trace_{{(\CC^2)}^{\otimes m}}(
            \ceil{\sigma}\ketbra\varphi\varphi\ceil{\sigma^{-1}}
            )
            \quad=\quad
            s_A \circ \sem{\term}_{\sem{\Defs}}.
            \tag*{\thmend}
          \end{align*}
  \end{itemize}
\end{corollary}

\begin{mathfig}
  \begin{proofrules}
    \textit{Terms}\hspace{1.5em}
    \term,\termB\ \Coloneqq\ \cdots \sor \abort

    \termJC[\EmpEnv]{\EmpEnv}{\abort\colon\unit}

    \sem{\,\abort\,}\ =\ 0
  \end{proofrules}
  \caption{
    Additional rules for \Qif{} with divergence.
    (Left): new syntax,
    (Middle): new typing rule,
    (Right): its categorical semantics.
  }
  \label{app:fig:rules:qif-bra}
  \Description{}
\end{mathfig}
\begin{mathfig}
  \begin{proofrules}
    {\textit{Output test channels}}
    \quad
    {\testch, \testchB, \testchC, \dots}

    \textit{Terms}\hspace{1.5em}
    \term,\termB\ \Coloneqq\ \cdots
    \sor (\testch \assign \term); \termB

    \textit{Type Judgements}\hspace{1.5em}
    \auxJQ{\Delta}{\term \colon A}{(\vinch, \voutch, \vtestch)}

    \infer*[]{
      \auxJQ{\Delta}{\term \colon \qbit}{(\vinch, \voutch, \vtestchB)}
      \\
      \auxJQ{\Delta'}{\termB \colon A}{(\vinchB, \voutchB, \vtestchC)}
    }{
      \auxJQ{\Delta,\Delta'}{(\testch \assign \term); \termB \colon A}{(\vinch\vinchB, \voutch\voutchB, \testch\vtestchB\vtestchC)}
    }

    \abort \SynDil (\kappa \assign \alpha); \unitval

    \csem{\,\auxJQ{\Delta}{\term \colon A}{(\vinch, \voutch, \vtestch)}\,}
    \ =\
    \bigl(
    \ident_{A} \otimes \discardDiag_{\voutch} \otimes
    \EmbeddingFunctor\bra{0}_\vtestch
    \bigr) \circ
    \EmbeddingFunctor\sem{\term} \circ
    \bigl(\ident_{\Delta} \otimes \EmbeddingFunctor\ket{1}_{\vinch}\bigr)
  \end{proofrules}
  \caption{
    Additional rules for \QifUnitary{} extended with output test channels.
  }
  \label{app:fig:rules:qifunitary-bra}
  \Description{}
\end{mathfig}
\begin{theorem}[Language with divergence]
  \label{thm:extended-lang}
  In the language \Qif{}, we add $\abort \colon \unit$,
  which causes execution to diverge,
  as in \cref{app:fig:rules:qif-bra}.
  Also, in the language \QifUnitary{}, we add another distinct kind of
  output channel,
  which we call an \emph{output test channel},
  as in \cref{app:fig:rules:qifunitary-bra}.
  Its syntactic dilation is defined there as well.
  Their denotational and operational semantics are defined in the same way as
  those of ordinary output channels.
  The $\CPM$-semantics $\csem{-}$ is also defined in
  \cref{app:fig:rules:qifunitary-bra},
  where we apply $\bra{0}$ to the qubit outcomes
  from output test channels.
  Every theorem stated above still holds for these extended languages.
  \thmend
\end{theorem}

\begin{theorem}[Full abstraction of \Qif{} with divergence]
  \label{thm:CPM-full-abst}
  In \Qif{} extended with divergence,
  the denotational semantics is fully abstract,
  \ie,
  for each $\JG = \JQ$ or $\JC$,
  let
  $\termJG{\Delta}{\term \colon A}$
  and
  $\termJG{\Delta}{\termB \colon A}$.
  Then,
  $\sem{C[\term]} = \sem{C[\termB]}$ for any context $C[\cdot]$
  of $\unit$ type,
  if and only if
  $\sem{\term} = \sem{\termB}$.
  \thmend
\end{theorem}

\begin{theorem}[Full abstraction of \QifUnitary{} with output test channels]
  \label{thm:Hilb-full-abst}
  In \QifUnitary{} extended with output test channels,
  the denotational semantics is fully abstract,
  \ie,
  let
  $\auxJQ{\Delta}{\term \colon A}{(\vinch, \voutch, \vtestch)}$
  and
  $\auxJQ{\Delta}{\termB \colon A}{(\vinch, \voutch, \vtestch)}$.
  Then,
  $\csem{C[\term]} = \csem{C[\termB]}$ for any context $C[\cdot]$ of
  $\unit$ type,
  if and only if
  $\csem{\term} = \csem{\termB}$.
  \thmend
\end{theorem}

\begin{theorem}
  [Restatement of \cref{thm:cat:switch-semantics}]
  \label{app:thm:cat:switch-semantics}
  The categorical semantics of $\SWITCH$ coincides with the quantum SWITCH.
  \thmend
\end{theorem}

\begin{theorem}[QC-QC inclusion]\label{thm:qcqc-core}
  Consider \Qif{} without the extension by $\abort$.
  Let $\Defs\colon\Gamma$ be well-typed global definitions and
  \[
    \termJC{f_1 \colon A_1\li B_1,\ldots,f_n \colon A_n\li B_n}{\term\colon X\li Y},
  \]
  where $A_i,B_i,X,Y$ are first order and purely quantum.
  Put $I_i=\sem{A_i}$, $O_i=\sem{B_i}$,
  $P=\sem X$, and $F=\sem Y$.
  Let $\mathcal E_i \colon I_i \longrightarrow O_i$ be quantum channels.
  Then for the quantum channel
  \[
    \mathsf S_M(\mathcal E_1,\ldots,\mathcal E_n)
    \defeq
    \sem M_{\sem\Defs}
    (\ulcorner\mathcal E_1\urcorner\otimes\cdots\otimes
    \ulcorner\mathcal E_n\urcorner)
    \colon P \longrightarrow F,
  \]
  there is a QC-QC~\cite{Wechs2021} with slots $\{1,\ldots,n\}$,
  slot $i$ having input space $I_i$ and output space $O_i$,
  global past $P$, and global future $F$ such that, for all CP maps
  $\mathcal E_1,\ldots,\mathcal E_n$, inserting them into its slots
  induces $\mathsf S_M(\mathcal E_1,\ldots,\mathcal E_n)$.
  \thmend
\end{theorem}

\subsection{Proof of Type Safety of Syntactic Dilation}

\begin{lemma}[Balancing lemma]
  \label{app:lem:syndil:balancing}
  For each \QifUnitary{} type derivation
  $\auxJQ{\Delta}{\term \colon A}{(\vinch, \voutch)}$,
  \[
    |\Delta| + |\vinch| = |A| + |\voutch|
  \]
  where $|\Delta| \defeq \sum_{(x \colon A) \in \Delta} |A|$.
\end{lemma}
\begin{proof}
  By induction on the derivation.
  We prove only the non-trivial cases.
  \begin{description}
    \item[Case $\inch$.]
          Let the derivation be
          $\auxJQ{\EmpEnv}{\inch \colon \qbit}{(\inch, \cdot)}$.
          Then
          $|\cdot| + |\alpha| = 1 = |\qbit| + |\cdot|$.
    \item[Case $(\outch\assign\term);\termB$.]
          Assume that we derived
          $\auxJQ{\Delta,\Delta'}
            {(\outch \assign \term); \termB \colon A}
            {(\vinch\vinchB, \outch\voutchB\voutchC)}$
          from
          $\auxJQ{\Delta}
            {\term \colon \qbit}
            {(\vinch, \voutchB)}$
          and
          $\auxJQ{\Delta'}
            {\termB \colon A}
            {(\vinchB, \voutchC)}$.
          Then
          $|\Delta,\Delta'| + |\vinch\vinchB|
            = (|\Delta| + |\vinch|) + (|\Delta'| + |\vinchB|)
            = (|\qbit| + |\voutchB|) + (|A| + |\voutchC|)
            = |A| + |\outch\voutchB\voutchC|$.
    \item[Case $\lambda x.\term$.]
          Assume that we derived
          $\auxJQ{\Delta}{\lambda x. \term \colon A \li B}{(\vinch, \voutch)}$
          from
          $\auxJQ{\Delta, x \colon A}{\term \colon B}{(\vinch, \voutch)}$.
          Then, by induction hypothesis,
          $|\Delta| + |A| + |\vinch| = |B| + |\voutch|$.
          Therefore, $|\Delta| + |\vinch| = |B| + |\voutch| - |A| = |A \li B| + |\voutch|$.
    \item[Case $\term\,\termB$.]
          Assume that we derived
          $\auxJQ{\Delta, \Delta'}{\term\,\termB \colon B}{(\vinch\vinchB, \voutch\voutchB)}$
          from
          $\auxJQ{\Delta}{\term \colon A \li B}{(\vinch, \voutch)}$
          and
          $\auxJQ{\Delta'}{\termB \colon A}{(\vinchB, \voutchB)}$.
          Then, by induction hypothesis,
          $|\Delta| + |\vinch| = |A \li B| + |\voutch|$
          and
          $|\Delta'| + |\vinchB| = |A| + |\voutchB|$.
          Therefore,
          $|\Delta, \Delta'| + |\vinch\vinchB|
            = |A \li B| + |A| + |\voutch\voutchB|
            = |B| - |A| + |A| + |\voutch\voutchB|
            = |B| + |\voutch\voutchB|
          $.
    \item[Case $\letx{x}{\term}{\termB}$.]
          $\letx{x}{\term}{\termB}$ can be identified with
          $(\lambda x. \termB) \term$.
    \item[Case $\letx{x\otimes y}{\term}{\termB}$.]
          This rule is essentially the same as the previous case.
    \item[Case $\qifx{\term}{\termB_1}{\termB_2}$.]
          Assume that we derived
          $\auxJQ{\Delta, \Delta'}
            {(\qifx{\term}{\termB_1}{\termB_2}) \colon \qbit \otimes A}
            {(\vinch\vinchB, \voutch\voutchB)}$
          from
          $\auxJQ{\!\Delta}{\term \colon \qbit}{(\vinch, \voutch)}$,
          $\auxJQ{\!\Delta'}{\termB_1 \colon A}{(\vinchB, \voutchB)}$,
          and
          $\auxJQ{\!\Delta'}{\termB_2 \colon A}{(\vinchB, \voutchB)}$.
          Then, by induction hypothesis,
          $|\Delta| + |\vinch| = 1 + |\voutch|$
          and
          $|\Delta'| + |\vinchB| = |A| + |\voutchB|$.
          Therefore,
          $|\Delta\Delta'| + |\vinch\vinchB|
            = 1 + |A| + |\voutch\voutchB|
            = |\qbit \otimes A| + |\voutch\voutchB|$.
    \item[Case $\F$.]
          Assume that
          $(\F \colon (A, n, m)) \in \Gamma$,
          $|\vinch| = n$ and
          $|\voutch| = m$ hold,
          and we derived
          $\auxJQ{\EmpEnv}
            {\F\langle \vinch, \voutch \rangle \colon A}
            { (\vinch, \voutch) }$.
          From the definition of non-linear context,
          we know $|A| = n - m$.
          Therefore,
          $|\vinch| = n = (n - m) + m = |A| + |\voutch|$.
          \qedhere
  \end{description}
\end{proof}

\begin{proof}[Proof of \Cref{thm:syn-dil:type-safety}]
  First, we observe that all $\Hilb$-terms remain unchanged
  after translation.
  Therefore, we only need to check $\CPM$-terms.
  We check the following selected rules.
  The rest are straightforward.
  \begin{description}
    \item[Case $\meas$.]
          The derivation before translation is
          $\termJC{\EmpEnv}{\meas \colon \qbit \li \bool}$,
          and after translation,
          we have the derivation
          $\auxJQ{\EmpEnv}{
              \lambda x^\qbit. \letx{y \otimes z}{\gCX(x \otimes \inch)}{(\outch \assign z; y)} \colon \qbit \li \qbit
            }{(\inch, \outch)}$.
    \item[Case $\ket1$.]
          The term
          $\termJC{\EmpEnv}{\ket1 \colon \qbit}$
          is translated to
          $\auxJQ{\EmpEnv}{\inch \colon \qbit}{(\inch,\cdot)}$.
    \item[Case $\btrue$.]
          The term
          $\termJC{\EmpEnv}{\btrue \colon \bool}$
          is translated to
          $\auxJQ{\EmpEnv}{\inch \colon \qbit}{(\inch,\cdot)}$.
    \item[Case $\bnot$.]
          The term
          $\termJC{\EmpEnv}{\bnot \colon \bool \li \bool}$
          is translated to
          $\auxJQ{\EmpEnv}{\gX \colon \qbit \li \qbit}{(\cdot, \cdot)}$.
    \item[Case $\cifx{\term}{\termB_1}{\termB_2}$.]
          We assume the following.
          \begin{itemize}
            \item $\termJC{\Delta_1}{\term \colon \bool}
                    \SynDil
                    \auxJQ{\Delta'_1}{\term' \colon \qbit}{(\vinch, \voutch)}$
            \item $\termJC{\Delta_2}{\termB_1 \colon A}
                    \SynDil
                    \auxJQ{\Delta'_2}{\termB'_1 \colon A'}{(\vinchB_1, \voutchB_1)}$
            \item $\termJC{\Delta_2}{\termB_2 \colon A}
                    \SynDil
                    \auxJQ{\Delta'_2}{\termB'_2 \colon A'}{(\vinchB_2, \voutchB_2)}$
          \end{itemize}
          From the previous \cref{app:lem:syndil:balancing},
          $|A'| - |\Delta'_2|
            = |\vinchB_1| - |\voutchB_1| = |\vinchB_2| - |\voutchB_2|$.
          Therefore, we can choose fresh $\vinchC$ and $\voutchC$ such that
          $|\voutchC| - |\vinchC|$ is equal to this quantity.
          The term
          $\qifx{\term'}
            {(\voutchB_2\voutchC \assign \vinchB_2\vinchC; \termB'_1)}
            {(\voutchB_1\voutchC \assign \vinchB_1\vinchC; \termB'_2)}
          $
          has the type $\qbit \otimes A' \withaug (\vinch\vinchB_1\vinchB_2\vinchC, \voutch\voutchB_1\voutchB_2\voutchC)$
          in \QifUnitary,
          and hence the translated term, which discards the first component to a fresh $\outchD$,
          has the type $A' \withaug (\vinch\vinchB_1\vinchB_2\vinchC, \voutch\voutchB_1\voutchB_2\voutchC\outchD)$.
          \qedhere
  \end{description}
\end{proof}

\subsection{Proof of Semantic Preservation of Syntactic Dilation}
\newcommand{\catA}{\mathcal{A}}
\newcommand*{\Dil}{{T(\mathcal{A})}}

For the proof, we define another category $\Dil$.
We first define a compact closed category $\catA$,
and obtain $\Dil$ using \emph{focussed orthogonality}
by \citet{HYLAND2003183}.

\begin{definition}
  The category $\catA$ consists of the following:
  An object of $\catA$ is a sequence of natural numbers $\vec{n}$,
  and a morphism $\catA(\vec{n},\vec{m})$ is a pair $(f, f')$
  of morphisms in $\CPM$ of type
  $f\colon \vec{n} \longrightarrow \vec{m}$,
  and $f' \colon (\sum\vec{n}) \longrightarrow (\sum\vec{m})$.
  \thmend
\end{definition}

This category inherits the monoidal structure of $\CPM$:
Observe that the sum of the elements of the sequence $\vec{n} \otimes \vec{m}$
is $\sum\vec{n} \times \sum\vec{m}$.
Moreover, $\catA$ is compact closed.

\begin{definition}
  For each set $c \subseteq \catA(I, A)$,
  we define $c^* \subseteq \catA(A, I)$ by
  $\{ m \colon A \to I \mid
    \forall x \in c, \exists r \in \Real_{\geq 0}, m\circ x = (r,r) \}$.
  We call $c$ closed if $c^{**} = c$.

  The category $\Dil$ consists of the following data:
  An object of $\Dil$ is a pair $(\vec{n}, c)$
  such that $c \subseteq \catA(I, \vec{n})$ is closed.
  A morphism $f \in \Dil((\vec{n},c_{\mathbf{n}}), (\vec{m},c_{\mathbf{m}}))$
  is a morphism in $\catA$ such that
  for each $x \in c_\mathbf{n}$, $f \circ x \in c_\mathbf{m}$.
  \thmend
\end{definition}

This category $\Dil$ is the tight category defined by \citet{HYLAND2003183}
obtained from the focussed orthogonality defined by
$F = \{ (r,r) \mid r \in \Real_{\geq 0} \} \subseteq \catA((1), (1))$.
Therefore, $\Dil$ is $\ast$-autonomous,
and two forgetful functors $\Dil \longrightarrow \CPM$
preserve all $\ast$-autonomous structures.
The monoidal products and duals are defined as follows:
\begin{align*}
  (\vec n, c_\mathbf{n})
  \otimes
  (\vec m, c_\mathbf{m})
   &
  =
  (\vec n \otimes \vec m, {\{
                           x \otimes y \mid x \in c_\mathbf{n}, y \in c_\mathbf{m}
                           \}}^{**})
  \\
  (\vec n, c_\mathbf{n})^*
   &
  =
  (\vec n, \{ x^* \colon \vec n \to I \mid x \in c_\mathbf{n} \}^*)
  .
\end{align*}
In particular, we can interpret linear lambda calculus in $\Dil$.
To interpret \Qif{} in $\Dil$, we give some more definitions.

\begin{definition}
  We define $\sem{\bool}_\Dil$ by $((1,1), c_\bool)$
  and $\sem{\qbit}_\Dil$ by $((2), c_\qbit)$
  where $c_\bool$ and $c_\qbit$ are defined as
  \begin{align*}
    c_\bool
    \defeq
    \{ (\rho, s_{(1,1)}(\rho)) \mid p, q\in \Real_{\geq 0}, \rho = (p,q) \},
    \qquad
    c_\qbit
    \defeq \{ (\rho, \rho) \mid \rho \in \CPM(I, (2)) \}.
  \end{align*}
  These $c_\bool$ and $c_\qbit$ are closed.
  For other types $A$ in \Qif, we recursively define $\sem{A}_\Dil$
  using the monoidal closed structure of $\Dil$.
  Note that, for any type $A$,
  $\sem{A}_{\Dil}$ is given by
  $(\sem{A}_{\CPM}, c_A)$ for some $c_A$.
  \thmend
\end{definition}

\begin{lemma}
  Let $\sem{A}_{\Dil} = (\sem{A}_{\CPM}, c_A)$.
  For any first-order type $A$,
  $c_A = \{ (x,s_A \circ x) \mid x \colon I \to \sem{A}_{\CPM} \}$
  and $c^*_A = \{ (y \circ s_A, y)
    \mid y \colon \sem{A[\qbit/\bool]}_{\CPM} \to I \}$.
\end{lemma}
\begin{proof}
  We first prove such $c_A$ are closed.
  For any $\vec{n} \in \CPM$, let
  $c = \{(x, s_{\vec{n}} \circ x) \mid x \in \CPM(I, \vec{n})\}$.
  If $(m,m') \in c^*$, then $m\circ x = m' \circ s\circ x$ for all $x$.
  Since $I$ generates $\CPM$, $m = m' \circ s$.
  Conversely, any $(m' \circ s, m')$ is in $c^*$.
  Thus $c$ is closed.

  Let
  $c_A = \{(x_A, s_{A} \circ x_A) \mid x_A \in \CPM(I, \vec n)\}$
  and
  $c_B = \{(x_B, s_{B} \circ x_B) \mid x_B \in \CPM(I, \vec m)\}$
  for some objects $A$ and $B$.
  Then $(y, y') \in c_{A\otimes B}^*$
  if and only if
  $y \circ (x_A \otimes x_B)
    = y' \circ s_{A \otimes B} \circ (x_A \otimes x_B)$.
  Since $x_A \otimes x_B$ are jointly epic in $\CPM$,
  $y = y' \circ s$.
  Therefore,
  $c_{A\otimes B}^* = \{ (y' \circ s_{A \otimes B}, y')
    \mid y' \colon (\sum \vec n \otimes \sum \vec m) \to I \}
  $.

  Since $\bool$ and $\qbit$ satisfy the claim,
  and since we have shown the claim holds for the tensor type
  $A \otimes B$ if $A$ and $B$ satisfy it,
  this proves that the claim holds for all first-order types.
\end{proof}

\begin{lemma}\label{app:lem:dil:purely-quantum}
  For any purely quantum type $A$,
  $c_A = \{ (x, x) \mid x \colon I \to \sem{A}_{\CPM} \}$.
\end{lemma}
\begin{proof}
  It is easy to prove that such $c_A$ are closed
  and that the property holds for $A \otimes B$ and $A^*$
  if this property holds for $A$ and $B$.
  The claim follows by induction on the structure of types.
\end{proof}

\begin{lemma}\label{app:lem:orthogonality-vec-space}
  Every closed set $c_A$ is a $\Real_{\geq 0}$-module, \ie,
  if $(x, x'), (y, y') \in c_A$,
  $(rx + sy, rx' + sy') \in c_A$ for any $r, s \in \Real_{\geq 0}$.
\end{lemma}
\begin{proof}
  If $(x, x')$ and $(y, y')$ are in $c_A$,
  then for each $(m, m')$ in $c_A^*$,
  $mx = m'x'$ and $my = m'y'$.
  Since
  $m \circ (rx + sy) = r \cdot mx + s \cdot my = r \cdot m'x' + s \cdot m'y' = m' \circ (rx' + sy')$,
  it follows $(rx + sy, rx' + sy') \in c_A$.
\end{proof}

\begin{mathfig}
  \begin{proofrules}
    \sem{\termJQ{\Delta}{\term\colon A}} = \sem{\term}_\CPM

    \sem{\ket{1}} = \ketbra{1}{1}

    \sem{\meas} = s_{(1,1)} \circ m

    \sem{\btrue} = \ketbra{1}{1}

    \sem{\bnot} = \gX

    \sem{\cifx{\term}{\termB_1}{\termB_2}}
    = (\discardDiag_{(1,1)}\otimes\ident_A)d_{\oplus_\CPM}^{-1}(\sem{\termB_2} \oplus_\JC \sem{\termB_1})d_{\oplus_\CPM}(m\sem{\term}\otimes \ident_{\Delta'})

    \dots
  \end{proofrules}
  \caption{Second projection to $\CPM$ of the categorical semantics of \Qif{} in $\Dil$.}
  \label{app:fig:dil-cat-sem}
  \Description{}
\end{mathfig}
We define a categorical semantics of \Qif{} in $\Dil$
so that the first projection coincides with the semantics in $\CPM$.
We sketch the second projection of the semantics in $\Dil$ in \cref{app:fig:dil-cat-sem}
and denote it by $\sem{\term}_{\Dil}$.
Thus, the semantics of $\term$ in $\Dil$ is defined by
$(\sem{\term}_\CPM, \sem{\term}_{\Dil})$.
Here, the semantics uses a valuation map
$\rho \colon \F \longmapsto \rho({\F}) \in \Dil(I, \sem{A})$.
The rules for the standard linear lambda calculus are omitted.
For the rule of classical conditional branching,
one can show that this defines a map in $\Dil$
using \cref{app:lem:orthogonality-vec-space}.

\begin{lemma}
  The semantics $\sem{\cifx{\term}{\termB_1}{\termB_2}}$
  defined in \cref{app:fig:dil-cat-sem} defines a map in
  $\Dil$.
\end{lemma}
\begin{proof}
  For simplicity, we prove the case
  in which the conditional term is a variable:
  \[
    \infer*[]{
      \Gamma \vdash \termB_1 \colon A\\
      \Gamma \vdash \termB_2 \colon A
    }{
      x \colon \bool, \Gamma \vdash \cifx{x}{\termB_1}{\termB_2} \colon A
    }.
  \]
  Let $\left((p,q), \begin{pmatrix}
      p & 0 \\0&q
    \end{pmatrix}\right) \in c_\bool$,
  and $(x,y) \in c_{\sem{\Gamma}}$.
  It suffices to prove that
  the pair
  \begin{align*}
    \left(
    \sem{\cifx{x}{\termB_1}{\termB_2}}_\CPM
    \circ \left((p,q) \otimes x \right),
    \quad
    \sem{\cifx{x}{\termB_1}{\termB_2}}_{\Dil}
    \circ \left(\begin{pmatrix}
                  p & 0 \\ 0 & q
                \end{pmatrix} \otimes y \right)
    \right)
  \end{align*}
  is in $c_A$.
  These elements can be calculated as follows:
  \begin{align*}
     &
    \sem{\cifx{x}{\termB_1}{\termB_2}}_{\CPM}
    \circ \left((p,q) \otimes x \right)
    \\ & =
    (\discardDiag_{(1,1)}\otimes\ident_A)d_{\oplus_\CPM}^{-1}(\sem{\termB_2}_{\CPM} \oplus_\JC \sem{\termB_1}_{\CPM})d_{\oplus_\CPM}
    \circ \left( (p, q) \otimes x \right)
    \\ & =
    q \cdot
    \left(
    \sem{\termB_1}_{\CPM} \circ x
    \right)
    \ +\
    p \cdot
    \left(
    \sem{\termB_2}_{\CPM} \circ x
    \right)
    ,
    \\\\
     &
    \sem{\cifx{x}{\termB_1}{\termB_2}}_{\Dil}
    \circ \left(\begin{pmatrix}
                  p & 0 \\ 0 & q
                \end{pmatrix} \otimes y \right)
    \\ & =
    (\discardDiag_{(1,1)}\otimes\ident_A)d_{\oplus_\CPM}^{-1}(\sem{\termB_2}_{\Dil} \oplus_\JC \sem{\termB_1}_{\Dil})d_{\oplus_\CPM}(m \otimes \ident_{\Delta'})
    \circ \left(\begin{pmatrix}
                  p & 0 \\ 0 & q
                \end{pmatrix} \otimes y \right)
    \\ & =
    (\discardDiag_{(1,1)}\otimes\ident_A)d_{\oplus_\CPM}^{-1}(\sem{\termB_2}_{\Dil} \oplus_\JC \sem{\termB_1}_{\Dil})d_{\oplus_\CPM}
    \circ \left( (p, q) \otimes y \right)
    \\ & =
    q \cdot
    \left(
    \sem{\termB_1}_{\Dil} \circ y
    \right)
    \ +\
    p \cdot
    \left(
    \sem{\termB_2}_{\Dil} \circ y
    \right)
    .
  \end{align*}
  Since for each $i = 1, 2$, the pair
  $
    \bigl(
    \sem{\termB_i}_{\CPM} \circ x,
    \sem{\termB_i}_{\Dil} \circ y
    \bigr)
  $
  belongs to $c_A$,
  by \cref{app:lem:orthogonality-vec-space},
  the linear combination is also in $c_A$.
\end{proof}

We are now ready to prove the preservation of semantics.

\begin{proof}[Proof of \Cref{thm:syn-dil:sem-pres}]
  For any term $\termJG{\Delta}{\term \colon A}$
  and $\auxJQ{\Delta'}{\term' \colon A'}{(\vinch, \voutch)}$
  such that $\term \SynDil \term'$,
  we prove the second projection of
  the semantics of $\term$
  in $\Dil$ coincides with $\csem{\term'}$:
  \begin{equation}
    \sem{\term}_\Dil
    =
    \csem{\term'}.
    \label{app:eq:syn-dil:second-projection}
  \end{equation}
  We prove the non-trivial cases.
  \begin{description}
    \item[Case $\meas$.]
          \begin{align*}
             &
            (\ident \otimes \discardDiag_\outch) \circ \EmbeddingFunctor\sem{
                                                                          \lambda x^\qbit. \letx{y \otimes z}{\gCX(x \otimes \inch)}{(\outch \assign z; y)}
                                                                        } \circ \ketbra{1}{1}_{\inch}
            \\ & =
            \Lambda_{(2)}
            \bigl(
            (\ident_{(2)} \otimes \discardDiag_{(2)})
            \circ \gCX \circ
            (\ident_{(2)} \otimes \ketbra{1}{1}_{\inch})
            \bigr)
            \\ & =
            \Lambda (s\circ m).
          \end{align*}
    \item[Case $\ket1$.] $\ketbra{1}{1} = \EmbeddingFunctor\sem{\inch} \circ \ketbra11_\alpha$.
    \item[Case $\btrue$.] $\ketbra{1}{1} = \EmbeddingFunctor\sem{\inch} \circ \ketbra11_\alpha$.
    \item[Case $\bnot$.] By definition.
    \item[Case $\cifx{\term}{\termB_1}{\termB_2}$.] We assume the following:
          \begin{itemize}
            \item $\sem{\term}_\Dil
                    = \csem{\term'}
                    = (\ident \otimes \discardDiag_{\voutch}) \circ \EmbeddingFunctor\sem{\term'}
                    \circ (\ident \otimes \ketbra{1}{1}_{\vinch})$,
            \item $\sem{\termB_1}_\Dil
                    = \csem{\termB_1'}
                    = (\ident \otimes \discardDiag_{\voutchB_1})
                    \circ \EmbeddingFunctor\sem{\termB'_1}
                    \circ (\ident \otimes \ketbra{1}{1}_{\vinchB_1})$,
            \item $\sem{\termB_2}_\Dil
                    = \csem{\termB_2'}
                    = (\ident \otimes \discardDiag_{\voutchB_2})
                    \circ \EmbeddingFunctor\sem{\termB'_2}
                    \circ (\ident \otimes \ketbra{1}{1}_{\vinchB_2})$,
            \item $n_i \defeq |\vinchB_i| + |\vinchC| = |\voutchB_i| + |\voutchC|$.
          \end{itemize}
          Now we compute the semantics.
          \begin{align*}
             &
            (\ident \otimes \discardDiag_{\voutch\voutchB_1\voutchB_2\voutchC\outchD}) \circ
            \EmbeddingFunctor\sem{
                               \letx{x \otimes y}{\bigl(
                                 \qifx{\term'}
                                 {(\voutchB_2\voutchC \assign \vinchB_2\vinchC; \termB'_1)}
                                 {(\voutchB_1\voutchC \assign\vinchB_1\vinchC; \termB'_2)}
                                 \bigr) }{(\outchD \assign x); y}
                             }
            \\ & \hspace{10cm}
            \circ
            (\ident \otimes \ketbra{1}{1}_{\vinch\vinchB_1\vinchB_2\vinchC})
            \\ & =
            (\discardDiag_{\outchD} \otimes \ident_A
            \otimes \discardDiag_{\voutchB_1\voutchB_2\voutchC})
            \circ
            \EmbeddingFunctor
            \Bigl(
            \ketbra{0}{0} \otimes \sem{(\voutchB_1\voutchC \assign\vinchB_1\vinchC; \termB'_2)}
            \quad +_\Hilb \quad
            \ketbra{1}{1} \otimes \sem{(\voutchB_2\voutchC \assign \vinchB_2\vinchC; \termB'_1)}
            \Bigr)
            \\ & \hspace{4cm}
            \circ
            \Bigl(
            (\ident_\qbit \otimes \discardDiag_{\voutch})
            \EmbeddingFunctor(
            \sem{\term'}
            )
            (\ident \otimes \ketbra{1}{1}_{\vinch} )
            \quad\otimes\quad (\ident \otimes  \ketbra{1}{1}_{\vinchB_1\vinchB_2\vinchC})
            \Bigr)
            \\ & =
            (\discardDiag_{\bool} \otimes \ident_{A}
            \otimes \discardDiag_{\voutchB_1\voutchB_2\voutchC})
            \circ
            \Bigl(
            \ketbra{\bfalse}{\bfalse} \otimes
            \EmbeddingFunctor\sem{\termB'_2} \otimes \ident_{\qbit^{\otimes n_1}}
            \\ & \hspace{4cm}
            \ \ +_\CPM \ \
            \ketbra{\btrue}{\btrue} \otimes
            \EmbeddingFunctor\sem{\termB'_1} \otimes \ident_{\qbit^{\otimes n_2}}
            \Bigr)
            \\ & \hspace{5cm}
            \circ
            \Bigl(
            (m \otimes \discardDiag_{\voutch})
            \EmbeddingFunctor(
            \sem{\term'}
            )
            (\ident \otimes \ketbra{1}{1}_{\vinch} )
            \quad\otimes\quad (\ident_{\Delta_2} \otimes \ketbra{1}{1}_{\vinchB_1\vinchB_2\vinchC})
            \Bigr)
            \\ & =
            (\discardDiag_{\bool} \otimes \ident_{A}) \circ
            \Bigl(
            \ketbra{\bfalse}{\bfalse} \otimes
            (\ident_A \otimes \discardDiag_{\voutchB_2})
            \EmbeddingFunctor\sem{\termB'_2}
            (\ident_{\Delta_2} \otimes \ketbra{1}{1}_{\vinchB_2} )
            \otimes \discardDiag_{\voutchB_1\voutchC}\ketbra{1}{1}_{\vinchB_1\vinchC}
            \\ & \hspace{3cm}
            \ \ +_\CPM \ \
            \ketbra{\btrue}{\btrue} \otimes
            (\ident_A \otimes \discardDiag_{\voutchB_1})
            \EmbeddingFunctor\sem{\termB'_1}
            (\ident_{\Delta_2} \otimes \ketbra{1}{1}_{\vinchB_1} )
            \otimes \discardDiag_{\voutchB_2\voutchC}\ketbra{1}{1}_{\vinchB_2\vinchC}
            \Bigr)
            \\ & \hspace{6cm}
            \circ
            \Bigl(
            (m \otimes \discardDiag_{\voutch})
            \EmbeddingFunctor(
            \sem{\term'}
            )
            (\ident_{\Delta_1} \otimes \ketbra{1}{1}_{\vinch} )
            \quad\otimes\quad \ident_{\Delta_2}
            \Bigr)
            \\ & =
            (\discardDiag_{\bool} \otimes \ident_{A}) \circ
            \bigl(
            \ketbra{\bfalse}{\bfalse} \otimes \sem{\termB_2}_\Dil
            \ +_\CPM \
            \ketbra{\btrue}{\btrue} \otimes \sem{\termB_1}_\Dil
            \bigr)
            (m\sem{\term}_\Dil \otimes \ident_{\Delta_2})
            \\ & =
            \sem{\cifx{\term}{\termB_1}{\termB_2}}_\Dil.
            \tag*{\qedhere}
          \end{align*}
  \end{description}
\end{proof}

\subsection{Proof of Type Safety}

\begin{lemma}[Uniqueness of type derivation]
  \label{app:lem:op-sem:unique-derivation}
  In \QifUnitary{},
  for each $\term$, $\widehat\Gamma$, and $\Delta$,
  a type derivation tree
  $\auxJQ{\Delta}{\term \colon A}{(\vinch, \voutch)}$
  is unique.
  Note that we distinguish between $((\term;\termB);\termC)$ and $(\term;(\termB;\termC))$,
  and between $((\term\otimes\termB)\otimes\termC)$ and $(\term\otimes(\termB\otimes\termC))$.
  Also, note that
  in \cref{fig:procedural:ver2-ext:typing-rules},
  we added a type annotation to the lambda abstraction, writing it as
  $\lambda x^A. \term$.
  \thmend
\end{lemma}

\begin{lemma}[Substitution of derivations]
  \label{app:lem:op-sem:sub-derivation}
  Let
  $\term$ and $\termB$ be terms such that
  $\auxJQ{\Delta}{\term \colon A}{(\vinch, \voutch)}$
  and
  $\auxJQ{\Delta}{\termB \colon A}{(\vinch, \voutch)}$.
  For each context $\ctxt[\cdot][]$,
  $\auxJQ{\Delta'}{\ctxt[\term][] \colon B}{(\vinch\vinchB, \voutch\voutchB)}$
  if and only if
  $\auxJQ{\Delta'}{\ctxt[\termB][] \colon B}{(\vinch\vinchB, \voutch\voutchB)}$.
\end{lemma}
\begin{proof}
  By a straightforward induction on the structure of contexts.
\end{proof}

\begin{lemma}[Substitution of variables]
  \label{app:lem:op-sem:subst}
  Let $\term$ be a term with a type derivation
  $\auxJQ{\Delta}{\term \colon A}{(\vinch, \voutch)}$,
  and $\termB$ be a term with a free variable $x$
  with sets of i/o channels $(\vinchB, \voutchB)$ such that
  $\vinch \cap \vinchB = \voutch \cap \voutchB = \emptyset$.
  Then
  $\auxJQ{\Delta, \Delta'}{\termB[\term / x] \colon B}{(\vinch\vinchB, \voutch\voutchB)}$,
  if
  $\auxJQ{x \colon A, \Delta'}{\termB \colon B}{(\vinchB, \voutchB)}$.
\end{lemma}
\begin{proof}
  By induction on the type derivation of $\termB$.
  The only interesting case is when
  $x$ is contained in the body of $\qif$.
  In this case, substitution creates more than one copy of $\term$.
\end{proof}

\begin{proof}[Proof of \Cref{thm:opsem:type-sefety}]
  By \cref{app:lem:op-sem:sub-derivation},
  it suffices to consider the case $\ctxt[][]=[\cdot]$.
  All $\beta$-reduction cases follow from
  \cref{app:lem:op-sem:subst}.
  The claim is trivial for the commuting conversions.
  The only remaining nontrivial cases are the rule for $\F$ and  elimination of $\qif$, which can also be checked easily by inspection of the typing rules.
\end{proof}

\subsection{Proof of Progress Theorem}

\begin{lemma}[Type of values]
  \label{app:lem:op-sem:type-of-value}
  Assume $\auxJQ{\EmpEnv}{\val \colon A}{(\vinch, \voutch)}$.
  \begin{itemize}
    \item If $A = \unit$, then $\val = \unitval$.
    \item If $A = \qbit$, then $\val = \inch$.
    \item If $A = B \otimes C$, then $\val = \valB \otimes \valB'$.
    \item If $A = B \li C$, then $\val = U$ or $\val = \lambda x. \term$.
          \thmend
  \end{itemize}
\end{lemma}

\begin{proof}[Proof of \Cref{thm:opsem:progress}]
  Induct on the term at the root of a configuration.
  \begin{description}
    \item[Case $x$.]
          A variable cannot occur there.
    \item[Case $\inch$.]
          Input channels are values.
    \item[Case $\unitval$.]
          Unit is a value.
    \item[Case $\lambda x.\term$.]
          Abstractions are values.
    \item[Case $U$.]
          Unitary constants are values.
    \item[Case $\F\langle\vinch,\voutch\rangle$.]
          A global call unfolds.
    \item[Case $\term\,\termB,\ \term\otimes\termB,\ \term;\termB,\ \letx{x\otimes y}{\term}{\termB}$.]
          For these cases, inspect a proper
          subterm in the ordinary empty context.
          If it can step, lift that step; otherwise it has an output-prefix value
          form by induction.
          A nonempty prefix moves outward by a commuting conversion.
          With no prefix, \cref{app:lem:op-sem:type-of-value} gives the applicable
          beta, unit, pair-let, or gate rule, or a tensor value.
    \item[Case $\letx{x}{\term}{\termB}$.]
          Ordinary let is immediately a beta redex without a value restriction.
    \item[Case $\qifx{L}{M_1}{M_0}$.]
          Here $\auxJQ{\EmpEnv}{L\colon\qbit}{(\vinchC,\voutchC)}$,
          and $M_1$ and $M_0$ have the same type $A$ with $\FirstOrderType{A}$
          and the same i/o channels, which are disjoint from $(\vinchC,\voutchC)$.
          We first inspect $L$ by induction.
          \begin{itemize}
            \item If $(\vinch\mapsto\ket\psi, L)\step(\vinch\mapsto\ket{\psi'}, L')$,
                  then
                  $(\vinch\mapsto\ket\psi, \qifx{L}{M_1}{M_0})
                    \step(\vinch\mapsto\ket{\psi'}, \qifx{L'}{M_1}{M_0})$,
                  because $\qifx{[\cdot]}{M_1}{M_0}$ is a frame:
                  $\qifx{\ctxt[][\qctrl]}{M_1}{M_0}$ is again a context $\ctxt[][\qctrl]$.
            \item If $L=(\outch\assign\inchB);L'$, a commuting conversion applies:
                  \[
                    \qifx{((\outch\assign\inchB);L')}{M_1}{M_0}
                    \step
                    (\outch\assign\inchB);\qifx{L'}{M_1}{M_0}.
                  \]
            \item Otherwise $L$ is a value of type $\qbit$, so $L=\inch$
                  by \cref{app:lem:op-sem:type-of-value}.
          \end{itemize}
          In the last case, we inspect $M_1$ by induction.
          If $(\vinch\mapsto\ket\psi, M_1)\step(\vinch\mapsto\ket{\psi'}, M_1')$,
          then
          $(\vinch\mapsto\ket\psi, \qifx{\inch}{M_1}{M_0})
            \step(\vinch\mapsto\ket{\psi''}, \qifx{\inch}{M_1'}{M_0})$.
          For a rule other than the gate rule, $\ket{\psi''}=\ket{\psi'}=\ket\psi$.
          For the gate rule at the hole of $\ctxt[][\qctrl]$,
          $\qifx{\inch}{\ctxt[][\qctrl]}{M_0}$ is the controlled context
          $\ctxt[][\qctrl,\inch\mapsto1]$,
          so $U$ acts as $C_{\qctrl,\inch\mapsto1}U$;
          here $\inch\notin\mathrm{dom}(\qctrl)$ because $\inch$ is a channel of $L$
          and hence not of $M_1$.
          The same holds for $M_0$ with $\inch\mapsto0$.

          Finally, if neither branch can step, then by induction
          \[
            \qifx{L}{M_1}{M_0}
            =
            \qifx{\inch}
            {((\voutchB_1 \assign \vinchB_1); \val_1)}
            {((\voutchB_0 \assign \vinchB_0); \val_0)}.
          \]
          We show that this term has the form of the left-hand side of the last rule
          of \cref{fig:procedural:ver2-ext:operational-semantics}.
          \begin{itemize}
            \item Each output channel of the branches occurs exactly once
                  in $\voutchB_1$ and once in $\voutchB_0$.
                  Reordering the output assignments, we assume
                  $\voutchB_1=\voutchB_0=\outch_1,\dots,\outch_n$.
            \item $\val_1$ and $\val_0$ are closed values of the first-order type $A$.
                  By \cref{app:lem:op-sem:type-of-value}, they consist of
                  $\unitval$, $\otimes$, and input channels, arranged according to $A$.
                  Hence both are obtained from a common value $\val[\cdots]$,
                  with one hole for each $\qbit$ in $A$, by filling the holes
                  with input channels.
            \item Each input channel of the branches occurs exactly once in each branch.
                  Let $\inchB_1,\dots,\inchB_m$ be the input channels of $M_0$ in order
                  of appearance, so that
                  $M_0=(\outch_1\assign\inchB_1);\dots;(\outch_n\assign\inchB_n);
                    \val[\inchB_{n+1},\dots,\inchB_m]$.
                  Since $M_1$ has the same input channels, there is a permutation $\sigma$ with
                  $M_1=(\outch_1\assign\inchB_{\sigma(1)});\dots;(\outch_n\assign\inchB_{\sigma(n)});
                    \val[\inchB_{\sigma(n+1)},\dots,\inchB_{\sigma(m)}]$.
          \end{itemize}
          Hence the last rule applies.
    \item[Case $(\xi:=L);M$.]
          For $(\xi:=L);M$, a step in $L$ lifts and a prefix in $L$ commutes outward.
          Otherwise $L$ is a qubit value, hence an input channel.
          A step in $M$ then lifts; if none exists, induction makes the whole term
          an output-prefix value form.
          \qedhere
  \end{description}
\end{proof}

\subsection{Proof of Termination}

To prove strong normalisation, we define a well-founded order on terms and show that every step decreases it.

\begin{definition}[Normalisation ordering]\label{def:opsem:normalisation-order}

  Fix finite well-typed definitions $\Defs$, ordered so that each body refers
  only to earlier definitions.
  An \emph{additive slice} of a term selects one branch at every occurrence of $\qif$,
  retaining its condition and a marked conditional node; at every other
  constructor it independently selects slices of all its subterms.
  The choices form a finite multiset, so identical slices are counted with
  their multiplicities.
  For example, writing $\qif^\bullet$ for the marked conditional node,
  $\qifx{\inch}{\gX\,\inchB}{\inchB}$ has two slices, one for each branch:
  \[
    \begin{array}{c@{\qquad\qquad}c}
      \multicolumn{2}{c}{\qifx{\inch}{\gX\,\inchB}{\inchB}}            \\[2pt]
      \swarrow                           & \searrow                    \\[2pt]
      \qif^\bullet\ \inch\ (\gX\,\inchB) & \qif^\bullet\ \inch\ \inchB
    \end{array}
  \]
  Similarly, $\qifx{\inch}{\inchB}{\inchB}$ has the slice
  $\qif^\bullet\ \inch\ \inchB$ with multiplicity two,
  and a tensor $T_1\otimes T_2$ of two such conditionals has $2\times2$ slices.

  Assign a weight in $\mathbb N^3$ to each slice by adding its node weights,
  and write $\Phi(T)$ for the multiset of slice weights of $T$.
  A variable, input channel, or unit value has weight $(0,0,1)$, a unitary
  constant has weight $(0,1,1)$, a marked conditional node has weight
  $(1,0,1)$, and every other constructor has weight $(0,0,1)$.
  Binders, type annotations, and channel names in a global call are not
  additional term nodes.
  Only the chosen branch is sliced; discarded branches contribute no further
  choices or multiplicities.
  For example, with $e_q=(1,0,1)$,
  \[
    \Phi(\qifx{L}{M_1}{M_0})
    =\biguplus_{b=0}^1
    \{\!\{e_q+a+d\mid a\in\Phi(L),\ d\in\Phi(M_b)\}\!\}.
  \]
  In particular, the two slices of $\qifx{\inch}{\gX\,\inchB}{\inchB}$ above
  have weights $(1,1,5)$ and $(1,0,3)$.
  A global call is atomic, with weight
  \[
    \omega_\F
    =\max_{\mathrm{lex}}\Phi(M_\F)+(0,0,1).
  \]
  This defines the global weights inductively along $\Defs$.
  Triples are compared lexicographically, with the conditional component first.

  Write $P>_{\mathrm{mul}}Q$ for the finite multiset extension of this order:
  $P=Z\uplus X$ and $Q=Z\uplus Y$, where $X$ is nonempty and every element of
  $Y$ is strictly smaller than some element of $X$.
  The finite multiset extension of a well-founded order is well founded.
  \thmend
\end{definition}

\begin{lemma}\label{lem:opsem:slice-substitution}
  In a well-typed term, every ordinary free variable occurs exactly once in
  each additive slice.
  If $x$ is free in $M$, then, with $e=(0,0,1)$,
  \[
    \Phi(M[N/x])
    =\{\!\{\,a+d-e\mid a\in\Phi(M),\ d\in\Phi(N)\,\}\!\}.
  \]
  The terms substituted for different variables have disjoint ordinary
  variable contexts and channel interfaces.
  \thmend
\end{lemma}
\begin{proof}
  By induction on the typing derivation of $M$.
  Every rule except that of $\qif$ splits the ordinary context among its
  premises, so $x$ occurs in exactly one of them.
  In $\qifx{L}{M_1}{M_0}$, the variable $x$ occurs either in $L$ or in both
  branches, and each slice selects exactly one branch.
  Hence each slice of $M$ contains exactly one occurrence of $x$.
  Substitution replaces this occurrence, of weight $e$, by a slice of $N$, and
  the choices in $M$ and $N$ are independent; this gives the displayed
  bijection between slices.
  The last claim holds because the typing rules split the ordinary contexts and
  the channels among the substituted terms, as in $\termB\otimes\termC$.
\end{proof}

\begin{definition}[Ordering for commuting conversions]\label{def:opsem:commuting-conversion-order}
  Call an output assignment $(\outch\assign L);M$, a tensor binding
  $\letx{x\otimes y}{L}{M}$, an ordinary binding $\letx{x}{L}{M}$, and
  a sequencing $L;M$ a \emph{prefix} $B(L,M)$
  with \emph{input} $L$ and \emph{continuation} $M$.
  Let $b(T)$ be the number of prefix nodes in $T$, and
  $\Psi(T)$ the sum of their depths, where an edge into the
  continuation of a prefix counts as zero and every other edge as one.
  That is, for a prefix $B(L,M)$ and any other constructor $D$,
  \begin{align*}
    b(B(L,M))               & =1+b(L)+b(M),                        &
    \Psi(B(L,M))            & =\Psi(L)+b(L)+\Psi(M),                 \\
    b(D(T_1,\ldots,T_k))    & =\textstyle\sum_i b(T_i),            &
    \Psi(D(T_1,\ldots,T_k)) & =\textstyle\sum_i(\Psi(T_i)+b(T_i)),
  \end{align*}
  and both are zero on atoms.
  \thmend
\end{definition}

\begin{proof}[Proof of \Cref{thm:opsem:termination}]
  We show that every step $T\step T'$ satisfies
  \[
    \Phi(T)>_{\mathrm{mul}}\Phi(T')
    \quad\text{or}\quad
    \bigl(\Phi(T)=\Phi(T')\ \text{and}\ \Psi(T)>\Psi(T')\bigr).
  \]
  This lexicographic order is well founded.
  We first check each of the exhaustive rules of
  \cref{fig:procedural:ver2-ext:operational-semantics}
  in the empty context, with $e=(0,0,1)$.
  \begin{description}
    \item[Case $\beta$-rules.]
          For $(\lambda x.\term)\,\termB$, $\letx{x}{\termB}{\term}$,
          $\letx{x\otimes y}{\termB\otimes \termC}{\term}$, and $\unitval;\term$,
          only the third component decreases:
          by \cref{lem:opsem:slice-substitution}, every slice weight decreases
          by $3e$, $2e$, $4e$, and $2e$, respectively.
    \item[Case $U(\inch_1\otimes\cdots\otimes\inch_n)$.]
          Only the second and third components decrease, as the unitary constant
          and the application node disappear:
          \[
            \Phi(U(\inch_1\otimes\cdots\otimes\inch_n))
            =\{\!\{w+(0,1,2)\}\!\}
            >_{\mathrm{mul}}
            \{\!\{w\}\!\}
            =\Phi(\inch_1\otimes\cdots\otimes\inch_n)
          \]
          for some $w\in\mathbb N^3$.
    \item[Case $\F\langle\vinchB,\voutchB\rangle$.]
          Renaming channels preserves $\Phi$, and every element of $\Phi(\term_\F)$
          is smaller than $\omega_\F$. Hence, even when the body has more slices,
          \[
            \Phi(\F\langle\vinchB,\voutchB\rangle)
            =\{\!\{\omega_\F\}\!\}
            >_{\mathrm{mul}}
            \Phi(\term_\F)
            =\Phi(\term_\F[\vinchB/\vinch,\voutchB/\voutch]).
          \]
    \item[Case $\qifx{\inch}{\term_1}{\term_0}$.]
          The branches $\term_1$ and $\term_0$ of the last rule have the same shape
          and consist of output assignments, $\otimes$, $\unitval$, and input channels
          only; let $k$ be their number of nodes.
          The right-hand side $\termB$ contains no conditional and exactly one unitary
          constant $\mathsf C\lceil\sigma\rceil$; let $k'$ be its number of nodes.
          The first component decreases:
          \[
            \Phi(\qifx{\inch}{\term_1}{\term_0})
            =\{\!\{(1,0,k+2),\,(1,0,k+2)\}\!\}
            >_{\mathrm{mul}}
            \{\!\{(0,1,k')\}\!\}
            =\Phi(\termB).
          \]
    \item[Case commuting conversions.]
          Each of the fifteen commuting conversions moves a prefix out of a
          single frame $D[\cdot]$:
          \[
            D[B(L,M)]\step B(L,D[M]),
          \]
          where $D[\cdot]$ is one of
          $T\,[\cdot]$, $[\cdot]\,T$, $T\otimes[\cdot]$, $[\cdot]\otimes T$,
          $\qifx{[\cdot]}{T_1}{T_0}$, or a prefix whose input is the hole.
          Let $e_D$ and $e_B$ be the node weights of $D$ and $B$, and
          let $\Phi_D$ be $\Phi(T)$ for the other subterm $T$ of $D$,
          or $\Phi(T_1)\uplus\Phi(T_0)$ if $D$ is a $\qif$.
          Since the hole of $D$ is not a continuation,
          $\Psi(D[T'])=\Psi(T')+b(T')+r_D$ for a constant $r_D$.
          Hence, with $l$, $m$, and $c$ ranging over $\Phi(L)$, $\Phi(M)$, and $\Phi_D$,
          \begin{align*}
            \Phi(D[B(L,M)])
             & =\{\!\{e_D+e_B+l+m+c\}\!\}
            =\Phi(B(L,D[M])),                                     \\
            \Psi(D[B(L,M)])
             & =\Psi(L)+b(L)+\Psi(M)+\bigl(1+b(L)+b(M)\bigr)+r_D, \\
            \Psi(B(L,D[M]))
             & =\Psi(L)+b(L)+\Psi(M)+b(M)+r_D.
          \end{align*}
          Thus $\Phi$ and $b$ are preserved, and $\Psi$ decreases by $1+b(L)>0$.
  \end{description}

  These decreases are preserved in every context $K$.
  Slice weights combine by addition and finite multiset union, so
  $\Phi(T)=\Phi(T')$ implies $\Phi(K[T])=\Phi(K[T'])$, and
  $\Phi(T)>_{\mathrm{mul}}\Phi(T')$ implies $\Phi(K[T])>_{\mathrm{mul}}\Phi(K[T'])$.
  Moreover, there are constants $d_K,k_K$ such that
  \[
    \Psi(K[T])=\Psi(T)+d_K\, b(T)+k_K,
  \]
  so the decrease of $\Psi$ by a commuting conversion, which preserves $b$,
  is also preserved.

  This proves strong normalisation, including for well-typed open terms.
  In particular, no configuration has an infinite reduction sequence.
  By \cref{thm:opsem:progress}, an irreducible configuration has the stated
  value form $(\voutch\assign\vinch);\val$.
\end{proof}

\subsection{Proof of Preservation of Categorical Semantics}

\begin{lemma}\label{app:lem:op-sem:sem-subst}
  For each term $\auxJQ{\Delta}{\term \colon A}{(\vinch, \voutch)}$
  and $\auxJQ{\Delta', x \colon A}{\termB \colon B}{(\vinchB,\voutchB)}$,
  \[
    \sem{\termB[\term/x]} =
    (\sem{\termB}\otimes \ident_{\voutch})
    \sigma (\ident_{\Delta'} \otimes \sem{\term} \otimes \ident_{\vinchB})
  \]
  with some appropriate swap $\sigma$.
\end{lemma}
\begin{proof}
  By a straightforward induction on the structure of $\termB$.
\end{proof}

\begin{proof}[Proof of \Cref{thm:opsem:sem-pres}]
  The preservation of categorical semantics for the rules
  regarding commuting conversion and $\F$ is trivial.

  It is also easy for the rules from linear lambda calculus.
  For example, we can check the preservation of $\beta$-reduction
  using the previous lemma as follows.
  \begin{align*}
     &
    \sem{\auxJQ{\Delta,\Delta'}{(\lambda x. \term)\,\termB \colon B}
      {(\vinch\vinchB, \voutch\voutchB)}}
    \quad\colon\quad
    \sem{\Delta} \otimes \sem{\Delta'} \otimes (\CC^2)^{\vinch} \otimes (\CC^2)^{\vinchB}
    \longrightarrow
    \sem{B} \otimes (\CC^2)^{\voutch} \otimes (\CC^2)^{\voutchB}
    \\ & =
    \begin{tikzpicture}[yscale=.5,baseline={(0,0)}]
      \node(D) at (0,2) {$\Delta$};
      \node(a) at (0,1) {$\vinch$};
      \node(D') at (0,-1) {$\Delta'$};
      \node(b) at (0,-2) {$\vinchB$};
      \draw (D) to ($(D) + (1,0)$);
      \draw (a) to ($(a) + (1,0)$);
      \draw (D') to ($(D') + (1,0)$);
      \draw (b) to ($(b) + (1,0)$);
      \node(A*) at (2,2) {$\!A^*\!\!$};
      \node(B) at (2,1) {$B$};
      \node(xi) at (2,0) {$\voutch$};
      \node(A) at (2,-1) {$\!A\!$};
      \node(ze) at (2,-2) {$\voutchB$};
      \draw (A*) to ($(A*) - (1,0)$);
      \draw (B) to ($(B) - (1,0)$);
      \draw (xi) to ($(xi) - (1,0)$);
      \draw (A) to ($(A) - (1,0)$);
      \draw (ze) to ($(ze) - (1,0)$);
      \draw (B) to ($(B) + (1,0)$);
      \draw (xi) to ($(xi) + (1,0)$);
      \draw (ze) to ($(ze) + (1,0)$);
      \node (m) at (1,1.) [rectangle,draw,fill=white,minimum height=3em] {$\lambda x. \term$};
      \node (n) at (1,-1.5) [rectangle,draw,fill=white,minimum height=3em] {$\termB$};
      \draw (A*) to[out=-30,in=90] ($(A*)!0.5!(A) + (0.7,0)$) to[out=-90,in=30] (A);
    \end{tikzpicture}
    \quad = \quad
    \begin{tikzpicture}[yscale=.5,baseline={(0,0.2)}]
      \coordinate(x) at (0.2,3);
      \node(D) at (0,2) {$\Delta$};
      \node(a) at (0,1) {$\vinch$};
      \node(D') at (0,-1) {$\Delta'$};
      \node(b) at (0,-2) {$\vinchB$};
      \draw (x) to ($(x) + (1,0)$);
      \draw (D) to ($(D) + (1,0)$);
      \draw (a) to ($(a) + (1,0)$);
      \draw (D') to ($(D') + (1,0)$);
      \draw (b) to ($(b) + (1,0)$);
      \node(A*) at (2,4) {$\!\!A^*\!\!\!$};
      \node(B) at (2,2) {$B$};
      \node(xi) at (2,1) {$\voutch$};
      \node(A) at (2,-1) {$\!\!A\!\!$};
      \node(ze) at (2,-2) {$\voutchB$};
      \draw (B) to ($(B) - (1,0)$);
      \draw (xi) to ($(xi) - (1,0)$);
      \draw (A) to ($(A) - (1,0)$);
      \draw (ze) to ($(ze) - (1,0)$);
      \draw (B) to ($(B) + (1,0)$);
      \draw (xi) to ($(xi) + (1,0)$);
      \draw (ze) to ($(ze) + (1,0)$);
      \node (m) at (1,2.) [rectangle,draw,fill=white,minimum height=3em] {$\term$};
      \node (n) at (1,-1.5) [rectangle,draw,fill=white,minimum height=3em] {$\termB$};
      \coordinate(x') at ($ (x) + (-0,1) $);
      \draw (x) to[out=180,in=180] (x') to (A*)
      to[out=-45,in=90] ($(A*)!0.5!(A) + (0.7,0)$) to[out=-90,in=45] (A);
    \end{tikzpicture}
    \quad = \quad
    \begin{tikzpicture}[yscale=.5,baseline={(0,-0.5)}]
      \node(D) at (0,1) {$\Delta$};
      \coordinate(x) at (0,2);
      \node(a) at (0,0) {$\vinch$};
      \node(D') at (0,-1) {$\Delta'$};
      \node(b) at (0,-2) {$\vinchB$};
      \draw (D) to ($(D) + (2.5,0)$);
      \draw (a) to ($(a) + (2.5,0)$);
      \draw (D') to ($(D') + (1,0)$);
      \draw (b) to ($(b) + (1,0)$);
      \node(B) at (3.5,1) {$B$};
      \node(xi) at (3.5,0) {$\voutch$};
      \node(ze) at (2,-2) {$\voutchB$};
      \draw (B) to ($(B) - (1,0)$);
      \draw (xi) to ($(xi) - (1,0)$);
      \draw (ze) to ($(ze) - (1,0)$);
      \coordinate (m) at (2.5,1.);
      \coordinate (n) at (1,-1.5);
      \draw ($(n) + (0,.9)$) to[out=0,in=180] ($(m) + (0,.9)$);
      \node at (m) [rectangle,draw,fill=white,minimum height=3em] {$\term$};
      \node at (n) [rectangle,draw,fill=white,minimum height=3em] {$\termB$};
    \end{tikzpicture}
    \\ & =
    \sem{\auxJQ{\Delta,\Delta'}{\term[\termB/x] \colon B}{(\vinch\vinchB,\voutch\voutchB)}}.
  \end{align*}

  For qif elimination, let $T_0$ be the pure linear map of the aligned
  else-branch, including its output channels.
  The then-branch map is $T_0\ceil\sigma$, so the whole map is
  \[
    (\ident_\alpha\otimes T_0)
    \circ\bigl(\ketbra00_\alpha\otimes I+
    \ketbra11_\alpha\otimes\ceil\sigma\bigr).
  \]
  This is exactly the map of the generated controlled-permutation application
  and fresh tensor bindings, followed by the single aligned output prefix.
  The rule therefore preserves the semantics.

  For the gate rule, induct on the controlled context $K_c$.
  \begin{description}
    \item[Case $ {[\cdot]}$.]
          At the empty context,
          $\sem{U\vinch}\ket\psi=\sem{\vinch}(U_{\vinch}\ket\psi)$.
    \item[Case $\qifx{\inch}{\ctxt[][c]}{\term}$ or $\qifx{\inch}{\term}{\ctxt[][c]}$.]
          Crossing the then-branch of a $\qif$ with known condition $\alpha$ replaces
          the action $V$ by
          $\ketbra00_\alpha\otimes I+\ketbra11_\alpha\otimes V$;
          the else-branch is analogous.
    \item[Case $ {D[\ctxt[][c]]}$.]
          At every other context, transport the
          same action on the allocated channel factors.
          At an abstraction this is naturality of currying: the acted-on input
          channels are distinct from its ordinary bound variable.
  \end{description}
  Thus the accumulated action is $C_cU_{\vinch}$.
\end{proof}

\subsection{Proofs for Language with Divergence}

\subsubsection{Syntactic Dilation}

The syntactic dilation of $\abort$ is defined as
$\abort \SynDil (\testch\assign\inch); \unitval$
where the right-hand side has the typing derivation
\[
  \auxJQ{\EmpEnv}{
    (\testch\assign\inch); \unitval
    \colon \unit
  }{(\inch, \cdot, \testch)},
\]
which proves \cref{thm:syn-dil:type-safety}
in the extended syntactic dilation.

We define the semantics of $\abort$ in $\Dil$ by
\[
  \sem{\abort} = (0,0) \quad\in\quad \Dil((1),(1)).
\]
We can calculate that the second projection ($= 0$)
is equal to $\csem{ (\testch\assign\inch); \unitval }$
as
\begin{align*}
  \csem{\,
    (\testch\assign\inch); \unitval
    \,}
  \ =\
  \EmbeddingFunctor
  \Bigl(
  \bra{0}
  \circ
  \sem{\,
    (\testch\assign\inch); \unitval
    \,}
  \circ
  \ket{1}
  \Bigr)
  \ =\
  \EmbeddingFunctor
  \langle 0 | 1 \rangle
  \ =\
  0
  ,
\end{align*}
which proves \cref{thm:syn-dil:sem-pres}
in the extended syntactic dilation.

\subsubsection{QifUnitary with Output Test Channels}

Type safety (\cref{thm:opsem:type-sefety}),
progress theorem (\cref{thm:opsem:progress}),
termination (\cref{thm:opsem:termination}),
and
semantic preservation (\cref{thm:opsem:sem-pres})
trivially hold in the extended language
since there is no difference between output channels
and output test channels in the operational semantics.

We restate the adequacy theorem (\cref{cor:adequacy})
in these extended languages.
\begin{corollary}[Adequacy for language with divergence]
  In the language \Qif{} extended with divergence,
  for any collection of function definitions $\Defs\colon\Gamma$ and any term
  $\termJC{\EmpEnv}{\term:A}$ such that
  $\FirstOrderType{A}$,
  the term $\term$ evaluates to $\sem{\term}_{\sem{\Defs}}$.
  That is,
  \begin{itemize}
    \item there exist \QifUnitary{}
          definitions $\widehat\Defs$,
          non-linear context $\widehat\Gamma$,
          and a term $\term'$
          such that
          $\Defs \SynDil \widehat\Defs$,
          $\Gamma \SynDil \widehat\Gamma$, and
          $\term \SynDil \term'$,
    \item they admit type derivations
          $\widehat\Defs \colon \widehat\Gamma$
          and
          $\auxJQ{\EmpEnv}{\term'\colon A}{(\vinch, \voutch, \vtestch)}$,
    \item the term $\term'$ evaluates to $\ket\varphi$ as
          $(\vinch \mapsto \ket{1 \dots 1}, \term') \step^*_{\widehat\Defs}
            \bigl(
            \vinch \mapsto \ket\varphi
            ,
            (\outch_1\assign\inch_{\sigma(1)});
            \dots;
            (\outch_m\assign\inch_{\sigma(m)});
            (\testch_1\assign\inch_{\sigma(m+1)});
            \dots;
            (\testch_k\assign\inch_{\sigma(m+k)});
            \val[\inch_{\sigma(m+k+1)}, \dots, \inch_{\sigma(n)}]
            \bigr)
          $, and
    \item by reordering the qubits,
          projecting the $(m+1)$st through $(m+k)$th qubits,
          and tracing out the first $m$ qubits,
          we obtain the state:
          \begin{align*}
            \trace_{{(\CC^2)}^{\otimes m}}(
            \
            \bra{0}^{\otimes k}
            \ceil{\sigma}\ketbra\varphi\varphi\ceil{\sigma^{-1}}
            \ket{0}^{\otimes k}
            \
            )
            \quad=\quad
            s_A \circ \sem{\term}_{\sem{\Defs}}.
            \tag*{\thmend}
          \end{align*}
  \end{itemize}
\end{corollary}

\subsection{Full Abstraction for Qif with Divergence}
\label{app:sec:full-abstraction}

\begin{lemma}
  \label{app:lem:(co)unit-definable}
  For each $n \in \Nat$,
  there exist terms $H_n$ and $E_n$
  whose types are
  \begin{align*}
    \termJC{\EmpEnv}{H_n \colon \qbit^{\otimes n} \otimes \qbit^{\otimes n}}
    \qquad\quad
    \termJC{x_n \colon \qbit^{\otimes n} \otimes \qbit^{\otimes n}}
    {E_n \colon \unit},
  \end{align*}
  and whose semantics are the unit and counit of the compact closed structure
  up to some non-zero reals $a_n, b_n > 0$.
  \begin{align*}
    \sem{H_n} = a_n \cdot \eta_{(2^n)}
    \qquad\quad
    \sem{E_n} = b_n \cdot \varepsilon_{(2^n)}.
  \end{align*}
\end{lemma}
\begin{proof}
  By induction on $n$.
  \begin{description}
    \item[Case $n = 0$.] Trivial.
    \item[Case $n = 1$.]
          We define $H_{1}$ and $E_{1}$ as
          \begin{align*}
             &
            \termJC{\EmpEnv}
            {
              \gCX (\gH \ket{0} \otimes \ket{0})
              \colon \qbit \otimes \qbit}
            \\ &
            \termJC{x_{\qbit} \colon \qbit \otimes \qbit}{
              \letx{y \otimes z}
              {\gCX\ x_{\qbit}}
              {\bra 0 (\gH y); \bra0 z\,}
              \colon \unit }.
          \end{align*}
          Here,
          $\gH$ is the Hadamard gate,
          $\gCX$ is the CNOT gate,
          and $\ket0$ is an abbreviation of $\gX\ket1$.
          In the second line, $\bra0$
          is the term defined by
          \begin{align*}
            \bra{0}
             & \defeq
            \lambda x.\, \cifx{(\meas x)}{\abort}{\unitval}
            \colon \qbit \li \unit.
          \end{align*}
          The semantics of $H_1$ gives the maximally entangled state,
          which is $\frac{1}{2}\eta_{(2)}$.
          The semantics of $E_1$ defines
          $\frac{1}{2}\varepsilon_{(2)}$.
    \item[Case $n + m$.]
          We define $H_{n + m}$ and $E_{n + m}$ as
          \begin{align*}
             &
            \termJC{\EmpEnv}
            {
              \letx{x_n \otimes y_n}{H_n}{
                \letx{x_m \otimes y_m}{H_m}{
                  (x_n \otimes x_m) \otimes (y_n \otimes y_m)
                }
              }
            }
            \\ & \hspace{15em}
            \colon (\qbit^n \otimes \qbit^m)\otimes (\qbit^n \otimes \qbit^m)
            \\ &
            \termJC{x_{n + m} \colon (\qbit^n \otimes \qbit^m)\otimes (\qbit^n \otimes \qbit^m)}{
              \Let {(y_n \otimes y_m) \otimes (z_n \otimes z_m)}
              = {x_{n + m}}
              \In
              \\ & \hspace{14em}
              {E_{n}[y_n \otimes z_n/x_n]; E_{m}[y_m \otimes z_m/x_m] }
              \colon \unit }.
          \end{align*}
          The semantics of $H_{n + m}$ is
          $a_n a_m \cdot \sigma_{n,m} (\eta_n \otimes \eta_m)$,
          where $\sigma$ swaps the middle factors.
          Thus, this yields $a_n a_m \cdot \eta_{n \otimes m}$.
          Similarly, $E_{n + m}$ defines
          $b_n b_m \cdot \varepsilon_{n \otimes m}$.
          \qedhere
  \end{description}
\end{proof}

\begin{lemma}
  \label{app:lem:mono-epi-tensor}
  In $\CPM$, monomorphisms and epimorphisms are closed under tensor products.
\end{lemma}
\begin{proof}
  Since it is compact closed, by duality,
  it suffices to prove that tensor products of monics are monic.

  We first prove that $(f_{ij}) \colon \vec n \longrightarrow \vec m$ in $\CPM$
  is monic if and only if $\ceil{f}$ is injective,
  where $\ceil{f} \colon\allowbreak
    \bigoplus_{i} \Mat_{n_i}(\CC) \longrightarrow \bigoplus_{j} \Mat_{m_j}(\CC)$
  is the $\CC$-linear map that corresponds to
  the matrix of $\CC$-linear maps $(f_{ij})$
  such that $f_{ij} \colon \Mat_{n_i}(\CC) \longrightarrow \Mat_{m_j}(\CC)$.
  The implication from injectivity of \(\ceil f\) to monicity of \(f\) is immediate.
  Conversely, suppose that $f$ is monic. To prove that $\ceil{f}$ is injective,
  it suffices to show that
  $\ceil{f}^+$ is injective,
  where $\ceil{f}^+$ is a function
  restricted to positive cones:
  \begin{align*}
     &
    \textstyle
    \ceil{f}^+
    \colon
    \{(\rho_i)_i \mid \rho_i \in \Mat^+_{n_i}(\CC)\}
    \longrightarrow
    \{(\rho_j)_j \mid \rho_j \in \Mat^+_{m_j}(\CC)\}
    \\
     &
    \text{where} \qquad
    \Mat^+_{n}(\CC) \defeq
    \{
    \rho \in \Mat_{n}(\CC)
    \mid
    \rho \colon \text{positive semidefinite}
    \}.
  \end{align*}
  Indeed, every matrix $A$ can be decomposed as
  $A = B_1 + i B_2 - B_3 - iB_4$ with some
  positive semidefinite matrices. If $A$ is in the kernel of $\ceil{f}$,
  \begin{align*}
    O =
    \ceil{f}(B_1 + i B_2 - B_3 - iB_4)
     & =
    \ceil{f}(B_1) + i \ceil{f}(B_2) - \ceil{f}(B_3) - i\ceil{f}(B_4)
    \\
     & =
    \ceil{f}(B_1) - \ceil{f}(B_3) + i (\ceil{f}(B_2) - \ceil{f}(B_4)).
  \end{align*}
  From the uniqueness of decomposition of a matrix into
  a sum of Hermitian and skew-Hermitian matrices,
  \begin{align*}
    \ceil{f}(B_1) - \ceil{f}(B_3) = 0,
    \qquad
    \ceil{f}(B_2) - \ceil{f}(B_4) = 0.
  \end{align*}
  Since $\ceil{f}^+$ is injective, $B_1 = B_3$ and $B_2 = B_4$,
  thus $A = O$.
  The injectivity of $\ceil{f}^+$ follows readily:
  since each sequence of positive semidefinite matrices $(\rho_i \in \Mat^+_{n_i}(\CC))$
  defines a matrix of completely positive maps $(1) \longrightarrow \vec n$,
  the map $(f_{ij})$ being monic means $\ceil{f}^+$ is injective.

  Now, one can check that
  $\ceil{f \otimes g} \cong \ceil{f} \otimes \ceil{g}$.
  Therefore, since the tensor product of injections is an injection,
  the tensor product of monos is mono.
\end{proof}

\begin{lemma}
  \label{app:lem:type-inj-surj}
  For each type $A$,
  there exist some terms $\term_A$ and $\termB_A$
  \begin{align*}
    \termJC{x_A \colon \qbit^{\otimes m_A}}{\term_A \colon A}
    \qquad\quad
    \termJC{y_A \colon A}{\termB_A \colon \qbit^{\otimes n_A}}
  \end{align*}
  such that $\sem{\term_A}$ defines an epimorphism
  and $\sem{\termB_A}$ defines a monomorphism.
\end{lemma}
\begin{proof}
  By induction on the definition of types.
  \begin{description}
    \item[Case $\unit$.]
          We can define the terms as
          $\term_{\unit} \defeq x_{\unit}$ and
          $\termB_{\unit} \defeq y_{\unit}$
          where $m_{\unit} = n_{\unit} = 0$.
    \item[Case $\qbit$.]
          Similarly, we can define the terms as
          $\term_{\qbit} \defeq x_{\qbit}$ and
          $\termB_{\qbit} \defeq y_{\qbit}$
          where $m_{\qbit} = n_{\qbit} = 1$.
    \item[Case $\bool$.]
          We define $\term_{\bool} \defeq \meas\,x_{\qbit}$
          and $\termB_{\bool} \defeq \cifx{y_{\bool}}{\ket1}{\gX\ket1}$
          where $m_{\bool} = n_{\bool} = 1$.
    \item[Case $A \otimes B$.]
          We define $\term_{A \otimes B} \defeq \term_A \otimes \term_B$
          and
          $\termB_{A \otimes B}
            \defeq \letx{y_A \otimes y_B}{y_{A \otimes B}}
            {\termB_A \otimes \termB_B}$.
          From \cref{app:lem:mono-epi-tensor},
          we prove $\sem{\term_{A \otimes B}}$ is epic
          and $\sem{\termB_{A \otimes B}}$ is monic.
    \item[Case $A \li B$.]
          We define $\term_{A \li B}$ and $\termB_{A \li B}$ as follows:
          \begin{align*}
             &
            \termJC{z_A \colon \qbit^{\otimes n_A}, x_B \colon \qbit^{\otimes m_B}}
            {
              \lambda y_A.\,
              (\letx{x_{n_A}}{\termB_A \otimes z_A}{E_{n_A}};\, \term_B)
              \colon A \li B
            }
            \\ &
            \termJC{y_{A\li B} \colon A \li B}
            {
              \letx{x_A \otimes z_A}{H_{m_A}}
              {
                \bigl(
                \letx{y_B}{
                  y_{A \li B}\,\term_A
                }{\termB_B}
                \bigr)
                \otimes z_A}
              \colon \qbit^{\otimes m_A + n_B}
            }.
          \end{align*}
          The semantics of $\term_{A \li B}$ is given by
          $\sem{\termB_A}^* \otimes \sem{\term_B}$
          up to some positive real,
          where $\sem{\termB_A}^*$ is the dual of $\sem{\termB_A}$
          obtained by the compact closed structure.
          Thus it is epic.
          Similarly, the semantics of $\termB_{A \li B}$ is
          given by
          $\sem{\term_A}^* \otimes \sem{\termB_B}$
          up to some positive real,
          so it is monic.
          \qedhere
  \end{description}
\end{proof}

\begin{proof}[Proof of \cref{thm:CPM-full-abst}]
  Since our semantics is compositional,
  proving $\sem{C[\term]} = \sem{C[\termB]}$ from $\sem{\term} = \sem{\termB}$
  is trivial.
  We prove the other way around.

  We first consider the case where $\JG = \JC$.
  Without loss of generality, we can assume that the terms are closed,
  \ie
  \begin{align*}
    \termJC{\EmpEnv}{\term \colon A}
    \qquad\quad
    \termJC{\EmpEnv}{\termB \colon A}.
  \end{align*}
  If their semantics differ, then
  the semantics of
  \begin{align*}
    \termJC{\EmpEnv}{\letx{y_A}{\term}{\termB_A} \colon \qbit^{\otimes n}}
    \qquad\quad
    \termJC{\EmpEnv}{\letx{y_A}{\termB}{\termB_A} \colon \qbit^{\otimes n}}
  \end{align*}
  are also different since the semantics of $\termB_A$ is monic.
  Let $\rho_\term \defeq \sem{\letx{y_A}{\term}{\termB_A}}$
  and
  $\rho_\termB \defeq \sem{\letx{y_A}{\termB}{\termB_A}}$.
  Then there exists a unit vector $\ket{\varphi} \in {(\CC^2)}^{\otimes n}$
  such that
  \[
    \bra{\varphi} \rho_\term \ket{\varphi}
    \neq
    \bra{\varphi} \rho_\termB \ket{\varphi}.
  \]
  Take some unitary $U$ such that
  $U \ket{1}^{\otimes n} = \ket{\overline\varphi}$,
  where $\ket{\overline\varphi}$ is the complex conjugate of $\ket\varphi$.
  Then
  \[
    C[\cdot]
    \quad\defeq\quad
    \letx{y_A}{[\cdot]}{
      \letx{x_n}{\termB_A \otimes U (\ket{1}^{\otimes n})}{
        E_{n}
      }}
  \]
  satisfies the property.

  We prove the case when $\JG = \JQ$.
  If the previous $C[\cdot]$ cannot distinguish $\term$ and $\termB$,
  then their semantics only differ in the global phase,
  \ie, there exists some $\theta \in (0,2\pi)$ such that
  $\sem{\termB}_\Hilb = e^{i\theta} \sem{\term}_\Hilb$.
  We consider the following terms:
  \begin{align*}
    \termJQ{x \colon \qbit,\, f \colon A \li \unit}
    {\qifx{x}{f \term}{f \term}\ \colon \qbit\otimes\unit}
    \\
    \termJQ{x \colon \qbit,\, f \colon A \li \unit}
    {\qifx{x}{f \termB}{f \term}\ \colon \qbit\otimes\unit}.
  \end{align*}
  The semantics of the first line is
  \begin{align*}
    \ident_{\qbit} \otimes (f \longmapsto f \sem{M}),
  \end{align*}
  while the semantics of the second line is
  \begin{align*}
    \begin{pmatrix}
      1 & 0           \\
      0 & e^{i\theta}
    \end{pmatrix}
    \otimes (f \longmapsto f \sem{M}).
  \end{align*}
  Their semantics are different not only up to global phase,
  which can be distinguished by a context in $\CPM$.
\end{proof}

\subsection{Full Abstraction for QifUnitary with Output Test Channels}

\begin{proof}[Proof of \cref{thm:Hilb-full-abst}]
  By taking the syntactic dilation of the terms constructed in
  \cref{app:lem:(co)unit-definable}
  we construct the following terms for each $n$
  (omitting the channel names):
  \begin{align*}
     &
    \auxJQ{\EmpEnv}{H'_n \colon \qbit^{\otimes n} \otimes \qbit^{\otimes n}}{(\_, \_, \_)}
    \\ &
    \auxJQ{x_n \colon \qbit^{\otimes n} \otimes \qbit^{\otimes n}}{E'_n \colon \unit}{(\_, \_, \_)}
  \end{align*}
  such that
  $\csem{H'_n} = a_n \cdot \eta_{(2^n)}$
  and
  $\csem{E'_n} = b_n \cdot \varepsilon_{(2^n)}$
  for some $a_n, b_n > 0$.
  Also, from \cref{app:lem:type-inj-surj},
  we construct the following terms for each type $A$:
  \begin{align*}
     &
    \auxJQ{x_A \colon \qbit^{\otimes m_A}}{\term'_A \colon A}{(\_, \_, \_)}
    \\ &
    \auxJQ{y_A \colon A}{\termB'_A \colon \qbit^{\otimes n_A}}{(\_, \_, \_)}
  \end{align*}
  such that
  $\csem{\term'_A}$ is epic
  and
  $\csem{\termB'_A}$ is monic.

  For each unitary $U$, let the context $C'[\cdot]$ be
  \begin{align*}
    C'[\cdot]
    \quad\defeq\quad
    \letx{y_A}{[\cdot]}{
      \letx{x_n}{(\termB'_A \otimes U \vinchB)}{
        E'_{n}
      }}.
  \end{align*}
  Then, as in the \Qif{} case, for each $\term$ and $\termB$ such that
  \begin{align*}
    \auxJQ{\EmpEnv}{\term \colon A}{(\vinch, \voutch, \vtestch)}
    \quad\qquad
    \auxJQ{\EmpEnv}{\termB \colon A}{(\vinch, \voutch, \vtestch)},
  \end{align*}
  there exists some unitary $U$ such that
  $\csem{\,C'[\term]\,} \neq \csem{\,C'[\termB]\,}$
  if $\csem{\term} \neq \csem{\termB}$.
\end{proof}

\subsection{Calculation of the Semantics of SWITCH}

\begin{proposition}\label{prop:procedural:kraus-semantics}
  Assume
  \( \Gamma, f \colon A \multimap A' \vdash_{\JQ} M \colon B \) and $x\colon A \vdash_{\CPM} F \colon A'$.
  Let the semantics \( \sem{F} := \sem{F}_{\CPM} \) of \( F \) be given by a completely positive map that can be written as
  \begin{equation*}
    \sem{F}(\rho)
    \quad=\quad
    \sum_{i \in I} V_i \rho V_i^\dagger
    \qquad\left(\mbox{or equivalently, }
    \sem{F} = \sum_{i \in I} \EmbeddingFunctor(V_i)
    \right)
  \end{equation*}
  in the Kraus decomposition representation.
  Then
  \begin{equation*}
    \sem{\Gamma \vdash_{\JC} \letx{f}{\lambda x. F}{ M } \colon B}
    \quad=\quad
    \sum_{i \in I} \EmbeddingFunctor
    \Bigl(
    \sem{\Gamma, f \colon A \multimap A' \vdash_{\JQ} M \colon B}
    \circ
    (\ident_\Gamma \otimes \Lambda(V_i))
    \Bigr)
  \end{equation*}
  holds in \( \CPM \) (where \( \Lambda \colon \Hilb(X \otimes Y, Z) \cong \Hilb(X, Y \multimap Z) \)).
\end{proposition}
\begin{proof}
  The Kraus decomposition implies that \( \sem{F} = \sum_{i} \EmbeddingFunctor(V_i) \) holds in \( \CPM \).
  Then
  \begin{align*}
     & \textstyle
    \sum_{i \in I} \EmbeddingFunctor
    \bigl(
    \sem{\Gamma, f \colon A \multimap A' \vdash_{\JQ} M \colon B}
    \circ
    (\ident_\Gamma \otimes \Lambda(V_i))
    \bigr)
    \\ & \textstyle
    =
    \sum_{i \in I}
    \EmbeddingFunctor
    \sem{\Gamma, f \colon A \multimap A' \vdash_{\JQ} M \colon B}
    \circ
    (\ident_\Gamma \otimes \Lambda(\EmbeddingFunctor (V_i)))
    \\ & \textstyle
    =
    \EmbeddingFunctor
    \sem{\Gamma, f \colon A \multimap A' \vdash_{\JQ} M \colon B}
    \circ
    \Bigl(\ident_\Gamma \otimes \Lambda \Bigl(
    \sum_{i \in I}
    \EmbeddingFunctor (V_i) \Bigr)\Bigr)
    \\ &
    =
    \EmbeddingFunctor
    \sem{\Gamma, f \colon A \multimap A' \vdash_{\JQ} M \colon B}
    \circ
    (\ident_\Gamma \otimes \Lambda \sem{F})
    \\ &
    =
    \sem{\Gamma \vdash_{\JC} \letx{f}{\lambda x. F}{M} \colon B}
  \end{align*}
  since \( \EmbeddingFunctor \) is a strict compact closed functor and \( \CPM \) is commutative-monoid-enriched.
  (Care is required in the order in which the equations are transformed because \( \EmbeddingFunctor \) is not a commutative-monoid-enriched functor, \ie~it does not preserve sums.)
\end{proof}

\begin{proof}[Proof of \cref{app:thm:cat:switch-semantics}]
  Let \( \mathit{SWITCH} \) be the term in \cref{eg:procedural:switch}.

  We describe \( \sem{\letx{f}{\lambda x. F}{\letx{g}{\lambda x. G}{\mathit{SWITCH}}}} \), using Kraus decompositions \( \sem{F} = \sum_{i \in I} \EmbeddingFunctor(V_i) \) and \( \sem{G} = \sum_{j \in J} \EmbeddingFunctor(W_j) \).
  First, we compute the value of \( \sem{\mathit{SWITCH}} \circ (\ident \otimes \Lambda(V) \otimes \Lambda(W)) \) when \( V, W \in \Hilb(\sem{A}, \sem{A}) \).
  Instead of directly calculating its value, we give an intuitive explanation.
  The morphism in question is the interpretation of \( (\letx{f}{V}{\letx{g}{W}{\mathit{SWITCH}}}) \)%
  \footnote{To be precise, the maps $V$ and $W$ are not unitaries in general,
    so this program is not defined in \( \Qif \).
    Here, we are rather considering an extension of \( \Qif \)
    with all morphisms in \( \Hilb \) added as constant operations,
    whose semantics is canonically defined.
  }%
  , and since the expression involves only terms from the $\Hilb$ layer, the basic \( \beta \)-rule for \( \keyword{let} \) applies.
  Therefore,
  \begin{align*}
    \sem{\mathit{SWITCH}} \circ (\ident \otimes \Lambda(V) \otimes \Lambda(W))
     & \quad=\quad
    \ketbra{1}{1} \otimes (V \circ W) + \ketbra{0}{0} \otimes (W \circ V).
  \end{align*}
  By applying \cref{prop:procedural:kraus-semantics} twice,
  \begin{align*}
    \sem{\letx{f}{\lambda x. F}{\letx{g}{\lambda x. G}{\mathit{SWITCH}}}}
     & \textstyle
    =
    \sum_{i \in I} \sum_{j \in J} \EmbeddingFunctor\Bigl(\sem{\mathit{SWITCH}} \circ (\ident \otimes \Lambda(V_i) \otimes \Lambda(W_j)) \Bigr)
    \\ & \textstyle
    =
    \sum_{i \in I} \sum_{j \in J} \iota\Bigl(\ketbra{1}{1} \otimes (V_i \circ W_j) + \ketbra{0}{0} \otimes (W_j \circ V_i)\Bigr).
  \end{align*}
  This is the same as the definition of the quantum SWITCH~\cite{Chiribella2013}, as expected.
\end{proof}

\subsection{Proof of QC-QC Inclusion}
\label{app:qcqc}

We prove \cref{thm:qcqc-core} by constructing circuits with positions at
which externally supplied operations are called. We call these positions
\emph{slots}. Quantum data may pass through the slots in different orders
in superposition. The circuit must still call each slot once for every choice of order. These are
\emph{quantum circuits with quantum control of causal order} (QC-QCs).

The construction has two parts. First, we handle programs in the
measurement-free language \QifUnitary{} that contain no internal function
abstractions. The precise \emph{first-order fragment} is defined below;
its function arguments take quantum data to quantum data.
We show that QC-QCs are closed under the operations used to assemble these
programs. Second, we reduce a general term of \cref{thm:qcqc-core} to this
fragment. Syntactic dilation replaces measurements by interactions with
additional qubits, and administrative normalisation removes internal
function abstractions and bindings without executing external calls.
We recall these transformations when they are needed in the second part.
The additional qubits introduced by dilation are prepared at the circuit
input and discarded at its output.

\subsubsection{QC-QC and the First-Order Fragment of \QifUnitary{}}
\label{app:qcqc:qcqc}
\label{app:qcqc:circuits}

We first recall the definition of QC-QCs~\cite{Wechs2021}.
In a QC-QC, we have quantum data which we call \emph{routing control}.
This routing control is a quantum register that coherently selects the order in which the external
operations are called as the control qubit of $\qif$-statements does.
A QC-QC with $N$ slots has $N$ stages, each of which calls one of the slots.
Before the $t$-th stage, the routing control records
\emph{the set $K$ of slots that have been called} and \emph{the slot $k$ that is to be called next}.
That is, if $S$ is the set of slots, then the routing control has an orthonormal basis $\ket{K,k}$ indexed by
$\{ (K, k) \in \mathcal P(S) \times S \mid |K| = t - 1, k \in S \setminus K \}$.
At the next stage, the slot $k$ is called,
and the routing control is updated to record $K \sqcup \{k\}$ and a new slot $k'$ to be called at time $t + 1$.

\begin{definition}[QC-QC]\label{def:qcqc-circuit}
  Fix a finite slot set $S$, put $N=|S|$, and let $P,F$ be the global
  past (input) and future (output) registers. For each $k\in S$, let $I_k,O_k$ be the
  input and output ports of slot $k$. All register spaces are finite
  dimensional. For $N\ge1$, a QC-QC consists of
  memory registers $H_1,\ldots,H_N$, an auxiliary output register $R$,
  and known internal linear maps $V_1,\ldots,V_{N+1}$, as follows.

  At query stage $t\in\{1,\ldots,N\}$, the $t$-th slot call, the routing
  control has orthonormal basis states $\ket{K,k}$ indexed by
  \[
    K\subseteq S,\qquad |K|=t-1,\qquad k\in S\setminus K.
  \]
  Define the spaces before and after query $t$ by
  \[
    \begin{aligned}
      \mathcal I_t
       & =\operatorname{span}\{\ket{K,k}\otimes v\otimes h:
      v\in I_k,\ h\in H_t\},                                \\
      \mathcal O_t
       & =\operatorname{span}\{\ket{K,k}\otimes w\otimes h:
      w\in O_k,\ h\in H_t\},
    \end{aligned}
  \]
  with $(K,k)$ ranging over the indices above and differently labelled
  components orthogonal.
  For a family $\phi=(\phi_k\colon I_k\to O_k)_{k\in S}$ of inserted
  linear maps, query $t$ is
  \begin{align*}
    Q_t^\phi \colon\mathcal I_t\longrightarrow\mathcal O_t,
    \quad
    \ket{K,k}\otimes v\otimes h
    \longmapsto\ket{K,k}\otimes\phi_kv\otimes h.
  \end{align*}
  The same $\phi_k$ is used at every occurrence of label $k$.

  The internal maps, independent of $\phi$, have types
  \[
    \begin{aligned}
      V_1     & \colon P\longrightarrow\mathcal I_1,                            \\
      V_{t+1} & \colon\mathcal O_t\longrightarrow\mathcal I_{t+1}
              &                                                   & (1\le t<N), \\
      V_{N+1} & \colon\mathcal O_N\longrightarrow F\otimes R.
    \end{aligned}
  \]
  For $1\le t<N$, $V_{t+1}$ maps the component labelled $(K,k)$ into
  the span of components labelled
  \[
    (K\cup\{k\},k'),\qquad k'\in S\setminus(K\cup\{k\}).
  \]
  The \emph{effective input space} of $V_1$ is $P$; that of $V_{t+1}$,
  for $1\le t\le N$, is the space spanned by outputs after the first $t$
  queries,
  \[
    \operatorname{span}_{\phi}\operatorname{im}
    (Q_t^\phi V_t\cdots Q_1^\phi V_1).
  \]
  Each internal map is an isometry on its effective input space.

  The circuit operator for inserted linear maps is the composition
  \[
    C(\phi)=V_{N+1}\,Q_N^\phi\,V_N\cdots V_2\,Q_1^\phi\,V_1
    \colon P\longrightarrow F\otimes R.
  \]
  Its induced map on a density operator $\rho$ on $P$ is
  \[
    \rho\longmapsto
    \operatorname{tr}_R\bigl(C(\phi)\rho C(\phi)^\dagger\bigr),
  \]
  where $\dagger$ is the adjoint and $\operatorname{tr}_R$ is the partial
  trace discarding the \emph{final ancilla} $R$.
  Write $\mathcal L(H)$ for the space of operators on a register $H$.
  An inserted CP map in slot $k$ goes from $\mathcal L(I_k)$ to
  $\mathcal L(O_k)$.
  A map $\rho\mapsto\phi_k\rho\phi_k^\dagger$ is \emph{pure},
  and since every CP map is a finite sum of pure maps,
  insertion of arbitrary CP maps is defined by additivity in each slot.
  \thmend
\end{definition}

For $N\ge1$, the execution order is $V_1$, $Q_1^\phi$, $V_2$, \ldots,
$Q_N^\phi$, $V_{N+1}$. The first map has control values
$\ket{\emptyset,k}$ and selects the first call. Each query sends only its
port through the selected operation; the other data stay in $H_t$.
The next internal map processes the returned data, adds the just-called
slot to the used set, and selects an unused slot.
The selection and update act coherently on superpositions of control
values. After $N$ calls every slot has been used, and the last map processes
the remaining control and memory into $F$ and $R$. Only $F$ is returned;
$R$ collects the additional output data that are discarded.

We now specify the programs to which we can apply the circuit construction
immediately.
We call a type \emph{ground}
when it is first-order and purely quantum:
its data are qubits and unit values combined by tensor products.

The \emph{first-order fragment} of \QifUnitary{} considered here consists of
well-typed terms with no lambda abstractions, with globally defined names replaced by their
bodies, and with ground result type. Their ordinary free variables have
context
\[
  (f_i \colon A_i\li B_i)_{i\in S},\ (q_j \colon C_j)_j,
\]
where $A_i,B_i,C_j$ are ground. The variables $f_i$ name the external
operations, and $q_j$ name quantum input data. A function type $A_i\li B_i$
here means an operation taking data of type $A_i$ to data of type $B_i$.
Every application has an external variable $f_i$
or a known unitary in function position, and every other subterm has ground
type.
Therefore, a term in this fragment can be thought of
as a direct description of a quantum circuit with holes for the external operations.

We prove that QC-QCs are closed under the basic constructs in the first-order fragment.
The basic constructs include calls to known unitary gates or the external operations,
tensor composition, sequential composition, and quantum conditionals.
For a family of linear maps
$\phi=(\phi_i\colon\sem{A_i}\to\sem{B_i})_{i\in S}$, let
$\sem T^{\phi}$ be the linear Hilbert-space interpretation of a term $T$
with every external call $f_i$ interpreted as $\phi_i$.

\begin{proposition}\label{prop:qcqc-circuits}
  Every term $T$ in this fragment has a QC-QC with the same slots and
  input/output registers, including its channel qubits, and with a
  one-dimensional final ancilla. Its circuit operator equals
  $\sem T^{\phi}$ for every family of linear slot maps $\phi$.
\end{proposition}

\begin{proof}
  Construct the circuits recursively. Wiring and known gates are internal
  operations, and an external call adds its slot after the circuit for its
  argument. Composition and tensor combine disjoint slot sets, carrying
  unused wires in memory.

  The only nontrivial case is the quantum conditional $\qifx{L}{T_1}{T_0}$.
  Run $L$ first, and let $c$ be the condition qubit it returns.
  By linearity, the branches share slots $S$, disjoint from the slots
  of $L$. By typing, they also share input/output registers $P,F$.
  Put $N=|S|$, and write $(H_t^b,V_t^b,R^b=\CC)$ for the circuit data
  of $T_b$.
  Choose isometries $J_t^b\colon H_t^b\to D_t$ into common registers,
  and put $H_t=\CC^2_c\otimes D_t$ for $1\le t\le N$, and $R=\CC$.
  Let $\iota_t^b\colon\mathcal I_t^b\to\mathcal I_t$ and
  $\omega_t^b\colon\mathcal O_t^b\to\mathcal O_t$ both act by
  \[
    \ket{K,k}\otimes z\otimes h
    \longmapsto
    \ket{K,k}\otimes z\otimes(\ket b_c\otimes J_t^b h).
  \]
  For $N\ge1$, the circuit from $\CC^2_c\otimes P$ to
  $\CC^2_c\otimes F$ has internal maps
  \[
    \begin{aligned}
      V_1     & =\sum_{b=0}^1\iota_1^b V_1^b
      (\bra b_c\otimes\ident_P),                                     \\
      V_{t+1} & =\sum_{b=0}^1\iota_{t+1}^b V_{t+1}^b
      (\omega_t^b)^\dagger
              &                                        & (1\le t<N), \\
      V_{N+1} & =\sum_{b=0}^1(\ket b_c\otimes\ident_F)
      V_{N+1}^b(\omega_N^b)^\dagger.
    \end{aligned}
  \]
  For $N=0$, use the single known map
  $V_1=\sum_b\ketbra b b_c\otimes V_1^b$.
  The embeddings preserve the labels $(K,k)$, so the used-set updates
  are inherited from the branches. The images for different $b$ are
  orthogonal in $c$, so these maps are isometric on their effective
  inputs.
  The embeddings leave the slot ports unchanged, so
  $Q_t^\phi\iota_t^b=\omega_t^b Q_t^{b,\phi}$. Since
  $(\omega_t^b)^\dagger\omega_t^b=\ident$, the
  intermediate embeddings cancel. The circuit operator is therefore
  $\sum_b\ketbra b b_c\otimes\sem{T_b}^\phi$.
  Composing with the circuit for $L$ gives the required denotation.
  The condition qubit is returned, and no additional output is
  discarded.
\end{proof}

\subsubsection{Dilation and Normalisation}
\label{app:qcqc:dilation}
\label{app:qcqc:normal-form}

We next remove the restriction on the program's internal function syntax.
Syntactic dilation translates a source term $M$ and its known global
definitions $\Defs$ to measurement-free terms
$\widehat M,\widehat\Defs$ in \QifUnitary{}.
The relation $M\SynDil\widehat M$ means that $\widehat M$ is such a
translation, and likewise for $\Defs\SynDil\widehat\Defs$.
The extra input and output channels name the auxiliary qubits needed for
this translation.
Let us state that the claim of \cref{thm:syn-dil:sem-pres}
can be strengthened to open terms with purely quantum context and result.

\begin{lemma}\label{lem:qcqc-open-dilation}
  Let $\Defs:\Gamma$ be well-typed global definitions, and let
  $\termJC{\Delta}{M:Y}$ with $\Delta,Y$ purely quantum.
  If $\Defs\SynDil\widehat\Defs$ and $M\SynDil\widehat M$, then
  $\sem M_{\sem\Defs}=\csem{\widehat M}$, with $\widehat M$ interpreted under
  $\sem{\widehat\Defs}$.
  \thmend
\end{lemma}

It remains to remove internal functions from the dilated term without
performing its circuit operations.
To this end, we expand known global definitions, apply the beta and let rules that
substitute arguments or bound values, and use commuting conversions to
move bindings and output assignments through surrounding term constructs.
We call these steps \emph{administrative reduction}.
These rules act in arbitrary contexts, including under abstractions and
inside conditional branches.
By this process we can reach a term to which none of these steps applies.

\begin{lemma}\label{lem:qcqc-normal-form}
  Every well-typed \QifUnitary{} term with ground result and ordinary
  context as in the first-order fragment has an administrative normal form
  in that fragment, with the same type, channel qubits, and open Hilbert
  denotation.
\end{lemma}

\begin{proof}
  We use termination to reach a normal form, then the type restriction to
  rule out residual abstractions. Administrative steps are subcases of the
  rules of \cref{thm:opsem:termination}, which rules out infinite reduction sequences
  on open terms in arbitrary contexts, so a normal form exists.

  Inspect the normal form from the outside inward.
  At each recursive step,
  the term being inspected and every variable introduced so far are ground;
  function positions are classified separately. Initially the result is
  ground by assumption, and the context consists of ground variables and
  the first-order slots.
  In an application, follow nested function positions to the first
  expression that is not itself an application. A lambda there would give
  a beta redex; a binding, sequencing, or output assignment would move out.
  Pairs, unit values, input-channel names, and the ground data variables cannot
  have function type, and quantum
  conditionals have ground results. Thus this expression is a slot or a
  known unitary, which takes one ground argument and returns ground data.
  Consequently every application is
  a call of a slot or gate with ground argument and result.

  For a tensor-let, inspect the term being decomposed before entering its
  body. An explicit pair is decomposed by the tensor-let rule.
  A binding, sequencing, or output assignment is moved outward by a
  commuting conversion.
  The remaining possible forms are a data variable,
  an application of a slot or gate, or a quantum conditional; each is ground
  in the current context.
  A slot variable itself has function type, not tensor type.
  Hence the variables introduced by the tensor-let are ground,
  preserving the invariant for its body.
  Pair components, the inputs and continuations of sequencing and output
  assignments, tensor-let bodies, and the conditions and branches of
  quantum conditionals all have ground type by typing.
  Thus no abstraction can occur at a ground result or argument. Ordinary lets and global names have been
  eliminated, so the result is in the stated fragment.
\end{proof}

\begin{proof}[Proof of \cref{thm:qcqc-core}]
  Let $x\colon X$ be fresh and put $N\defeq M\,x$. By the $\eta$-law,
  $M$ and $\lambda x.N$ have the same denotation.
  Dilate $N$ and $\Defs$ by \cref{thm:syn-dil:type-safety}; since the
  context and result of $N$ are purely quantum, this interface is
  unchanged. Then \cref{lem:qcqc-normal-form} applied to the dilated body
  $\widehat N$ gives a term $T$ in the first-order fragment with the same
  open Hilbert denotation as $\widehat N$.
  By \cref{prop:qcqc-circuits}, $T$ has a QC-QC with the original slots.
  Prepare its input-channel qubits in $\ket1$ as part of the first
  internal map, and take its output-channel qubits as the final ancilla
  to be discarded. The resulting QC-QC has global past $P$ and global
  future $F$.

  For pure inserted maps, \cref{prop:qcqc-circuits} gives equality of
  the open linear operators. Since $\EmbeddingFunctor$ preserves names
  and evaluation (\cref{sec:cat-sem}), this equality and
  \cref{lem:qcqc-open-dilation} identify the name of the induced map
  with the value of $\sem M_{\sem\Defs}$ on the names of the inserted
  maps. Additivity in each slot extends this equality to arbitrary CP
  maps, as in \cref{prop:procedural:kraus-semantics}.
  Thus, by the defining name equality of $\mathsf S_M$, the circuit
  induces $\mathsf S_M(\mathcal E_1,\ldots,\mathcal E_n)$.
\end{proof}

\section{Vacuum Extensions and the Ill-Definedness Issue}
\label{app:sec:vacuum}

\citet{Barsse2026} also defined a programming language
in which arbitrary quantum channels may appear under quantum control.
Their language has classical recursion, but in order to compare it with other languages, we focus on the fragment without recursion.

In order to define a semantics of general quantum-controlled
channels in their language,
they use the idea of \emph{vacuum extension}~\cite{ChiribellaK19-quantum-shannon,Kristjansson2020}.
Roughly speaking, rather than taking the semantics
to be only a CPTP map from the input space $\mathcal H_\text{in}$
to the output space $\mathcal H_\text{out}$,
they consider a map
$\mathcal H_\text{in} \oplus \CC
  \longrightarrow
  \mathcal H_\text{out} \oplus \CC$,
namely, some interaction with the \emph{vacuum space} $\CC$.

In fact, we show that incorporating interactions with the vacuum space into the semantics is not fundamentally different from considering interactions with auxiliary qubits.
Moreover, our observations reveal that their semantics corresponds to the most trivial case considered in the literature~\cite{Dong2019,Abbott2020}.
Specifically, in this section, we define a translation from their language into \QifUnitary{} and compare it with other languages.

\newcommand{\vacctrl}{\mathrm{vacuum}\text{-}\mathrm{control}}

\subsection{Syntax}

We define a language \Vac{}
in \cref{fig:vacuum-syntax},
which is essentially the same as
the fragment of the language defined in \citet{Barsse2026} that excludes recursion.
\begin{mathfig}
  \begin{align*}
    \begin{array}{llrl}
      \textit{Commands}
       & S & \Coloneqq &
      \noop
      \sor \clet{q}{\ket 0}
      \sor \bdiscard q
      \sor U[\vec{q}]
      \sor S_1; S_2
      \sor \qifx{q}{S_1}{S_2}
    \end{array}
  \end{align*}
  \begin{proofrules}
    \infer*{}{
      \Gamma \vdash \noop \after \Gamma
    }

    \infer*{}{
      \Gamma \vdash \clet{q}{\ket 0} \after \Gamma, q
    }

    \infer*{}{
      \Gamma, q \vdash \bdiscard q \after \Gamma
    }

    \infer*{}{
      \Gamma, \vec q \vdash U[\vec q] \after \Gamma, \vec q
    }

    \infer*{
      \Gamma \vdash S_1 \after \Gamma'
      \\
      \Gamma' \vdash S_2 \after \Gamma''
    }{
      \Gamma \vdash S_1; S_2 \after \Gamma''
    }

    \infer*{
      \Gamma \vdash S_1 \after \Gamma'
      \\
      \Gamma \vdash S_2 \after \Gamma'
      \\
      q \text{ does not appear in } S_1, S_2
    }{
      \Gamma, q \vdash
      \qifx{q}{S_1}{S_2}
      \after \Gamma', q
    }
  \end{proofrules}
  \caption{Syntax and derivation rules of \Vac}
  \label{fig:vacuum-syntax}
  \Description{}
\end{mathfig}

Their language includes $\noop$ (the skip command), $\clet{q}{\ket0}$ (qubit initialisation), $\bdiscard q$ (qubit discarding), $U[\vec q]$ (unitary application), $S_1;S_2$ (sequencing), and $\qifx{q}{S_1}{S_2}$ (a quantum conditional statement).

In their derivation rules, the context $\Gamma$
simply represents a set of qubit names.
A derivation $\Gamma \vdash S \after \Gamma'$ states that,
if the qubits in $\Gamma$ are available before executing $S$,
then the qubits in $\Gamma'$ are available afterwards.

Measurements with classical conditional branching can be defined as syntactic sugar for
the following:
\[
  \meas q \ (0 \rightarrow S_0, 1 \rightarrow S_1)
  \quad\defeq\quad
  \clet{q'}{\ket0}; \gCX[q,q']; \bigl(\qifx{q'}{S_1}{S_0}\bigr); \bdiscard q'.
\]

\subsection{Denotational Semantics}

As noted above, their semantics considers interactions with
the vacuum space.
Their categorical semantics is defined in
the following category $\CQC$.

\begin{definition}
  The category $\CQC$ has natural numbers as objects
  and morphisms $(F, K) \colon\allowbreak n \longrightarrow m$ where
  $F \colon \Mat_{n}(\CC) \longrightarrow \Mat_{m}(\CC)$
  is a completely positive trace preserving map,
  and
  $K \colon \CC^n \longrightarrow \CC^m$
  is a linear map such that
  the following map
  $\Mat_{n + 1}(\CC) \longrightarrow \Mat_{m + 1}(\CC)$
  \begin{align*}
    \left(
    \begin{array}{c|c}
      \rho      & v      \\ \hline
      w^\dagger & \alpha
    \end{array}
    \right)
    \quad \longmapsto \quad
    \left(
    \begin{array}{c|c}
      F(\rho)             & K v    \\ \hline
      w^\dagger K^\dagger & \alpha
    \end{array}
    \right)
  \end{align*}
  defines another completely positive trace preserving map.
  The composition is defined componentwise,
  \ie,
  $(G, L) \circ (F, K) = (GF, LK)$.
  \thmend
\end{definition}

The semantics of $\qif$ is defined via the following
$\vacctrl$ operation,
which makes use of the linear maps $K$ and $L$.

\begin{definition}
  Let $(F, K), (G, L) \in \CQC(n,m)$.
  We define a map
  $\vacctrl((F, K), (G, L))
    \in \CQC(2 n, 2 m)$ by the pair
  $
    (
    H,
    \ketbra{0}{0} \otimes K
    + \ketbra{1}{1} \otimes L
    ),
  $
  where $H$ is defined by
  \begin{align*}
     &
    \ketbra00_q \cdot \ketbra00_q
    \otimes F
    \ +\
    \ketbra11_q \cdot \ketbra11_q
    \otimes G
    \\ &
    \ +\
    \ketbra00_q \cdot \ketbra11_q
    \otimes K (\cdot) L^\dagger
    \ +\
    \ketbra11_q \cdot \ketbra00_q
    \otimes L (\cdot) K^\dagger
    .
    \tag*{\thmend}
  \end{align*}
\end{definition}

The categorical semantics of \Vac{} is defined in
\cref{fig:vacuum-cat-sem}.
\begin{mathfig}
  \begin{align*}
    \sem{\noop}_\Gamma
     & = (\ident, \ident)
    \\
    \sem{\clet{q}{\ket0}}_\Gamma
     & = (\ketbra00_q, \ket0_q)
    \\
    \sem{\bdiscard q}_{\Gamma, q}
     & = (\trace_q, \bra0_q)
    \\
    \sem{U[\vec q]}_{\Gamma, \vec q}
     & = (U_{\vec q} (\cdot) U^\dagger_{\vec q}, U_{\vec q})
    \\
    \sem{S_1; S_2}_{\Gamma}
     & = \sem{S_2}_{\Gamma'} \circ \sem{S_1}_\Gamma
    \\
    \sem{\qifx{q}{S_1}{S_2}}_{\Gamma, q}
     & = \vacctrl(\sem{S_2}, \sem{S_1})
  \end{align*}
  \caption{Categorical semantics of \Vac{} in $\CQC$}
  \label{fig:vacuum-cat-sem}
  \Description{}
\end{mathfig}

The semantics is canonical except for the discarding operation.
The semantics $\sem{\bdiscard q}$ is defined by $(\trace_q, \bra0_q)$,
but in fact one may instead choose any unit vector $\ket\varphi$
and define the semantics as $(\trace_q, \bra\varphi_q)$.
This is the non-canonical semantic choice
that we examine in this paper.
This choice is respected in the semantics of $\qif$.
For example,
$\bdiscard q$ and $(\gH[q]; \bdiscard q)$
are observably distinguishable
as their semantics differ in the second component.

We define another denotational semantics
of \Vac{} based on Kraus operators
to make its difference from other languages clear.
In \cref{fig:vacuum-Kraus-sem},
we define \emph{Kraus semantics} $\ksem{S}$
by a multiset of linear operators
with a \emph{specified 0-th element}.
That is,
the semantics $\ksem{S}$
will be defined as a list of Kraus operators
$\{K_0,K_1,\dots, K_n\}$,
where the 0-th map $K_0$ is fixed,
but the order of the rest does not matter.
\begin{mathfig}
  \begin{align*}
    \ksem{\noop\,}_\Gamma
     & = \{ \ident \}
    \\
    \ksem{\clet{q}{\ket0}\,}_\Gamma
     & = \{ \ket0_q \}
    \\
    \ksem{\bdiscard q\,}_{\Gamma, q}
     & = \{ \bra0_q, \bra1_q  \}
    \\
    \ksem{\,U[\vec q\,]}_{\Gamma, \vec q}
     & = \{ U_{\vec q} \}
  \end{align*}
  \begin{proofrules}
    \infer*{
      \ksem{S_1}_{\Gamma}
      = \{ K_0, K_1, \dots, K_n  \}
      \\
      \ksem{S_2}_{\Gamma'}
      = \{ L_0, L_1, \dots, L_m  \}
    }{
      \ksem{\,S_1; S_2\,}_\Gamma
      = \{ L_0K_0, L_1K_0, \dots, L_mK_n \}
    }

    \infer*{
      \ksem{S_1}_{\Gamma}
      = \{ K_0, K_1, \dots, K_n  \}
      \\
      \ksem{S_2}_{\Gamma}
      = \{ L_0, L_1, \dots, L_m  \}
    }{
      {
          \begin{array}{ll}
            \ksem{\,\qifx{q}{S_1}{S_2}\,}_{\Gamma, q}
            = \{ \  &
            \ketbra00\otimes L_0
            + \ketbra11\otimes K_0,
            \hspace{7em}
            \\ &
            \ketbra00\otimes L_1,
            \dots,
            \ketbra00\otimes L_m,
            \\ &
            \ketbra11\otimes K_1,
            \dots,
            \ketbra11\otimes K_n
            \ \ \}
          \end{array}
        }
    }
  \end{proofrules}
  \caption{Kraus semantics of \Vac{}}
  \label{fig:vacuum-Kraus-sem}
  \Description{}
\end{mathfig}

In the Kraus semantics of $\qif$, we can clearly see that
the 0-th Kraus operator is treated differently:
$K_0$ and $L_0$ are coherently superposed,
but the other combinations $K_i$ and $L_j$ are not.

We can show that our Kraus semantics actually ``implements'' the original semantics.
\begin{theorem}
  For each well-formed command $S$ in \Vac,
  let $\sem{S} = (F, K)$
  and $\ksem{S} = \{K_0, K_1, \dots, K_n\}$.
  Then, $F = \sum_{i=0}^n K_i (\cdot) K_i^\dagger$
  and $K = K_0$.
\end{theorem}
\begin{proof}
  By a straightforward induction on $S$.
\end{proof}

\subsection{Syntactic Dilation}

Every channel admits a purification.
That is, by Stinespring dilation, a channel can be represented by a unitary circuit with auxiliary input qubits initialised to $\ket1$ and auxiliary output qubits discarded at the end.
In particular, a channel represented by Kraus operators $\{K_i\}_{i \in I}$ admits such a purification in which, if the auxiliary qubits are measured in the computational basis rather than discarded, the possible outcomes are $\{ o_i \}_{i \in I}$ and the branch with outcome $o_i$ implements $K_i$.
Since the Kraus semantics of \Vac{} selects one Kraus operator $K_0$ from the Kraus decomposition $\{ K_i \}_{i = 0}^n$, it can be regarded as a purification of the channel that singles out the case in which all auxiliary-qubit measurement outcomes are $1$.
The following circuit represents this situation,
where $\Gamma \vdash S \after \Gamma'$ and $\sem{S} = (F,K)$:
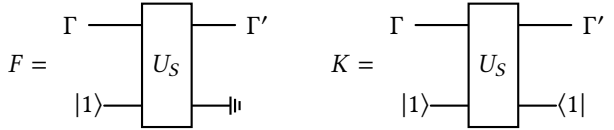
\[
  F =
  \begin{quantikz}[yscale=0.5]
    \lstick{$\Gamma$} & \gate[2]{U_S} & \rstick{$\Gamma'$} \\
    \ket1 & &  \ground{}
  \end{quantikz}
  \qquad
  K =
  \begin{quantikz}[yscale=0.5]
    \lstick{$\Gamma$} & \gate[2]{U_S} & \rstick{$\Gamma'$} \\
    \ket1 & & \bra1
  \end{quantikz}.
\]

In the interpretation of $\qif$ under the Kraus semantics,
superposition with the control qubit occurs only for the $K_0$--$L_0$ pair;
in all other cases, the control qubit is projected.
This corresponds to using auxiliary qubits that are completely separated
between the then- and else-branches.
If all auxiliary qubits are measured as $1$ at the end,
then the branch taken cannot be determined,
so superposition with the control qubit is retained.
Otherwise, for example, if at least one auxiliary qubit used in the then-branch
is measured as $0$, then the control qubit is known to be $\ket1$,
and hence it is projected.
Therefore, this semantics corresponds to a semantics
without the index correspondence described in our overview.

Using this idea, we define a semantics-preserving translation from $\Vac$
to $\QifUnitary$ in \cref{fig:vacuum-qifunitary}.
The syntactic dilation is defined as the relation
$\Gamma \vdash S \after \Gamma' \SynDil \term \withaug (\vinch, \voutch)$
so that the resulting term has a type derivation
$\auxJQ[\cdot]{q_1 \colon \qbit, \dots, q_n \colon \qbit}{\term \colon \qbit^{\otimes |\Gamma'|}}{(\vinch,\voutch)}$
where $\Gamma = q_1,\dots, q_n$.
\begin{mathfig}
  \begin{proofrules}
    \infer*{}{
      \Gamma \vdash \noop \after \Gamma
      \SynDil
      \textstyle\bigotimes_{q \in \Gamma} q
      \withaug (\cdot, \cdot)
    }

    \infer*{}{
      \Gamma \vdash \clet{p}{\ket0} \after \Gamma, p
      \SynDil
      \bigl(\textstyle\bigotimes_{q \in \Gamma} q\bigr) \otimes \gX\inch
      \withaug (\inch, \cdot)
    }

    \infer*{}{
      \Gamma, p \vdash \bdiscard p \after \Gamma
      \SynDil
      (\outch \assign \gX p);
      \textstyle\bigotimes_{q \in \Gamma} q
      \withaug (\cdot, \outch)
    }

    \infer*{}{
      \Gamma, \vec p \vdash U[\vec p] \after \Gamma, \vec p
      \SynDil
      \bigl(\textstyle\bigotimes_{q \in \Gamma} q\bigr)
      \otimes U (\bigotimes \vec p)
      \withaug (\cdot, \cdot)
    }

    \infer*{
      \Gamma \vdash S_1 \after \Gamma'
      \SynDil \term
      \withaug (\vinch, \voutch)
      \\
      \Gamma' \vdash S_2 \after \Gamma''
      \SynDil \termB
      \withaug (\vinchB, \voutchB)
    }{
      \Gamma \vdash S_1; S_2 \after \Gamma''
      \SynDil
      \textstyle
      \letx{\bigotimes_{q \in \Gamma'}q}{\term}{\termB}
      \withaug (\vinch\vinchB, \voutch\voutchB)
    }

    \infer*{
      \Gamma \vdash S_1 \after \Gamma'
      \SynDil \term_1 \withaug (\vinch_1, \voutch_1)
      \\
      \Gamma \vdash S_2 \after \Gamma'
      \SynDil \term_2 \withaug (\vinch_2, \voutch_2)
      \\\\
      |\vinch_1| - |\voutch_1|
      =
      |\vinch_2| - |\voutch_2|
      =
      |\voutchB| - |\vinchB|
    }{
      {
          \begin{aligned}
             &
            \Gamma, q \vdash \qifx{q}{S_1}{S_2} \after \Gamma', q
            \\ &
            \SynDil
            \textstyle
            \qifx{q}
            {(\voutch_2\voutchB \assign \vinch_2\vinchB); \term_1}
            {(\voutch_1\voutchB \assign \vinch_1\vinchB); \term_2}
            \withaug
            (\vinch_1\vinch_2\vinchB, \voutch_1\voutch_2\voutchB)
          \end{aligned}
        }
    }
  \end{proofrules}
  \caption{Syntactic dilation for \Vac{}}
  \label{fig:vacuum-qifunitary}
  \Description{}
\end{mathfig}

The difference between the syntactic dilations of \Qif{} and \Vac{}
appears in the treatment of $\qif$.
In \Qif, we require the branches to use the same i/o channels,
whereas in \Vac, we require them to use distinct i/o channels.
Each input channel not used in the term $\term_i$ is immediately
assigned to output channels to be discarded with the default value $\ket1$.
Therefore, if $0$ is measured in any of the auxiliary outputs,
the corresponding branch has been selected classically
and the control qubit is projected.

\begin{theorem}
  The syntactic dilation preserves the semantics.
  That is, if $\sem{S} = (F, K)$ and $S \SynDil \term \withaug (\vinch, \voutch)$,
  then $\csem{\term} = F$.
  Moreover, $K$ can be recovered as follows:
  \[
    K =
    (\ident_{\Gamma'} \otimes \bra1_{\voutch})
    \circ \sem{\term}
    \circ (\ident_\Gamma \otimes \ket1_{\vinch}).
  \]
\end{theorem}
\begin{proof}
  By induction on the structure of $S$.
  The only non-trivial case is $\qif$.
  Let $\ksem{S_1} = \{K_0,\dots, K_n\}$,
  $\ksem{S_2} = \{L_0,\dots, L_m\}$
  and let $\bigl(\qifx{q}{S_1}{S_2}\bigr) \SynDil M \withaug
    (\vinch_1\vinch_2\vinchB, \voutch_1\voutch_2\voutchB)
  $.
  Then the semantics is:
  \begin{align*}
     &
    (\ident_{\Gamma',q} \otimes \discardDiag_{\voutch_1\voutch_2\voutchB})
    \circ \iota\sem{\term}
    \circ (\ident_{\Gamma,q} \otimes \iota\ket1_{\vinch_1\vinch_2\vinchB})
    \\ &
    =
    \textstyle
    (\ident \otimes \discardDiag_{\voutchB})
    \circ
    \sum_{b_1 \in \{0,1\}^{\voutch_1}}
    \sum_{b_2 \in \{0,1\}^{\voutch_2}}
    \iota
    \bigl(
    (\ident_{\Gamma',q} \otimes
    \bra{b_1}_{\voutch_1} \otimes \bra{b_2}_{\voutch_2}
    \otimes \ident_{\voutchB})
    \circ \sem{\term}
    \circ (\ident_{\Gamma,q} \otimes \ket1_{\vinch_1\vinch_2\vinchB})
    \bigr).
  \end{align*}
  We calculate the term inside the parentheses as:
  \begin{align*}
     &
    (\ident_{\Gamma',q} \otimes
    \bra{b_1}_{\voutch_1} \otimes \bra{b_2}_{\voutch_2}
    \otimes \ident_{\voutchB})
    \circ \sem{\term}
    \circ (\ident_{\Gamma,q} \otimes \ket1_{\vinch_1\vinch_2\vinchB})
    \\ & \textstyle
    =
    (\ident_{\Gamma',q} \otimes
    \bra{b_1}_{\voutch_1} \otimes \bra{b_2}_{\voutch_2}
    \otimes \ident_{\voutchB})
    \\ & \textstyle \qquad
    \circ (\ketbra00_q \otimes \sem{\term_2} \otimes \ident_{\vinch_1\vinchB}
    + \ketbra11_q \otimes \sem{\term_1} \otimes \ident_{\vinch_2\vinchB})
    \\ & \textstyle \qquad
    \circ (\ident_{\Gamma,q} \otimes \ket1_{\vinch_1\vinch_2\vinchB})
    \\ & \textstyle
    =
    (
    \ketbra00_q
    \otimes
    \bra{b_2}_{\voutch_2}
    \sem{\term_2}
    \ket{1 \dots 1}_{\vinch_2}
    \otimes
    \bra{b_1}_{\voutch_1}
    \ket{1 \dots 1}_{\vinch_1\vinchB}
    )
    \\ & \textstyle \qquad
    +
    (
    \ketbra11_q
    \otimes
    \bra{b_1}_{\voutch_1}
    \sem{\term_1}
    \ket{1 \dots 1}_{\vinch_1}
    \otimes
    \bra{b_2}_{\voutch_2}
    \ket{1 \dots 1}_{\vinch_2\vinchB}
    )
    \\ & \textstyle
    =
    \begin{cases}
      (\ketbra00 \otimes L_0
      + \ketbra11 \otimes K_0)
      \otimes \ket1_{\voutchB}
       & \text{if }\ b_1 = 1 \dots 1 \land b_2 = 1 \dots 1
      \\
      \ketbra00 \otimes L_{\rho(b_2)}
      \otimes \ket1_{\voutchB}
       & \text{if }\ b_1 = 1 \dots 1 \land b_2 \neq 1 \dots 1
      \\
      \ketbra11 \otimes K_{\sigma(b_1)}
      \otimes \ket1_{\voutchB}
       & \text{if }\ b_1 \neq 1 \dots 1 \land b_2 = 1 \dots 1
      \\
      0
       & \text{if }\ b_1 \neq 1 \dots 1 \land b_2 \neq 1 \dots 1
    \end{cases}
  \end{align*}
  where $K_{\sigma(b_1)}$ and $L_{\rho(b_2)}$
  are some corresponding Kraus operators.
  Therefore, the whole CPTP map admits a Kraus decomposition
  \[
    \ksem{S}
    \ =\
    \{ \ \
    \ketbra00\otimes L_0
    + \ketbra11\otimes K_0,
    \ketbra00\otimes L_1,
    \dots,
    \ketbra00\otimes L_m,
    \ketbra11\otimes K_1,
    \dots,
    \ketbra11\otimes K_n
    \ \ \}.
    \tag*{\qedhere}
  \]
\end{proof}

\subsection{Discussion}

\paragraph{Algebraic laws}
As with the quantum SWITCH,
we can easily see that the program
$\qifx{q}{S_1; S_2}{S_2; S_1}$
in \Vac{}
does not represent the quantum SWITCH in general.
This immediately follows from the fact that
its semantics satisfies the equation
\[
  \sem{\qifx{q}{S_1; S_1'}{S_2; S_2'}}
  \ =\
  \sem{\qifx{q}{S_1}{S_2}; \qifx{q}{S_1'}{S_2'}}
\]
but does not satisfy
\[
  \sem{
    \qifx{q}{S_1}{S_1}
  }
  \ =\
  \sem{S_1}.
\]
The former is not expected for the quantum SWITCH,
which is obtained when $S_1 = S_2'$ and $S_2 = S_1'$,
whereas the latter is expected because it is
a special case of the quantum SWITCH
where $S_2 = \noop$.

\paragraph{Operational semantics}
One notable feature of our syntactic dilation is that it provides an algorithmic
translation from \Vac{} to \QifUnitary{}, whose target language has an
operational semantics defined as a quantum circuit.
This kind of operational semantics, defined as a sequence of
unitary gate applications, was not made explicit in
the original paper~\cite{Barsse2026}:
they defined an operational semantics
as a transition system rather than as a quantum circuit,
and its circuit implementation was left open.

\paragraph{Classical recursion}
For the semantics of classical recursion, they choose as the 0-th Kraus operator the execution path that exits the loop in the 0-th iteration (\ie, when the condition is not met from the outset).
This means that, when classical recursion appears within a quantum conditional, the quantum branch degenerates into a classical branch if the loop is executed at least once.

Our language does not have recursion, but recursion can be added as a $\CPM$ construct.
However, if recursion is added to \Qif{}, we cannot define a simple syntactic dilation because there is no evident purified $\Hilb$-term corresponding to recursion.
Here, we outline an idea for defining an operational semantics for \Qif{} with recursion, leaving the details for future work.
This suggests the following operational strategy.
Instead of taking the syntactic dilation of a while loop itself, we can consider the syntactic dilations of its approximations:
\begin{align*}
  \while_1 \defeq \abort \qquad
  \while_{n+1} \defeq \cifx{x}{\term}{\while_{n}},
\end{align*}
each of which can be purified.
Simply taking the syntactic dilation of these approximations may still result in a nonterminating execution, since we do not know when to stop. Our idea is instead to \emph{dynamically apply} the syntactic dilation during execution.
That is, when evaluating a while loop, we unroll one iteration, take its syntactic dilation, and execute the resulting unitary.
After executing the unitary, we measure some discarded qubits.
The measurement outcome of the discarded qubits encodes the classical condition of the while loop: if it is $\btrue$, we repeat this procedure; otherwise, we mark the recursion as resolved.
\section{Ying's QuGCL and the Ill-Definedness Issue}\label{appx:sec:ying}

\emph{QuGCL}, introduced by \citet{YingYF12-first-Ying},
refined in \citet{YingYF14-alternation},
and later documented by \citet{Ying16-the-book},
is a programming language with both $\qif$ and measurement constructs.
QuGCL has two semantics:
the \emph{canonical semantics} and the \emph{generalised
semantics}.
Their semantics uses coefficients in certain linear combinations: the canonical semantics chooses fixed coefficients, whereas the generalised semantics leaves them arbitrary and includes the canonical semantics as a special case.

In this section, we give a translation from \Ying{},
which is expressively equivalent to QuGCL
(with quantum variables restricted to have qubit type
rather than general qu$d$its),
into \QifUnitary{} with divergence\footnote{
  QuGCL has $\abort$, so we also add it to the language.
} as defined in \cref{app:fig:rules:qifunitary-bra}.
Our translation is non-deterministic in the sense that it chooses an arbitrary unitary $U$.
We show that this freedom of choice coincides with the freedom of choice in the generalised semantics.
More concretely, we establish a two-way correspondence between these choices: for every choice of coefficients in \Ying{}, a suitable choice of $U$ yields a semantics-preserving translation into \QifUnitary{}; conversely, every choice of $U$ determines coefficients for which the source and target semantics coincide.
We also show that this freedom of choice corresponds to the \emph{initial state of the environment} discussed in \citet{Abbott2020}.

\subsection{Syntax of \Ying{}}
\begin{figure}
  \[
    \begin{array}{llrl}
      \textit{Qubits}               & q, p, \dots
      \\
      \textit{Measurement outcomes} & x, y, \dots
      \\
      \textit{Commands}             & S          & \Coloneqq &
      \abort \sor \noop \sor U[\vec{q}] \sor S_1; S_2
      \\ &   & \sor      &
      \qifx{q}{S_1}{S_2}
      \\ &   & \sor      &
      \cifx{(x \leftarrow \meas(q))}{S_1}{S_2}
    \end{array}
  \]
  \caption{Syntax of \Ying.}
  \label{fig:ying:syntax}
  \Description{}
\end{figure}
The syntax of \Ying{} is given in \Cref{fig:ying:syntax}.
Programs are imperative and act only on the current quantum store;
no command allocates or deletes qubits.
The command $U[\vec{q}]$
simply applies the unitary $U$ to the qubits $\vec{q}$
without consuming or renaming them.
The quantum conditional $\qifx{q}{S_1}{S_2}$ executes the controlled versions of $S_1$ and $S_2$, leaving the control qubit $q$ alive.
The classical conditional $\cifx{(x \leftarrow \meas(q))}{S_1}{S_2}$ first measures $q$, stores the classical result $\ket{0}$ or $\ket{1}$ back into $q$, and then branches.
The outcome variable $x$ records which branch was chosen.
We assume each outcome variable appears at most once in a program.
Thus, these variables cannot be accessed by subsequent commands.
Linearity of qubits is therefore enforced syntactically.
Finally, $\abort$ denotes the command which makes the execution diverge.

\subsection{Denotational Semantics of \Ying{}}
\Ying{} enjoys a two-layer semantics.
The first layer, the \emph{semiclassical semantics},
interprets a program as an operator-valued function,
while the second layer, the \emph{purely quantum semantics},
interprets a program as a completely positive map.

Intuitively, the semiclassical semantics provides a Kraus decomposition: it records which linear map corresponds to each possible measurement outcome.
The semantics of a program $S$ is given as a map
$\ysem{S} \colon \Delta{S} \to \Hilb(n,n)$
where $\Delta{S}$ represents the set of possible outcomes of a program $S$,
\ie, assignments of $\{0,1\}$ to measurement outcomes,
and $n$ is the number of qubits used in $S$.
The set $\Delta S$ is defined by
\begin{align*}
  \Delta\abort = \Delta\noop = \Delta(U[\vec{q}])
   & \ =\ \{ \emptyset \}                                          \\
  \Delta(S_1; S_2)
   & \ =\ \{ f \sqcup g  \mid f\in \Delta S_1, g \in \Delta S_2 \} \\
  \Delta(\qifx{q}{S_1}{S_2})
   & \ =\ \{ f \sqcup g  \mid f\in \Delta S_1, g \in \Delta S_2 \} \\
  \Delta(\cif (x \leftarrow \meas(q)) \ifthen S_1 \ifelse S_2)
   & \ =\ \{ (x \mapsto 1) \sqcup f \mid f\in \Delta S_1 \}
  \\ & \qquad\quad
  \cup\ \{ (x \mapsto 0) \sqcup g \mid g \in \Delta S_2 \}.
\end{align*}
Each classical outcome variable $x$ is assumed to appear only once, so it uniquely marks an occurrence of measurement.
The following defines the semiclassical semantics of \Ying{}.
\begin{align*}
  \ysem{\noop}_\emptyset
   & = \mathrm{id},
  \\
  \ysem{S_1;S_2}_{f \sqcup g}
   & = \ysem{S_2}_{g} \circ \ysem{S_1}_{f},
  \\
  \ysem{\abort}_\emptyset
   & = 0,
  \\
  \ysem{U[\vec{q}]}_\emptyset
   & = U_{\vec{q}} \otimes \mathrm{id},
  \\
  \ysem{\qifx{q}{S_1}{S_2}}_{f \sqcup g}
   & = \alpha_{S_1,f} \ketbra{0}{0} \otimes \ysem{S_2}_{g}
  + \alpha_{S_2,g} \ketbra{1}{1} \otimes \ysem{S_1}_{f},
  \\
  \ysem{\cif (x \leftarrow \meas(q)) \ifthen S_1 \ifelse S_2}_{(x\mapsto 1) \sqcup f}
   & = \ketbra{1}{1} \otimes \ysem{S_1}_{f},
  \\
  \ysem{\cif (x \leftarrow \meas(q)) \ifthen S_1 \ifelse S_2}_{(x\mapsto 0) \sqcup g}
   & = \ketbra{0}{0} \otimes \ysem{S_2}_{g}
\end{align*}
where the coefficients $\alpha_{S_i,f}$ are any complex numbers satisfying
$\sum_{f \in \Delta S_1} |\alpha_{S_1,f}|^2
  = \sum_{g \in \Delta S_2} |\alpha_{S_2,g}|^2 = 1
$.
In the canonical semantics of \citet{Ying16-the-book} these coefficients are instantiated as
\[
  \alpha_{S, f} \defeq
  \sqrt
  \frac{
    \trace( \ysem{S}_{f}^{\dagger}\ysem{S}_{f} )
  }{
    \sum_{g\in\Delta S} \trace( \ysem{S}_{g}^{\dagger}\ysem{S}_{g} )
  }.
\]
If the denominator vanishes, choose any normalised coefficient vector:
all branch operators are then zero, so this choice leaves the purely
quantum semantics of guarded composition unchanged.

The purely quantum semantics of \Ying{} is a completely positive map on density operators, obtained from the Kraus decomposition of the semiclassical semantics.
\[
  \sem{S}(\rho) = \sum_{f\in\Delta(S)} \ysem{S}_{f}\ \rho\ \ysem{S}_{f}^{\dagger}.
\]

\subsection{Syntactic Dilation}

Let $Q = \{q_1,\dots,q_n\}$ be the set of qubits in the program and
define the context $\Gamma$ to be
$q_1 \colon \qbit, \dots,\allowbreak q_n \colon \qbit$.

In \cref{fig:app:ying-syndil}, we define a translation from
\Ying{} to \QifUnitary{}.
In the rule for $\qif$, we use an arbitrary unitary $U_i$.
The following theorem shows that we can choose $U_i$
so that the semantics matches the original semantics of \Ying{}.
\begin{mathfig}
  \begin{proofrules*}
    \textstyle
    \noop \SynDil \bigotimes_{k = 1,\dots,n} q_k \withaug (\cdot,\cdot,\cdot)

    \abort \SynDil (\testch \assign \inch);\,\bigotimes_{k=1,\dots,n}q_k \withaug (\inch,\cdot,\testch)

    U[\vec q] \SynDil
    (U \vec q) \otimes \bigotimes_{k \not\in \vec q} q_k
    \withaug (\cdot,\cdot,\cdot)

    \infer*{
      S_1 \SynDil \term_1
      \withaug (\vinch_1,\voutch_1,\vtestch_1)
      \\
      S_2 \SynDil \term_2
      \withaug (\vinch_2,\voutch_2,\vtestch_2)
    }{
      \textstyle
      S_1;S_2 \SynDil
      \letx{\bigotimes_{k=1,\dots,n} q_k}{\term_1}{\term_2}
      \withaug (\vinch_1\vinch_2,\voutch_1\voutch_2,\vtestch_1\vtestch_2)
    }

    \infer*{
      S_i \SynDil \term_i \withaug (\vinch_i, \voutch_i, \vtestch_i)
      \\
      \vinch_i = \vinchB_i\vinchC_i
    }{
      {
          \begin{aligned}
             &
            \qifx{q}{S_1}{S_2} \SynDil
            \\
             &
            \qifx{
              q
            }{
              (\voutch_2 \assign U_2 \vinchB_2);
              (\vtestch_2 \assign \gX \vinchC_2);
              \term_1
            }{
              (\voutch_1 \assign U_1 \vinchB_1);
              (\vtestch_1 \assign \gX \vinchC_1);
              \term_2
            }
            \withaug (\vinch_1\vinch_2, \voutch_1\voutch_2, \vtestch_1\vtestch_2)
          \end{aligned}
        }
    }

    \infer*{
      S_i \SynDil \term_i \withaug (\vinch_i, \voutch_i, \vtestch_i)
      \\
      \vinch_i = \vinchB_i\vinchC_i
    }{
      {
          \begin{aligned}
            \cifx{(x \leftarrow \meas(q))}{S_1}{S_2} \SynDil
            \hspace{2em}
             &
            \\
            \Let {q \otimes x} = {\gCX (q \otimes \inchB)} \In
            (\outchB \assign x);
            \qif {q}
             & \ifthen
            (\voutch_2 \assign \vinchB_2);
            (\vtestch_2 \assign \gX \vinchC_2);
            \term_1
            \\
             & \ifelse
            (\voutch_1 \assign \vinchB_1);
            (\vtestch_1 \assign \gX \vinchC_1);
            \term_2
            \\ & \qquad
            \withaug (\vinch_1\vinch_2\inchB,
            \voutch_1\voutch_2\outchB,
            \vtestch_1\vtestch_2)
          \end{aligned}
        }
    }
  \end{proofrules*}
  \caption{Syntactic dilation from \Ying{} to \QifUnitary{}.}
  \label{fig:app:ying-syndil}
  \Description{}
\end{mathfig}

\begin{theorem}
  Let $S$ be a command in \Ying{}.
  Then there exists a term $\term$
  in \QifUnitary{} extended with divergence,
  as defined in \cref{app:sec:theorems-list},
  such that $S \SynDil \term \withaug (\vinch, \voutch, \vtestch)$ and the following properties hold.
  First, the purely quantum semantics is preserved, \ie,
  \[
    \sem{S} = \csem{\term}.
  \]
  Second, there exists an injective map $f \mapsto b_f$
  from $\Delta S$ to Boolean sequences such that, for every $f \in \Delta S$,
  \[
    \ysem{S}_f =
    (\ident_\Gamma \otimes \bra{b_f}_{\voutch} \otimes \bra{0}_{\vtestch})
    \circ \sem{\term}
    \circ (\ident_\Gamma \otimes \ket1_{\vinch}).
  \]
  The conditions can be represented diagrammatically as follows:
  \[
    \ysem{S}_f =
    \begin{quantikz}[row sep=0.5em]
      \lstick{$\Gamma$} & \gate[3]{\sem{\term}} & \rstick{$\Gamma$} \\
      \setwiretype{n} & & \bra{b_f}_{\voutch} \setwiretype{q}\\
      \ket1 & & \bra0_{\vtestch}\ ,
    \end{quantikz}
    \qquad
    \sem{S} =
    \begin{quantikz}[row sep=0.5em]
      \lstick{$\Gamma$} & \gate[3]{\sem{\term}} & \rstick{$\Gamma$} \\
      \setwiretype{n} & & \ground{} \setwiretype{q}\\
      \ket1 & & \bra0_{\vtestch}\ .
    \end{quantikz}
  \]
\end{theorem}
\begin{proof}
  By induction on the structure of $S$.
  \begin{description}
    \item[Case $\noop$.] Trivial.
    \item[Case $\abort$.] Trivial.
    \item[Case $ {U[\vec q]}$.] Trivial.
    \item[Case $S_1; S_2$.] If $S_i \SynDil \term_i$ and
          \begin{align*}
            \ysem{S_1}_f =
            (\ident_\Gamma
            \otimes \bra{b_f}_{\voutch_1}
            \otimes \bra{0}_{\vtestch_1})
            \circ \sem{\term_1}
            \circ (\ident_\Gamma \otimes \ket1_{\vinch_1}),
            \\
            \ysem{S_2}_g =
            (\ident_\Gamma
            \otimes \bra{b_g}_{\voutch_2}
            \otimes \bra{0}_{\vtestch_2})
            \circ \sem{\term_2}
            \circ (\ident_\Gamma \otimes \ket1_{\vinch_2}),
          \end{align*}
          then, since the semantics of $S_1;S_2$ is just a composition
          $\ysem{S_1;S_2}_{f \sqcup g} = \ysem{S_2}_g \circ \ysem{S_1}_f$,
          \begin{align*}
            \ysem{S_1;S_2}_{f \sqcup g} =
            (\ident_\Gamma
            \otimes \bra{b_f,b_g}_{\voutch_1\voutch_2}
            \otimes \bra{0}_{\vtestch_1\vtestch_2})
            \circ \sem{\term_2} \circ \sem{\term_1}
            \circ (\ident_\Gamma \otimes \ket1_{\vinch_1\vinch_2}),
          \end{align*}
          where we omit obvious permutation maps.
          The equation for the purely quantum semantics follows
          immediately from here.
    \item[Case $\qifx{q}{S_1}{S_2}$.] Assume $S_i \SynDil \term_i$ and
          the following:
          \begin{align*}
            \ysem{S_1}_f =
            (\ident_\Gamma
            \otimes \bra{b_f}_{\voutch_1}
            \otimes \bra{0}_{\vtestch_1})
            \circ \sem{\term_1}
            \circ (\ident_\Gamma \otimes \ket1_{\vinch_1}),
            \\
            \ysem{S_2}_g =
            (\ident_\Gamma
            \otimes \bra{b_g}_{\voutch_2}
            \otimes \bra{0}_{\vtestch_2})
            \circ \sem{\term_2}
            \circ (\ident_\Gamma \otimes \ket1_{\vinch_2}).
          \end{align*}
          Let $\term \defeq
            \qifx{
              q
            }{
              (\voutch_2 \assign U_2 \vinchB_2);
              (\vtestch_2 \assign \gX \vinchC_2);
              \term_1
            }{
              (\voutch_1 \assign U_1 \vinchB_1);
              (\vtestch_1 \assign \gX \vinchC_1);
              \term_2
            }
          $.
          Then
          \begin{align*}
             &
            (\ident_\Gamma
            \otimes \bra{b_f,b_g}_{\voutch_1\voutch_2}
            \otimes \bra{0}_{\vtestch_1\vtestch_2})
            \circ \sem{\term}
            \circ (\ident_\Gamma \otimes \ket1_{\vinch_1\vinch_2})
            \\ & =
            \ketbra00_q \otimes
            (\ident_\Gamma
            \otimes \bra{b_f,b_g}_{\voutch_1\voutch_2}
            \otimes \bra{0}_{\vtestch_1\vtestch_2})
            \circ \sem{
              (\voutch_1 \assign U_1 \vinchB_1);
              (\vtestch_1 \assign \gX \vinchC_1);
              \term_2
            }
            \circ (\ident_\Gamma \otimes \ket1_{\vinch_1\vinch_2})
            \\ &
            \quad +
            \ketbra11_q \otimes
            (\ident_\Gamma
            \otimes \bra{b_f,b_g}_{\voutch_1\voutch_2}
            \otimes \bra{0}_{\vtestch_1\vtestch_2})
            \circ \sem{
              (\voutch_2 \assign U_2 \vinchB_2);
              (\vtestch_2 \assign \gX \vinchC_2);
              \term_1
            }
            \circ (\ident_\Gamma \otimes \ket1_{\vinch_1\vinch_2})
            \\ & =
            \bra{b_f} U_1 \ket{1}
            \cdot
            \ketbra00_q \otimes
            (\ident_\Gamma
            \otimes \bra{b_g}_{\voutch_2}
            \otimes \bra{0}_{\vtestch_2})
            \circ \sem{
              \term_2
            }
            \circ (\ident_\Gamma \otimes \ket1_{\vinch_2})
            \\ &
            \quad +
            \bra{b_g} U_2 \ket{1}
            \cdot
            \ketbra11_q \otimes
            (\ident_\Gamma
            \otimes \bra{b_f}_{\voutch_1}
            \otimes \bra{0}_{\vtestch_1})
            \circ \sem{
              \term_1
            }
            \circ (\ident_\Gamma \otimes \ket1_{\vinch_1})
            \\ & =
            \bra{b_f} U_1 \ket{1} \cdot
            \ketbra00_q \otimes \ysem{S_2}_g
            +
            \bra{b_g} U_2 \ket{1} \cdot
            \ketbra11_q \otimes \ysem{S_1}_f
          \end{align*}
          Here, we choose $U_i$ to satisfy
          $U_i \ket{1 \cdots 1} = \sum_{f \in \Delta S_i} \alpha_{S_i,f} \ket{b_f}$;
          such a unitary exists by the normalisation condition
          $\sum_{f \in \Delta S_i} |\alpha_{S_i,f}|^2
            = 1$,
          and the fact that the correspondence $f \mapsto b_f$ is injective.
          Thus, the expression above is equal to
          \[
            \alpha_{S_1,f} \cdot
            \ketbra00_q \otimes \ysem{S_2}_g
            +
            \alpha_{S_2,g} \cdot
            \ketbra11_q \otimes \ysem{S_1}_f
            = \ysem{\qifx{q}{S_1}{S_2}}_{f \sqcup g}.
          \]
    \item[Case $\cif (x \leftarrow \meas(q)) \ifthen S_1 \ifelse S_2$.]
          Let us assume $S_i \SynDil \term_i$ and
          the following:
          \begin{align*}
            \ysem{S_1}_f =
            (\ident_\Gamma
            \otimes \bra{b_f}_{\voutch_1}
            \otimes \bra{0}_{\vtestch_1})
            \circ \sem{\term_1}
            \circ (\ident_\Gamma \otimes \ket1_{\vinch_1}),
            \\
            \ysem{S_2}_g =
            (\ident_\Gamma
            \otimes \bra{b_g}_{\voutch_2}
            \otimes \bra{0}_{\vtestch_2})
            \circ \sem{\term_2}
            \circ (\ident_\Gamma \otimes \ket1_{\vinch_2}).
          \end{align*}
          Let $S \defeq (\cif\,(x \leftarrow \meas(q)) \ifthen S_1 \ifelse S_2)$
          and $S \SynDil \term$.
          Then
          \begin{align*}
             &
            (\ident_\Gamma
            \otimes \bra{b_f}_{\voutch_1}
            \otimes \bra{1 \cdots 1}_{\voutch_2}
            \otimes \bra{0}_{\outchB}
            \otimes \bra{0}_{\vtestch_1\vtestch_2})
            \circ \sem{\term}
            \circ (\ident_\Gamma \otimes \ket1_{\vinch_1\vinch_2\inchB})
            \\ & =
            \bigl((\ident_\Gamma \otimes \bra0_\outchB)
            \circ \sem{\gCX(q,\inchB)}
            \circ (\ident_\Gamma \otimes \ket1_\inchB)\bigr)
            \\ & \quad
            \circ
            \bigl(
            (\ident_\Gamma
            \otimes \bra{b_f}_{\voutch_1}
            \otimes \bra{1 \cdots 1}_{\voutch_2}
            \otimes \bra{0}_{\vtestch_1\vtestch_2})
            \circ \sem{\qif\,q\,\ifthen\dots}
            \circ (\ident_\Gamma \otimes \ket1_{\vinch_1\vinch_2})
            \bigr)
            \\ & =
            \ketbra11_q
            \circ
            \bigl(
            (\ident_\Gamma
            \otimes \bra{b_f}_{\voutch_1}
            \otimes \bra{1 \cdots 1}_{\voutch_2}
            \otimes \bra{0}_{\vtestch_1\vtestch_2})
            \circ \sem{\qif\,q\,\ifthen\dots}
            \circ (\ident_\Gamma \otimes \ket1_{\vinch_1\vinch_2})
            \bigr)
            \\ & =
            \ketbra11_q
            \otimes
            \bigl(
            (\ident_\Gamma
            \otimes \bra{b_f}_{\voutch_1}
            \otimes \bra{1 \cdots 1}_{\voutch_2}
            \otimes \bra{0}_{\vtestch_1\vtestch_2})
            \\ & \quad
            \circ \sem{
              (\voutch_2 \assign \vinchB_2);
              (\vtestch_2 \assign \gX \vinchC_2);
              \term_1
            }
            \circ (\ident_\Gamma \otimes \ket1_{\vinch_1\vinch_2})
            \bigr)
            \\ & =
            \ketbra11_q
            \otimes
            \bigl(
            (\ident_\Gamma
            \otimes \bra{b_f}_{\voutch_1}
            \otimes \bra{0}_{\vtestch_1})
            \circ \sem{
              \term_1
            }
            \circ (\ident_\Gamma \otimes \ket1_{\vinch_1})
            \bigr)
            \\ & =
            \ketbra11_q \otimes \ysem{S_1}_f
            = \ysem{S}_{(x\mapsto 1) \sqcup f}.
            \tag*{\qedhere}
          \end{align*}
  \end{description}
\end{proof}

\subsection{Discussion}

\paragraph{Degree of freedom}

In the proof above,
we choose a unitary $U$ satisfying
$U\ket1 =
  \sum_{f \in \Delta S}\allowbreak \alpha_{S,f} \ket{b_f}$.
Conversely, given any unitary $U$ whose action on $\ket1$ lies in
$\langle\ \ket{b_f}\ \rangle_{f \in \Delta S}$,
we can choose coefficients
$\{ \alpha_{S,f} \}_{f \in \Delta S}$
so that the two semantics coincide.
Thus, the freedom to choose the coefficients, discussed by \citet{Ying16-the-book}, can be identified with the freedom to choose the initial environment state
$\sum_{f \in \Delta S} \alpha_{S,f} \ket{b_f}$ discussed by \citet{Abbott2020}.
Our observations suggest that the arbitrariness identified by \citeauthor{Abbott2020} in 2020 was, in a sense, a rediscovery---from a different perspective---of the arbitrariness identified by \citeauthor{YingYF14-alternation} in 2014.

\paragraph{Canonical choice of coefficients}

When the \emph{canonical coefficients} \cite{YingYF12-first-Ying} are chosen in the semantics,
unlike the semantics of \Vac{} defined in \cref{app:sec:vacuum}, the semantics does not equate the following:
\begin{align}
  \sem{\qifx{q}{S_1; S_1'}{S_2; S_2'}}
  \ &\neq\
  \sem{\qifx{q}{S_1}{S_2}; \qifx{q}{S_1'}{S_2'}}
  .
  \label{eq:app:algebraic-low-of-qif}
\end{align}
A counterexample to the equality is obtained when
$S_1 = S_1' = \meas(p)$ and $S_2 = S_2' = \noop$,
where
$\meas(p) \defeq \cifx{(x \leftarrow \meas(p))}{\noop}{\noop}$.
The semantics of these terms is calculated as follows:
\begin{align*}
  &
  \textstyle
  \sem{\qifx{q}{(\meas(p);\meas(p))}{(\noop;\noop)}}
  \\
  & =
  \sem{\qifx{q}{\meas(p)}{\noop}}
  \\
  & =
  \textstyle
  \iota (\tfrac{1}{\sqrt2} \ketbra00_q + \ketbra{10}{10}_{qp})
  \ +\
  \iota (\tfrac{1}{\sqrt2} \ketbra00_q + \ketbra{11}{11}_{qp}),
\end{align*}
\begin{align*}
  &
  \textstyle
  \sem{
    (\qifx{q}{\meas(p)}{\noop});
    (\qifx{q}{\meas(p)}{\noop})
  }
  \\
  & =
  \sem{\qifx{q}{\meas(p)}{\noop}}
  \circ
  \sem{\qifx{q}{\meas(p)}{\noop}}
  \\
  & =
  \textstyle
  \iota (\tfrac{1}{\sqrt2} \ketbra00_q)
  \ +\
  \iota (\tfrac{1}{2} \ketbra00_q + \ketbra{10}{10}_{qp})
  \ +\
  \iota (\tfrac{1}{2} \ketbra00_q + \ketbra{11}{11}_{qp}).
\end{align*}

One may regard this phenomenon as a negative result, but we argue that it is a perfectly logical and desirable property.
In the semantics of $\meas(p); \meas(p)$,
the program produces four Kraus operators $\{ \ketbra00, 0, 0, \ketbra11 \}$
corresponding to four different classical outputs of the measurement:
00, 01, 10, and 11.
For a term $\qifx{q}{(\meas(p);\meas(p))}{S}$,
the semantics has the following Kraus decomposition:
\begin{align}
    \begin{aligned}
      \{ \quad
      &
      \beta_1 \ketbra11_q \otimes \ketbra00_p + \alpha_{00} \ketbra00_q \otimes K_1, \,\dots,\,
      \beta_n \ketbra11_q \otimes \ketbra00_p + \alpha_{00} \ketbra00_q \otimes K_n
      ,
      & \\ &
      \beta_1 \ketbra11_q \otimes \mghost{\ketbra00_p}{0} + \alpha_{01} \ketbra00_q \otimes K_1, \,\dots,\,
      \beta_n \ketbra11_q \otimes \mghost{\ketbra00_p}{0} + \alpha_{01} \ketbra00_q \otimes K_n
      ,
      & \\ &
      \beta_1 \ketbra11_q \otimes \mghost{\ketbra00_p}{0} + \alpha_{10} \ketbra00_q \otimes K_1, \,\dots,\,
      \beta_n \ketbra11_q \otimes \mghost{\ketbra00_p}{0} + \alpha_{10} \ketbra00_q \otimes K_n
      ,
      & \\ &
      \beta_1 \ketbra11_q \otimes \ketbra11_p + \alpha_{11} \ketbra00_q \otimes K_1, \,\dots,\,
      \beta_n \ketbra11_q \otimes \ketbra11_p + \alpha_{11} \ketbra00_q \otimes K_n
      &
      \ \}
    \end{aligned}
    \label{eq:app:Kraus-measmeas}
\end{align}
where $\{ K_1, \dots, K_n \}$ is the Kraus decomposition of $S$,
\ie, $\{ \ysem{S}_f \}_{f \in \Delta S}$.
On the other hand, for the term $\qifx{q}{\meas(p)}{S}$,
which measures the qubit $p$ only once,
the semantics is described as:
\begin{align*}
    \begin{aligned}
      \{ \quad
      &
      \beta_1 \ketbra11_q \otimes \ketbra00_p + \alpha_{0}' \ketbra00_q \otimes K_1, \,\dots,\,
      \beta_n \ketbra11_q \otimes \ketbra00_p + \alpha_{0}' \ketbra00_q \otimes K_n
      ,
      & \\ &
      \beta_1 \ketbra11_q \otimes \ketbra11_p + \alpha_{1}' \ketbra00_q \otimes K_1, \,\dots,\,
      \beta_n \ketbra11_q \otimes \ketbra11_p + \alpha_{1}' \ketbra00_q \otimes K_n
      &
      \ \}.
    \end{aligned}
\end{align*}
Since measuring a qubit is an idempotent operation,
$\meas(p);\meas(p)$ and $\meas(p)$ should behave identically.
In fact, in $\meas(p);\meas(p)$ the classical outputs 01 and 10 never happen,
so the corresponding Kraus operators are 0.
That means, in the semantics \eqref{eq:app:Kraus-measmeas},
it is reasonable to ask the coefficients $\alpha_{01}$ and $\alpha_{10}$ to be 0.
This is the case in the canonical semantics.
Coming back to the definition of the canonical coefficients,
we observe that the coefficients respect the norm of the Kraus operators:
\[
  \alpha_{S, f} \defeq
  \sqrt
  \frac{
    \trace( \ysem{S}_{f}^{\dagger}\ysem{S}_{f} )
  }{
    \sum_{g\in\Delta S} \trace( \ysem{S}_{g}^{\dagger}\ysem{S}_{g} )
  }.
\]
Therefore, the coefficients $\alpha_{01}$ and $\alpha_{10}$
corresponding to the two zero operators in the Kraus decomposition
are actually $0$.
Let us now return to the original equation \eqref{eq:app:algebraic-low-of-qif} with $S_1 = S_1' = \meas(p)$ and $S_2 = S_2' = \noop$.
We have observed how the semantics respects the interaction between two measurements $S_1; S_1'$ on the left-hand side.
On the other hand, $S_1$ and $S_1'$ are placed separately on the right-hand side, so they interact less.
Thus, it makes sense for the two terms to have different semantics.

In our view, the core feature of their canonical semantics is that its coefficients respect the norms of the Kraus operators.
In particular, when a Kraus operator is zero, it can be ignored.
This is a feature we cannot see in other languages.
In the language defined by \citet{Badescu2015},
it simply makes a uniformly distributed superposition,
\ie, $\alpha_{00} = \alpha_{01} = \alpha_{10} = \alpha_{11}$.
The same holds for the language defined by \citet{Barsse2026}.
For example, when we have $\meas(p); \gX(p); \meas(p)$ in a branch of $\qif$,
the 0-th Kraus operator becomes $0$,
resulting in no superposition being created by the $\qif$ statement.

\end{document}